\documentclass[preprint,5p,times]{elsarticle}
\usepackage[utf8]{inputenc}
\usepackage[T1]{fontenc}
\usepackage{bm}
\usepackage{xparse}
\usepackage{siunitx}
\usepackage{amsmath,amsthm}
\usepackage[breaklinks]{hyperref}
\usepackage{listings}
\usepackage{graphicx}
\usepackage{url}
\usepackage{xcolor,colortbl}
\usepackage{color}
\usepackage{multirow}

\newtheorem{theorem}{Theorem}[section]
\newtheorem{lemma}[theorem]{Lemma}
\theoremstyle{definition}
\newtheorem{definition}{Definition}[section]

\definecolor{temporal}{RGB}{112, 186, 164}
\definecolor{red}{RGB}{219, 191, 186}
\definecolor{orange}{RGB}{191, 167, 214}

\journal{Future Generation Computer Systems}

\begin{document}
\begin{frontmatter}

\title{A Scalable and Energy Efficient GPU Thread Map for $m$-Simplex Domains}

\author[a]{Crist\'obal A. Navarro\corref{author}}
\author[a]{Felipe A. Quezada}
\author[b]{Benjamin Bustos}
\author[b]{Nancy Hitschfeld}
\author[b]{Rolando Kindelan}
\cortext[author] {Corresponding author.\\\textit{E-mail address:} cristobal.navarro@uach.cl}
\address[a]{Instituto de Informática, Universidad Austral de Chile.}
\address[b]{Computer Science Department (DCC), University of Chile.}

\begin{abstract}
This work proposes a new GPU thread map for $m$-simplex domains that improves its speedup along with the $m$-dimension and is energy efficient compared to other state of the art approaches. The main contributions of this work are i) the formulation of an improved new block-space map $\mathcal{H}: \mathbb{Z}^m \mapsto \mathbb{Z}^m$ for regular orthogonal simplex domains, which is analyzed in terms of resource usage, and ii) the experimental evaluation in terms of speedup and energy efficiency with respect to a bounding box approach. Results from the analysis show that $\mathcal{H}$ has a potential speedup of up to $2\times$ and $6\times$ for $2$ and $3$-simplices, respectively. Experimental evaluation shows that $\mathcal{H}$ is competitive for $2$-simplices, reaching $1.2\times \sim 2.0\times$ of speedup for different tests, which is on par with the fastest state of the art approaches. For $3$-simplices $\mathcal{H}$ reaches up to $1.3\times \sim 6.0\times$ of speedup making it the fastest. The extension of $\mathcal{H}$ to higher dimensional $m$-simplices is feasible and has a potential speedup that scales as $m!$ given a proper selection of parameters $r, \beta$ which are the scaling and replication factors of the geometry, respectively. In terms of energy consumption, although $\mathcal{H}$ is among the highest in power consumption, it compensates by its short duration, making it one of the most energy efficient approaches. Lastly, further improvements with Tensor and Ray Tracing Cores are analyzed, giving insights to leverage each one of them. The results obtained in this work show that $\mathcal{H}$ is a scalable and energy efficient map that can contribute to the efficiency of GPU applications when they need to process $m$-simplex domains, such as Cellular Automata or PDE simulations.
\end{abstract}

\end{frontmatter}

\section{Introduction}
GPU computing has become a well established research field in the last years \cite{4490127,Nickolls,navhitmat2014} thanks to the increasing performance of GPU hardware and the existence of a general purpose programming model, being CUDA \cite{nvidia_cuda_guide} and OpenCL \cite{opencl08} the most known implementations of such model. Notorious advancements have been made to better understand how data-parallel algorithms can be more efficient on the GPU architecture. Some of these advancements include the improvement of thread access patterns for coalesced memory access \cite{10.1145/2517327.2442523}, reduction of thread branching \cite{4408272, 8947754}, better use of the $L_1$ programmable cache \cite{7445236}, fast stencil computations \cite{10.1145/2304576.2304619}, efficient reduction algorithms \cite{harris2007optimizing,navarro2020gpu}, multi-GPU acceleration \cite{hermann2010multi,yadan2013multi, stuart2011multi}, and most recently the mapping of different computational patterns to tensor cores \cite{sorna2018optimizing, navarro2020gpu, quezada2022squeeze, navarro2020efficient} as well as ray tracing cores \cite{zhu2022rtnn, zellmann2020accelerating,morrical2020accelerating}. One performance aspect of GPU computing that is also being researched in the last years is the study and design of efficient GPU thread maps for data-parallel domains with complex geometries \cite{navarro2020efficient, DBLP:conf/hpcc/NavarroH14, CLEI-2016-navarro, 8392762}. Several problems exhibit a non-boxed shape data domain in which data-parallel GPU computation must be mapped. Examples of such problems are regular orthogonal $m$-simplex domains, which exist when working with half-triangular matrices, or when doing a cellular automata simulation on a discrete voxel-based tetrahedron. A regular orthogonal $m$-simplex is an $m$-dimensional polytope, where its facets are equal, define a convex hull and one of its vertices has all of its incident  facets orthogonal one to each other. The discrete voxelized/cell-like version, denoted $\Delta_n^m$, is of interest in data-parallel computations and is composed of unitary volumetric elements with location $x = \{x_1, x_2, ..., x_m\}$ that satisfy 
\begin{equation}
    \Delta_n^m \equiv \{x \in \mathbb{Z}_+^m | 0 \le x_i \le n \land x_1 + x_2 + ...
x_m \le n\}
\end{equation}
where the absolute Manhattan distance from any element $x$ to the orthogonal corner of the $m$-simplex cannot be greater than $n$. Mapping GPU threads efficiently to a regular orthogonal $m$-simplex is not straightforward in the GPU programming model, because it was designed for box-shaped domains. 

This work proposes a new map $\mathcal{H}: \mathbb{Z}^m \mapsto \mathbb{Z}^m$ for mapping GPU threads onto regular orthogonal $m$-simplex\footnote{For simplicity, sentences using just the term simplex/simplices will also refer to the regular orthogonal ones.} domains efficiently. The map assigns thread-blocks in parallel space to unique locations in the $m$-simplex shaped data, using $\mathcal{O}(1)$ arithmetic operations. A dedicated analysis is devoted to the special cases of $2$-simplex and $3$-simplex domains, where new efficient maps are proposed and described, based on a binary self-similar organization of sub-orthotopes. The approach can offer a parallel space improvement of $\sim2\times$ and $\sim6\times$ for $m=2$ and $m=3$, respectively, that results in a potential performance improvement given that the maps costs $\mathcal{O}(1)$ arithmetic operations and no $m$-th roots are required. Experimental results support the theoretical bounds, showing that $\mathcal{H}$ is competitive in $2$-simplices, the fastest one in $3$-simplices, and energy efficient. For the higher dimensional case, an analysis shows that building an efficient set of self-similar orthotopes is possible with a potential speedup of $m!$, given optimal values are found for the scaling and replication factors, $r$ and $\beta$ respectively.  

This work improves from a previous conference publication \cite{9188081} by including an improved formulation and a full experimental evaluation using more comparison approaches, additional tests and several GPUs. Also, it extends the previous work by including a complete evaluation of energy efficiency for all approaches, as well as new insights for leveraging Tensor and Ray Tracing cores in modern GPU architectures.  

\section{The Problem of Mapping GPU Threads onto Simplices}
\label{sec_problem}
The GPU programming model offers four constructs\footnote{This
work uses CUDA's naming scheme. OpenCL names are (1) work-element, (2) work-group and (3) work-space,
respectively.} that allow the execution of highly parallel algorithms; (1)
thread, (2) block, (3) grid and (4) kernel.  Threads are the smallest work elements and they are in charge of executing the instructions of a GPU kernel, which is the main parallel program written by the programmer. A block is an intermediate structure that contains a set of threads organized as an Euclidean box. Blocks provide fast shared memory access as well as local synchronization for all of its threads.  The grid is the largest construct of all three and it keeps all blocks together spatially organized for the execution of a GPU kernel. These three constructs are illustrated in Figure \ref{fig_constructs} and play an important role when
mapping processing resources to a data-parallel domain.
In 2022 the CUDA programming model introduced the \textit{thread-block cluster}, an optional construct that fills the gap in between a block and the grid. It can be used for leveraging memory locality among blocks. 
\begin{figure}[ht!]
\centering
\includegraphics[scale=0.25]{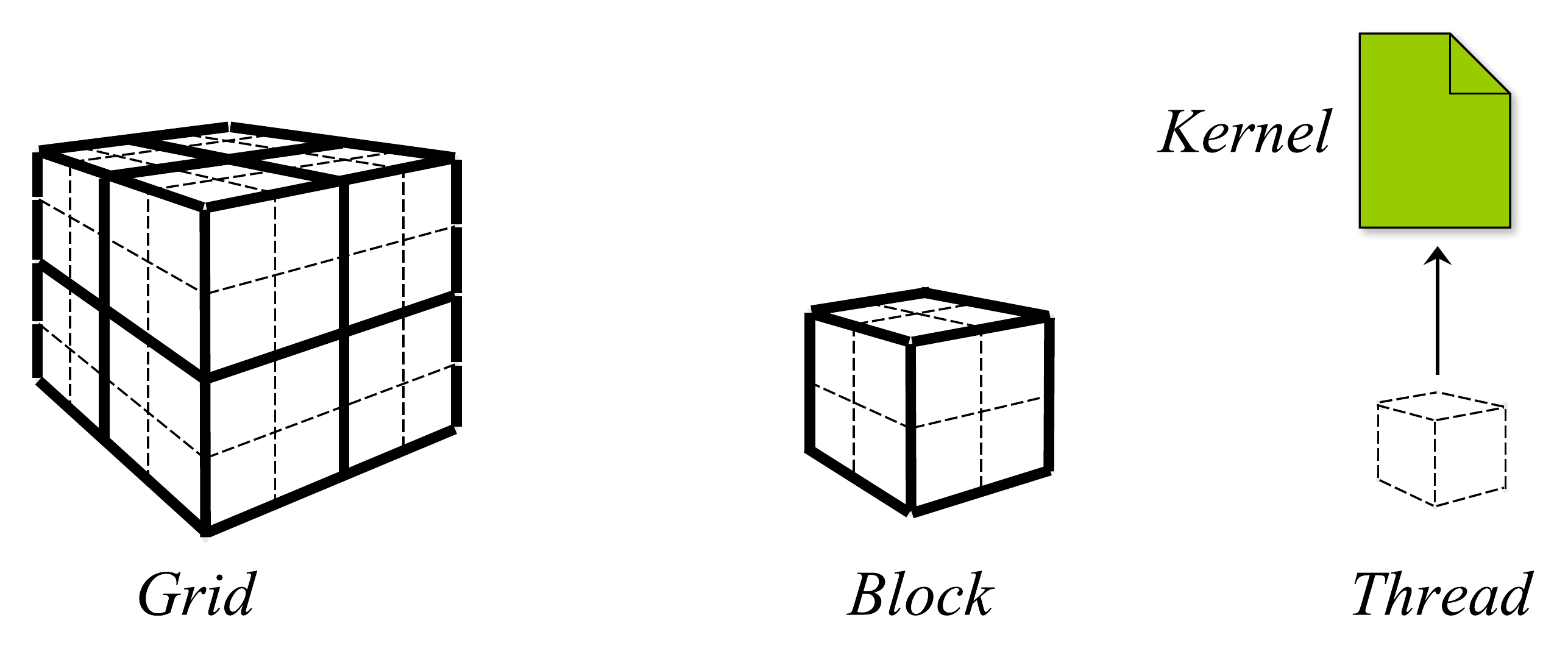}
\caption{GPU constructs using the CUDA naming scheme.}
\label{fig_constructs}
\end{figure}

For every GPU computation there is a stage where threads, which live in grid space, get mapped to different locations in data space. This mapping is defined as
\begin{equation}
    f(x): \mathbb{Z}^k \rightarrow \mathbb{Z}^m
\end{equation}
where $k$-dimensional locations $x=(x_1, x_2, ..., x_k)$, in grid space, map to $m$-dimensional locations $y = (y_1, y_2, \cdots, y_m)$ in data space, typically as an injective or bijective function. GPU parallel spaces can be formally described as orthotopes $\Pi_n^m$ in $m=1,2,3$ dimensions that contain discrete elements
organized in blocks. Each element corresponds to a thread that must be assigned
to the data domain.

A known approach for mapping threads in GPUs is to build a \textit{bounding-box} (BB) of volume\footnote{CUDA offers up to $k=3$. Higher dimensional orthotopes can be still be represented by linearizing down to 3-dimensions.} $l_1 \times l_2 \cdots l_k$ where $l_i$ is the length of the data space in the $i$-th dimension. Without loss of generality, we can assume a regular geometry to express the bounding box volume as $n^k$ where $n$ is the needed length to cover the entire data domain in all dimensions. Using the map 
\begin{equation}
    f(x) = x
\end{equation}
one can obtain the corresponding location of every thread in data space. Such map is highly convenient, simple and efficient for the class of problems where the data space is an orthotope as well, such as vectors, tables, images, matrices or other arbitrary box-shaped data domains. However, this approach is no longer convenient neither efficient once data domains start taking other shapes. 

There is a class of highly parallelizable problems where data organizes in the form of a regular simplex. Computations such as the Euclidean distance matrix (EDM) \cite{5695222, Li:2010:CME:1955604.1956601, Man:2011:GIC:2117688.2118809}, collision detection
\cite{AvrilGA12}, adjacency matrices \cite{kepner2011graph}, cellular automata simulation on triangular domains \cite{ConwaysLife}, matrix
inversion \cite{Ries:2009:TMI:1654059.1654069} and even the \textit{n-body} problem \cite{DBLP:journals/corr/abs-1108-5815, Bedorf:2012:SOG:2133856.2134140, Ivanov:2007:NPT:1231091.1231100}, among others, follow the shape of a 
regular orthogonal $2$-simplex. Here, such simplex is denoted as $\Delta_n^2$, with $n$ its length along one dimension, and its working space is
$V(\Delta_n^2)=n(n+1)/2 \in \mathcal{O}(n^2)$. 
The problem is that for this class of problems the default \textit{bounding-box} (BB) GPU mapping 
approach turns out to be inefficient because the orthotope of parallel space, denoted $\Pi_n^2$, produces $n^2$ threads, where $n(n-1)/2 \in \mathcal{O}(n^2)$ are unused, as shown in Figure \ref{fig_bb_strategy}.
\begin{figure}[ht!]
\centering
\includegraphics[scale=0.09]{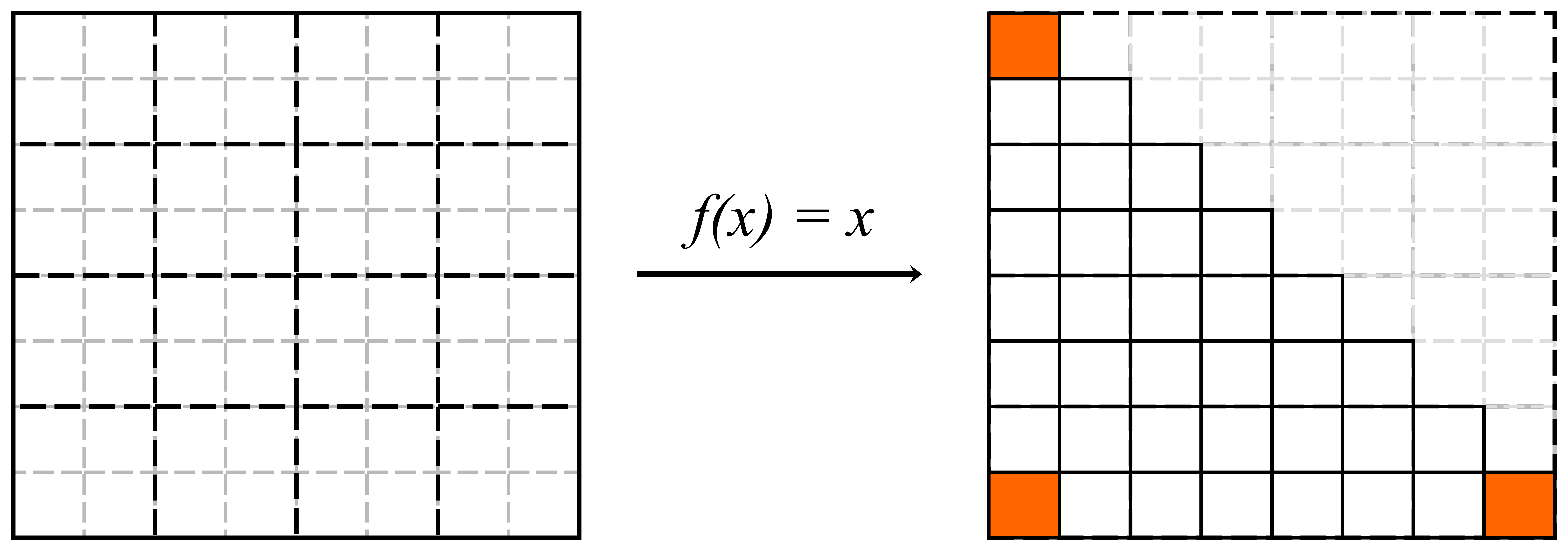}
    \caption{The bounding-box approach almost doubles the
    number of threads.}
\label{fig_bb_strategy}
\end{figure}

The problem is not restricted just to two-dimensional cases.  Simulation applications that act on voxelized or cell-like regular tetrahedrons, such as PDEs or Cellular Automata, also employ a $3$-simplex. In the $3$-simplex class, the interaction space is $V(\Delta_n^3)=n(n+1)(n+2)/6 \in \mathcal{O}(n^3)$ elements, organized as a discrete regular orthogonal tetrahedron.
Once again, the default \textit{bounding-box} (BB) approach turns out to be 
inefficient as it generates a parallel space $V(\Pi_n^3)$ with
$\mathcal{O}(n^3)$ unnecessary threads (see Figure \ref{fig_bb_strategy_tetrahedron}).
\begin{figure}[ht!]
\centering
\includegraphics[scale=0.10]{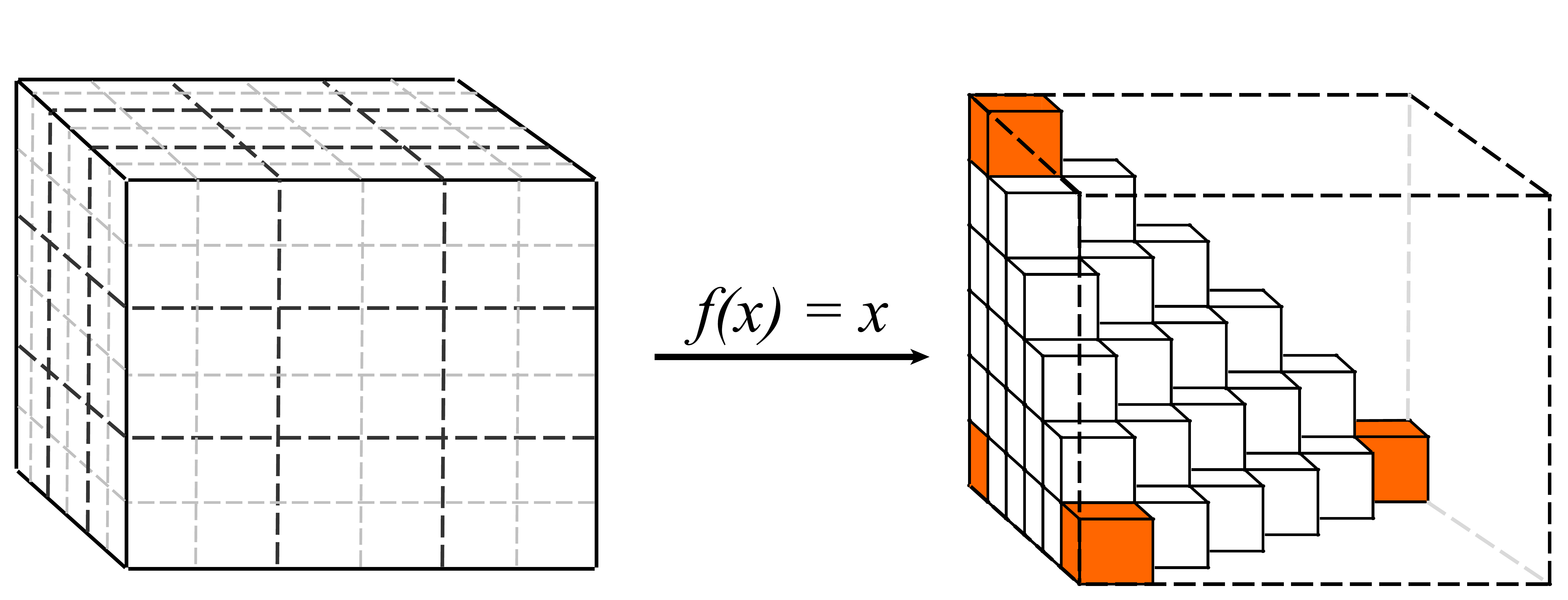}
\caption{The bounding-box generates near six times the number of threads.}
\label{fig_bb_strategy_tetrahedron}
\end{figure}

The issues recently described are two specific cases of the problem of mapping an $m$-orthotope to a regular orthogonal $m$-simplex, while matching their volumes at least asymptotically. Figure \ref{fig_dom-simplex} illustrates the geometry of discrete orthogonal simplices for $m=1,2,3$.
\begin{figure}[ht!]
\centering
\includegraphics[scale=0.11]{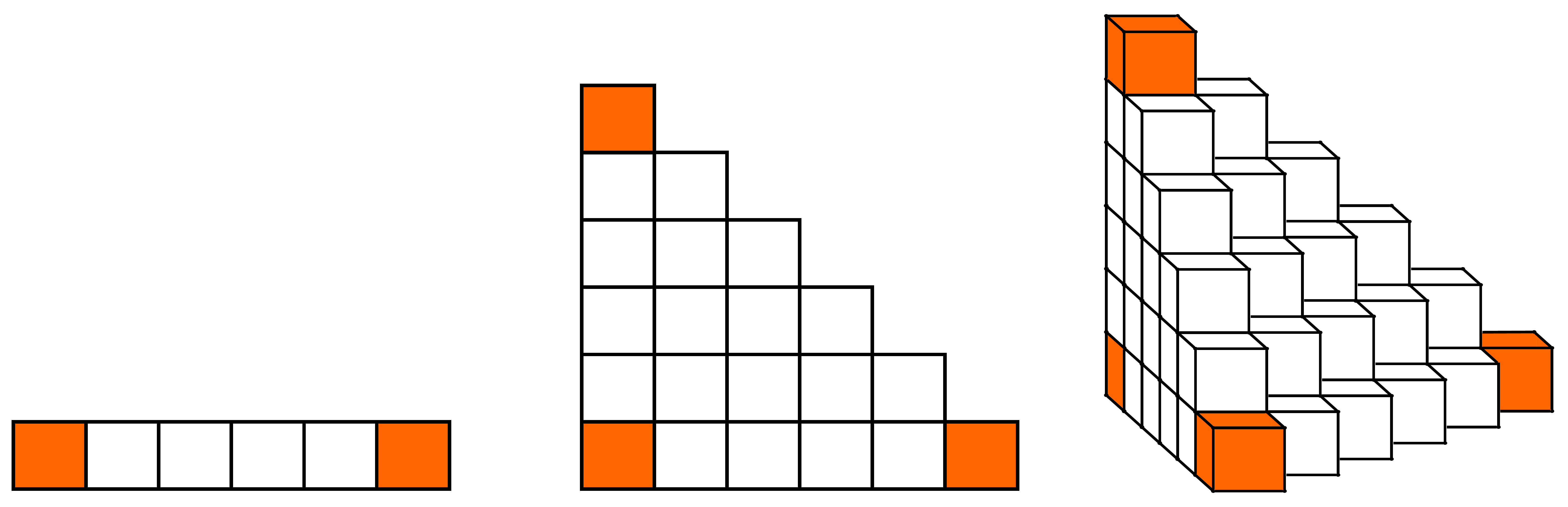}
\caption{Discrete orthogonal $m$-simplices for $m=1, 2, 3$.}
\label{fig_dom-simplex}
\end{figure}

The expression for the volume of an $m$-simplex corresponds to the
\textit{simplicial polytopic numbers}, which are defined as
\begin{equation}
    V(\Delta_n^m) = {{n+m-1}\choose{m}} = \frac{n(n+1)(n+2)...(n+m-1)}{m!}
    \label{eq_volume}
\end{equation}
where an induction \cite{1054599} can prove its formula, based on the fact that
the volume of $\Delta_n^{m+1}$ is the sum of the volumes of $n$ stacked
$m$-simplices of lengths $\{1, 2, 3, ..., n\}$, \textit{i.e.},
\begin{equation}
    V(\Delta_n^{m+1}) = \sum_{i=1}^n{V(\Delta_i^m)}
    \label{eq_volume_enum}
\end{equation}
and when combined with the properties of sums of binomial coefficients, leads
to formula (\ref{eq_volume}).

The parallel space efficiency challenge arises when trying to process an $m$-simplex domain with GPU computing, as its parallel programming model can only manage orthotopes of threads, namely $\Pi_n^m$, and not arbitrary shapes as simplices. Although the \textit{bounding-box} approach will work and is simple to implement, its main problem is that in the limit $n\mapsto \infty$, the fraction of extra parallel space of $V(\Pi_n^m)$ that lies outside of the $m$-simplex approaches to
\begin{equation}
    \alpha(\Pi, \Delta)_n^m = \lim_{n\to\infty} \frac{V(\Pi_n^m)}{V(\Delta_n^m)}
    - 1 =  m!-1
\end{equation}
thus being inefficient in terms of parallel space as all the extra thread-blocks
must be discarded at run-time through conditional or arithmetic filters. The
implication is a performance penalty at the mapping stage of any kernel execution.
Moreover, these extra thread-blocks can produce a negative impact on the
performance of multiple concurrent kernels as they could make use of hardware
resources that could be potentially available to other kernels.

One solution, used in the past years, is to enumerate all threads or
thread-blocks of a given orthotope in sequence, and then use this enumeration as
an indexing scheme. Such approach allows to formulate a map
\cite{DBLP:conf/hpcc/NavarroH14, CLEI-2016-navarro, 8392762} of the form
$\lambda:\mathbb{Z}^1 \mapsto \mathbb{Z}^m$ with $V(\Pi_n^m) =
\mathcal{O}(V(\Delta_n^m))$.  Although $\lambda$ can be computable by
elementary and arithmetic functions, it also requires the computation of $m$-th
order roots as part of the solution of an $m$-th order equation, which limits
the range of problem size $n$ due to potential numerical error in the mapping of
thread coordinates. Moreover, the method is restricted to $m \le 4$ as no general analytical
solutions are available for polynomials of $m > 4$. Finding a different kind
of map, free of these problems, would improve the state of the art regarding 
GPU thread maps for simplex domains.

The limitations described can be overcomed, in great part, by taking advantage
of the dimensionality available in the parallel space, which was an aspect not
exploited by the enumeration principle proposed in the past \cite{8392762}.  
Although parallel GPU spaces cannot
have a geometry different from an orthotope, they still are topologically
equivalent to an $m$-simplex and a self-similar set can match the desired
domain.  Therefore, finding an homeomorphism of the form
$\mathcal{H}: \mathbb{Z}^m \mapsto \mathbb{Z}^m$ would produce a map, efficient in
parallel space and with zero dimensional distance which would free it
from the computation of expensive $m$-th roots.

\section{Related Work}
\label{sec_related-work}
One of the earliest works in mapping the GPU efficiently onto a $2$-simplex is by Jung \textit{et. al.} \cite{Jung2008} in 2008, whom proposed a \textit{rectangular box strategy} (RB) for accessing and storing a triangular matrix (upper or lower) more efficiently on GPU. Data structures become practically half the size with respect to classical methods based on the full matrix. The strategy was originally intended to modify the data space (\textit{i.e.,} the matrix),
however one can apply the same concept to the thread space. One disadvantage of RB is that it only works on $2$-simplices, \textit{i.e.}, applying the idea at higher dimensions produces unnecessary threads.

Ries \textit{et. al.} contributed with a parallel GPU method for the triangular
matrix inversion \cite{Ries:2009:TMI:1654059.1654069}.  The authors identified
that the parallel space indeed can be improved by using a \textit{recursive
partition} (REC) of the grid, based on a \textit{divide and conquer} strategy.
The approach takes $O(\log_2(n))$ time by doing a balanced partition of the
structure, from the orthogonal point to the diagonal.

Q. Avril \textit{et. al.} proposed a GPU mapping function for collision
detection based on the properties a \textit{upper-triangular map}
\cite{AvrilGA12}. The map is a thread-space function $u(x) \rightarrow (a, b)$,
where $x$ is the linear index of a thread $t_x$ and the pair $(a,b)$ is a unique
two-dimensional coordinate in the upper triangular matrix. Since the map works
in thread space, the map is accurate only in the range $n \in [0, 3000]$ for a simplex of side length $n$. 

Navarro \textit{et al.} proposed a block-space map, denoted $\lambda$, for
$2$-simplices and $3$-simplices \cite{DBLP:conf/hpcc/NavarroH14,
CLEI-2016-navarro, 8392762}, based on the solution of an $m$ order equation that is
formulated from the linear enumeration of the discrete elements. The authors
report performance improvement for $2$ and $3$-simplices. By being a block-space
map, the authors report correct mapping coordinates up to problems of $n = 62900$ 
for the $2$-simplex case and up to $n=1546$ for the $3$-simplex case. Beyond these problem sizes, the map requires 64-bit floating point square roots (FP64) which penalize the speedup. 

This work explores a new concept of map, denoted $\mathcal{H}$, that costs $\mathcal{O}(1)$ time and is as fast as $\lambda$ and RB methods, while being free of numerical precision dimensional limitations. The proposed map $\mathcal{H}$ preserves locality at block-level and only uses integer arithmetic and bit-shift operations for its computation, making it an exact map by design.

\section{Formulation of $\mathcal{H}$}
\label{sec_formulation}

The formulation of $\mathcal{H}$ focuses on the special cases $m=2,3$, where the mapping can also be supported by illustrations, although it is not limited to just those dimensions, as explained later in Section \ref{sec:considerationshigher-dimensions}. Also, the analysis first assumes $n$ as a power of two \textit{i.e.}, $n = 2^k$ with $k \in \mathbb{Z}_+$, and later in sub-section \ref{subsec:general_n} expands to the general case of $n \in \mathbb{Z}_+$. The case when $m=1$ is not considered in the analysis as both the orthotope and simplex match in geometry thus no special map is required. 

For the rest of the manuscript, the concept of an \textit{efficient map} will be
based on the following definition:
\begin{definition}
\label{def_1}
An \textit{efficient map} is a function $\mathcal{H}:\mathbb{Z}^k_+
    \mapsto \mathbb{Z}^m_+$ that can map any thread coordinate
    $\omega_{x_1,x_2,...,x_k}$ from a parallel space $\Pi^k$ of volume $V(\Pi^k)
    = V(\Delta_n^m) \pm o(n^m)$, onto a unique location in $\Delta_n^m$, using
    $\mathcal{O}(1)$ arithmetic or bitwise operations, excluding explicit transcendental operations and roots.
\end{definition}
The following sub-sections formulate maps for $\Delta^2_n$, $\Delta^3_n$
and analyze their efficiency according to Definition \ref{def_1}.

\subsection{Mapping to $2$-Simplices}
Given a $2$-simplex, denoted as $\Delta_n^2$, with $n$ being its side length, its space $V(\Delta_n^2)$ is given by the triangular numbers
\begin{equation}
    V(\Delta_n^2) = \frac{n(n+1)}{2}.
\end{equation}
Each location of the $2$-simplex is a $(x,y)$ pair where $x + y \le n$ satisfies when the origin is at the orthogonal vertex. The goal is to find an orthotope $\Pi^2$ with an asymptotic volume of
$\mathcal{O}(V(\Delta_n^2))$ and an efficient map $\mathcal{H}$, that
with $O(1)$ arithmetic operations, assigns threads of $\Pi^2$
to unique locations in $\Delta_n^2$.  

\begin{lemma}
\label{lemma_1}
    There exists a set $S_n^2$ of self-similar regular orthotopes where $V(S_n^2) = V(\Delta_{n-1}^2)$.
\end{lemma}

\begin{proof}
Let $S_n^2$ be a set of self-similar regular orthotopes where its total volume is defined by the recurrence
\begin{equation} 
\label{eq_m2recurrence}
    V(S_n^2) = \Big(\frac{n}{2}\Big)^2 + 2V(S_{n/2}^2)
\end{equation}
where the largest regular orthotope of $S_n^2$ has dimensions $\frac{n}{2} \times \frac{n}{2}$ and the initial condition is $V(S_1^2) = 0$. The expansion of Eq. (\ref{eq_m2recurrence}) produces the geometric series
\begin{alignat}{2}
    \label{eq_m2exp}
    V(S_n^2) &= 2^0\Big(\frac{n}{2^1}\Big)^2 + 2^1\Big(\frac{n}{2^2}\Big)^2
    +\dots +2^{\log_2(n) -1}\Big(\frac{n}{2^{\log_2 n}}\Big)^2\\
    &= \frac{n^2}{2} \sum_{i=1}^{\log_2 n} \Big(\frac{1}{2^{i}}\Big).
\end{alignat}
where its reduction via $\sum_{i=0}^{k}{a^i} = \frac{a^{k+1} - 1}{a-1}$, results in
\begin{alignat}{2}
    V(S_n^2) &= \frac{n^2}{2} \Bigg(-1 +  \sum_{i=0}^{\log_2 n} {\Big(\frac{1}{2}\Big)^i}\Bigg)\\
    &= \frac{n^2}{2} \Bigg( -1 + \frac{ (1/2)^{\log_2{n} + 1} - 1 }{1/2 - 1}
    \Bigg) \\
    \label{eq_m2}
    &= \frac{n(n-1)}{2} = \Delta_{n-1}^2.
\end{alignat}
\end{proof}
Given that its possible to represent a $2$-simplex as a set of regular orthotopes, it is now necessary to group them as a grid to fit in the GPU programming model.

\begin{lemma}
\label{lemma_2}
    A set $S_n^2$ can be organized into a single irregular super-orthotope $\Pi_{\frac{n}{2}, n-1}^2$ of dimensions $\frac{n}{2} \times n-1$.
\end{lemma}

\begin{proof}
    In the expansion of formula ($\ref{eq_m2exp}$), each term can be thought as the contribution of a sub-set of $2^{i-1}$ regular orthotopes of equal size, each one of dimensions $\frac{n}{2^i} \times \frac{n}{2^i}$, where $i$ is the $i$-th level of recursion. When considering whole sub-sets, we have that  its space is $2^{i-1} \frac{n^2}{2^{2i}} = \frac{n}{2} \times \frac{n}{2^i}$, therefore it is possible to stack these sub-sets to form a super-orthotope $\Pi_{n/2, n-1}^2$ where its height is the result of the geometric series $\sum_{i=1}^{\log_2 n} \frac{n}{2^i} = n-1$.
\end{proof}
Figure \ref{fig_m2-ortho-simplex-map}, left, illustrates the super-orthotope and how all sub-sets stack into different layers.  
\begin{figure}[ht!]
\centering
\includegraphics[scale=0.10]{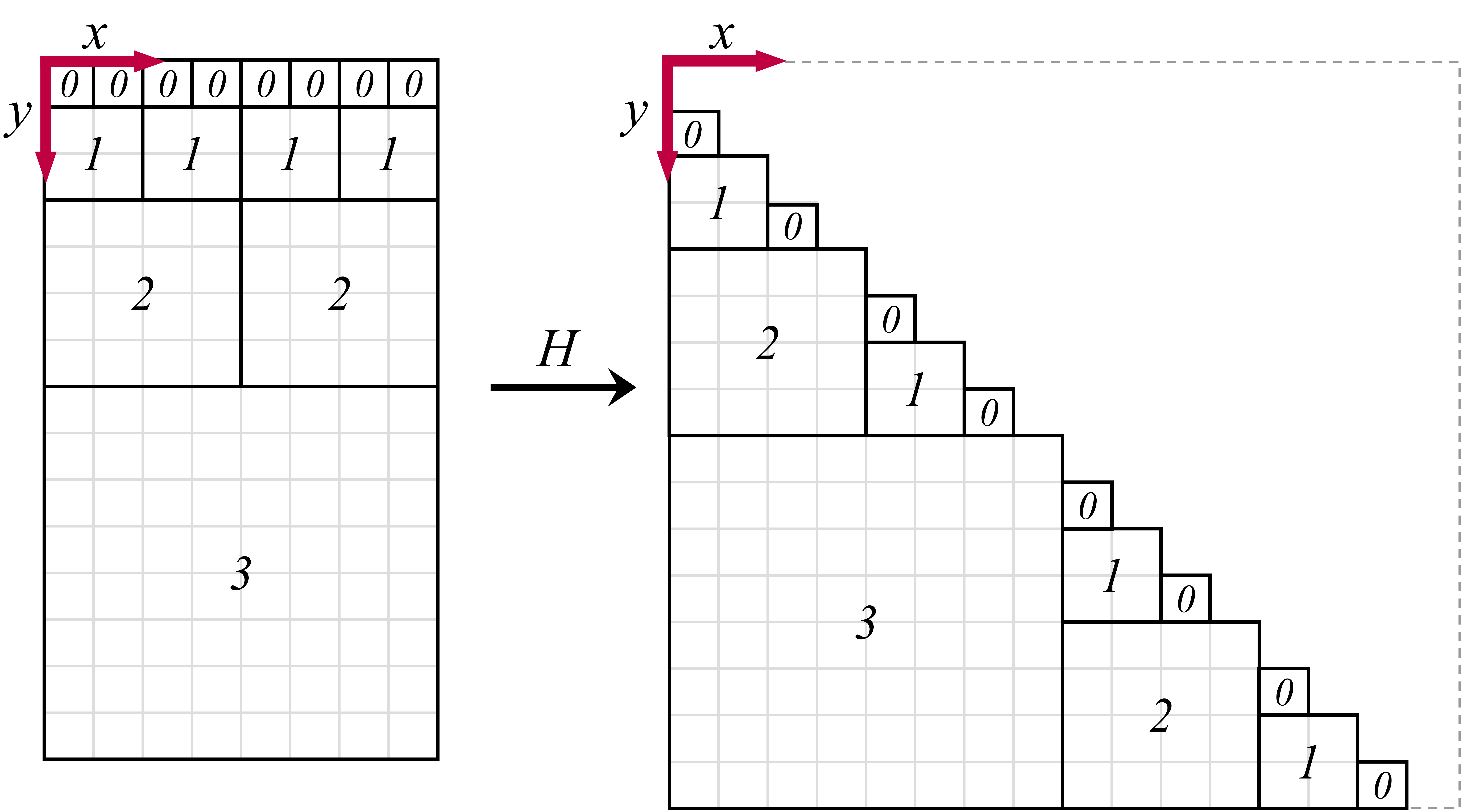}
    \caption{$\Pi_{\frac{n}{2},n-1}^2$ (left) and $\Delta_{n-1}^2$ (right) are two different organizations of $S_{n}^2$.}
\label{fig_m2-ortho-simplex-map}
\end{figure}

\begin{theorem}
\label{theorem_1}
    Given a discrete orthogonal $2$-simplex $\Delta_{n-1}^2$, there exists a 
    map $\mathcal{H}: \mathbb{Z}_+^2 \mapsto \mathbb{Z}_+^2$ from
    $\Pi_{\frac{n}{2},n-1}^2$ to $\Delta_{n-1}^2$ that is efficient by Definition \ref{def_1}.
\end{theorem}
\begin{proof}
    The proof proceeds by construction.
    Let $S_n^2$ be a set of self-similar orthotopes which by Lemma \ref{lemma_1} matches $\Delta_{n-1}^2$ in space. By Lemma \ref{lemma_2}, all regular
    orthotopes of $S_n^2$ can be organized to form an irregular super-orthotope $\Pi_{n/2,n-1}^2$. For convenience, the origin of $\mathbb{Z}^2$ plane is redefined at the top left corner of $\Pi^2$ and the vertical axis points down (Figure \ref{fig_m2-ortho-simplex-map}). 
    The point $\omega = (\omega_x,\omega_y)$ is defined as the two-dimensional 
    coordinate of a discrete unit in $\Pi^2$, which corresponds to a
    thread-block, \textit{i.e.}, not a thread but a group of threads. Lemma \ref{lemma_2} also establishes that sub-orthotopes share the
    same stack height level only if they are equal in size. Such property allows
    to compute the stack height level $b$ of any discrete element of $\Pi^2$ as
    \begin{equation}
        b = 2^{\lfloor \log_2(\omega_y +1)\rfloor}
    \end{equation}
    as well as 
    \begin{equation}
        q = \Big\lfloor \frac{\omega_x}{b} \Big\rfloor
    \end{equation}
    which is the index of the orthotope for which the pair $\omega$ belongs to,
    at level $b$, from left to right.  Finally, arithmetic operations between
    $\omega,b,q$ produce the map
\begin{equation}
    \label{eq_map2d}
    \mathcal{H}(\omega) = (\omega_x + qb, \omega_y + 2qb + 1)
\end{equation}
    which is efficient by Definition \ref{def_1} as $V(\Pi_{n}^2) = V(\Delta_{n-1}^2) + o(n^2)$ satisfies using $O(1)$ arithmetic operations excluding roots.
\end{proof}
    The map $\mathcal{H}$ from Eq. (\ref{eq_map2d}) is a considerable
    improvement over previous proposed maps based on the enumeration principle as it does not use square roots. Regarding the number of sequential steps, it is also more efficient than the recursive map by Ries \textit{et.  al} which was graphically similar but required $O(\log_2(n))$ \cite{Ries:2009:TMI:1654059.1654069} time. Additionally, since blocks have a constant size of $\rho^2 \ll n$ (with $\rho$ the number of threads in each dimension of the block), any extra number of threads only occurs at the diagonals or lower boundary, which is asymptotically $\le 2n\rho^2 \in o(n^2)$.

    The computation of $\mathcal{H}$ requires a constant number of arithmetic operations, as well as two transcendental functions in $\mathbb{Z}_+$; $\lfloor \log_2 (\omega_y+1) \rfloor$ and $2^{\lfloor \log_2(\omega_y+1) \rfloor}$.  In general, function $\lfloor \log_2(...) \rfloor$ can be computed fast on GPU by using the relation
    \begin{equation}
        \lfloor \log_2(x) \rfloor = bits - clz(x)
    \end{equation}
    where $bits$ is the number of bits of the word and $clz$ counts the number of leading zero bits of $y$, which is an efficient hardware-level
    instruction available in GPUs. The exponential $2^{\lfloor \log_2(\omega_y+1) \rfloor}$ is computed as
    \begin{equation}
        2^{\lfloor \log_2(\omega_y+1) \rfloor} = 2 << (bits - clz(\omega_y+1)).
    \end{equation}
    Considering that the two transcendental functions can be computed using bit-level operations, and that the parameters are re-used by registers, it is expected that the parallel space improvement from $O(n^2)$ to $O(n)$ unnecessary threads can indeed result in a significant performance improvement, which for the case of 2-simplices can be up to $2\times$ \cite{DBLP:conf/hpcc/NavarroH14, 8392762}. Moreover, since no square roots are required, $\mathcal{H}$ has the potential to be faster and more precise than previous mapping techniques based on the enumeration principle \cite{AvrilGA12, DBLP:conf/hpcc/NavarroH14, 8392762}.

\subsection{Extending $\mathcal{H}$ to the general case $n \in \mathbb{Z}_+$}
\label{subsec:general_n}
    The analysis of $\mathcal{H}$ has assumed problems with sizes of
    the form $n = 2^k$. For general $n$, we present three extensions, with the 
    third being the chosen and most efficient one. 
    \begin{enumerate}
        \item \textit{Approach $n$ from above}: build a single orthotope $\Pi_{n'}^2$, where $n' =
                2^{\lceil \log_2(n) \rceil}$ and filter out the threads outside
                the domain. This approach keeps simplicity at the cost of adding
                extra threads. The main disadvantage is for values of $n = 2^k +
                c$ when $c$ is a small number, \textit{i.e.,} $c \ll 2^{k+1} - 2^k$, as it can generate up to $4\times$ the number of necessary threads.
        \item \textit{Approach $n$ from below}: apply a set of super-orthotopes
            $\Pi_{n_1}^2, \Pi_{n_2}^2, \Pi_{n_i}^2, ...$, where $n_i =
            \log_2\Big(n - \sum_{k=1}^{i-1} n_k\Big)$ for $i \ge 2$, $n_1 =
            \lfloor \log_2(n) \rfloor$, plus a set of additional maps for
            the sub-regions that remain un-mapped below each simplex, at each
            level. This approach, while not adding extra threads, adds
            complexity in the algorithm and sequential steps.
        \item \textit{Concurrent Trapezoids}: This approach combines the 
            ideas of the other two; it extends the mapped geometry to a
            trapezoid instead of a simplex, and assumes kernel concurrency, which is
            a feature that exists in actual GPU architectures. The idea is to
            launch, concurrently, a set of special orthotopes that map onto
            trapezoids in the data domain. The set of concurrent trapezoids 
            follows the principle of approaching $n$ from below, except for the
            smallest one that approaches the remaining size from above. This last trapezoid triggers when $2^{\lceil{\log_2(n)}\rceil} - n_k < T$, with $T$
            an arbitrary threshold value that serves as a mechanism to
            limit the size of the trapezoids set.
    \end{enumerate}
    The \textit{Concurrent Trapezoids} approach is the most efficient one in
    terms of map efficiency because of the following three features:
    (1) it works concurrently, keeping the cost of $O(1)$ time unchanged, (2) the set of trapezoids will be small as the triggering
    condition for the last trapezoid has a high probability to occur at an early
    stage (only on rare occasions will produce a set of $\log_2(n)$ elements which is the worst case)
    when using an adequate threshold, and (3) a small modification in the
    expression of $\mathcal{H}$ makes the map reusable even when the
    mapped domain is now a trapezoid instead of a simplex.  
    Figure \ref{fig_trapezoids} illustrates the idea behind the \textit{concurrent
    trapezoids} approach.
    \begin{figure}[ht!]
    \centering
    \includegraphics[scale=0.15]{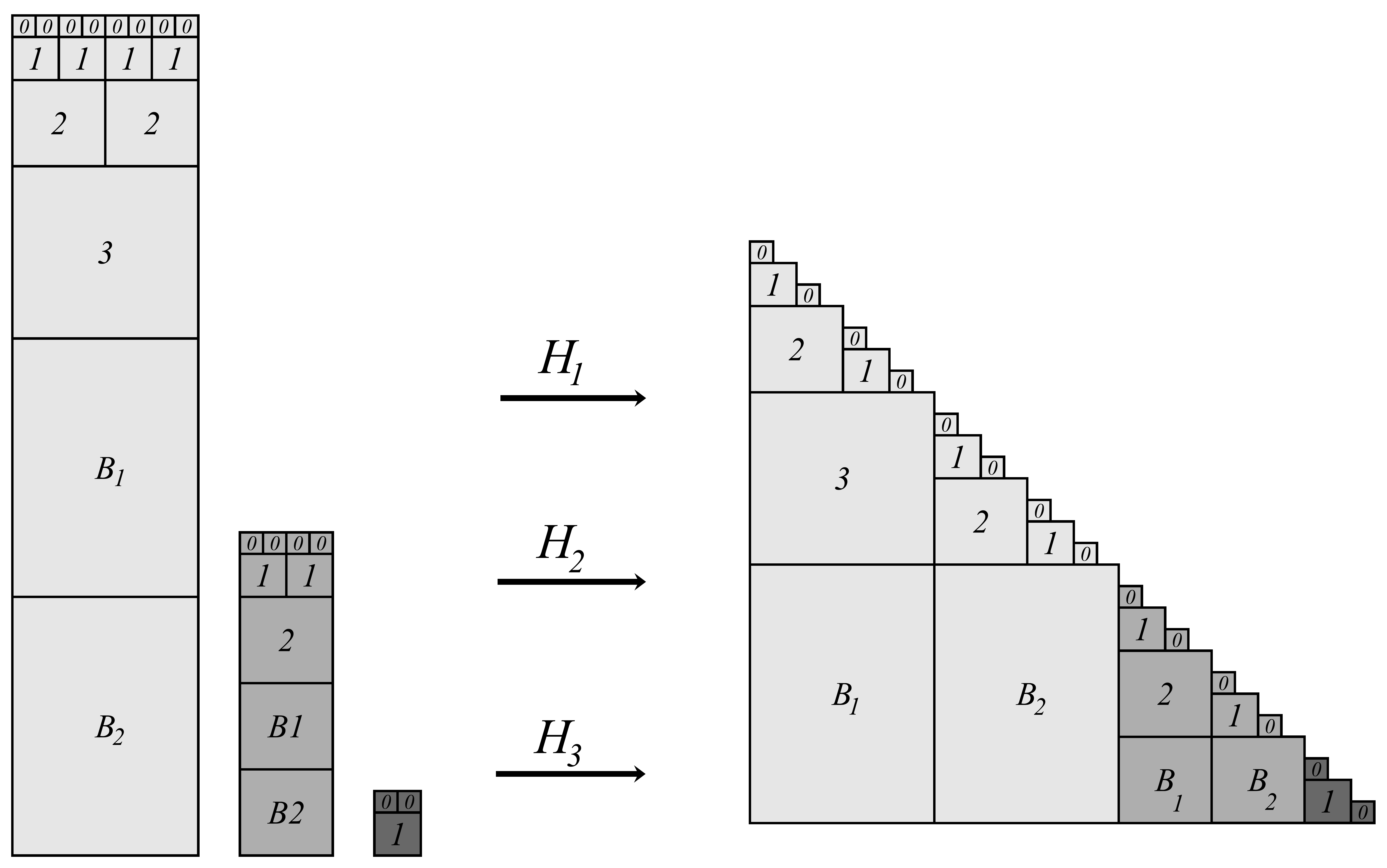}
        \caption{The concurrent trapezoids approach for a simplex of $n=27$ with
        threshold of $T = 1$.}
    \label{fig_trapezoids}
    \end{figure}

    Under the \textit{Concurrent Trapezoids} approach, $\mathcal{H}$
    receives three extra parameters that are common for all threads; i) an
    offset $\delta(x,y)$ to specify the starting coordinate of each trapezoid
    map in the simplex, ii) the orthotope height of the simplex region with the $B_1$ block
    included, namely $h_1$, and iii) the height of $B_2$, namely $h_2$.  With
    these three extra parameters, all threads can map to their corresponding
    locations in the whole trapezoid. The expression for $\mathcal{H}$
    acting on trapezoids is now:
    \begin{equation}
        \label{eq_trapezoids}
        \mathcal{H}(\omega) = (\delta_x + \omega_x + qb + kG_x, \delta_y + \omega_y -kh_2 + 2qb + 1)
    \end{equation}
    where $k = (h_1 - \omega_y) >> 31$ is a $1 | 0$ signed 32-bit integer mask that is zero when threads are in the simplex or $B_1$ region (in thread space), otherwise is takes the value $k=1$ to move the $B_2$ region of thread-blocks to the corresponding location in the data domain. Parameter $G_x$ is the width of the orthotope (grid) and is accessible to all threads. 
    
    The next sub-section presents the extension of $\mathcal{H}$ for $3$-simplices.

\subsection{Mapping to $3$-Simplices}
For a $3$-simplex of size $n$ per dimension, denoted as $\Delta_n^3$, its volume
is given by the tetrahedral numbers
\begin{equation}
    V(\Delta_n^3) = \frac{n(n+1)(n+2)}{6}.
\end{equation}

The formulation extends from the 2D one, now satisfying $x + y + z \le n$ and packing a set of $3$-orthotopes.
\begin{lemma}
\label{lemma_3}
    Given a discrete orthogonal $3$-simplex $\Delta_n^3$, a set $S_n^3$ of self-similar orthotopes produces a space of $V(S_n^3) = V(\Delta_{n-1}^3)$.
\end{lemma}
\begin{proof}
Let $S_n^3$ be a set of self-similar regular orthotopes where its total volume
is defined by the recurrence
\begin{equation}
    V(S_n^3) = \Big(\frac{n}{2}\Big)^3 + 2V(S_{n/2}^3) = \frac{n^3}{2}
    \sum_{i=1}^{\log_2(n)} \Big(\frac{1}{4}\Big)^i.
\end{equation}
A reduction via geometric series produces
\begin{alignat}{2}
    \label{eq_m3}
    V(S_n^3) &= \frac{n^3 - n}{6} = V(\Delta_{n-1}^{3}).
\end{alignat}
\end{proof}
Here the diagonal plane is not considered, thus $V(S_{n}^3) = V(\Delta_{n-1}^3)$. The way how $S_n^3$ (red cubes) is placed onto $\Delta_n^3$ (gray region) is illustrated in Figure \ref{fig_optimization-m3}.  
\begin{figure}[ht!]
\centering
\includegraphics[scale=0.056]{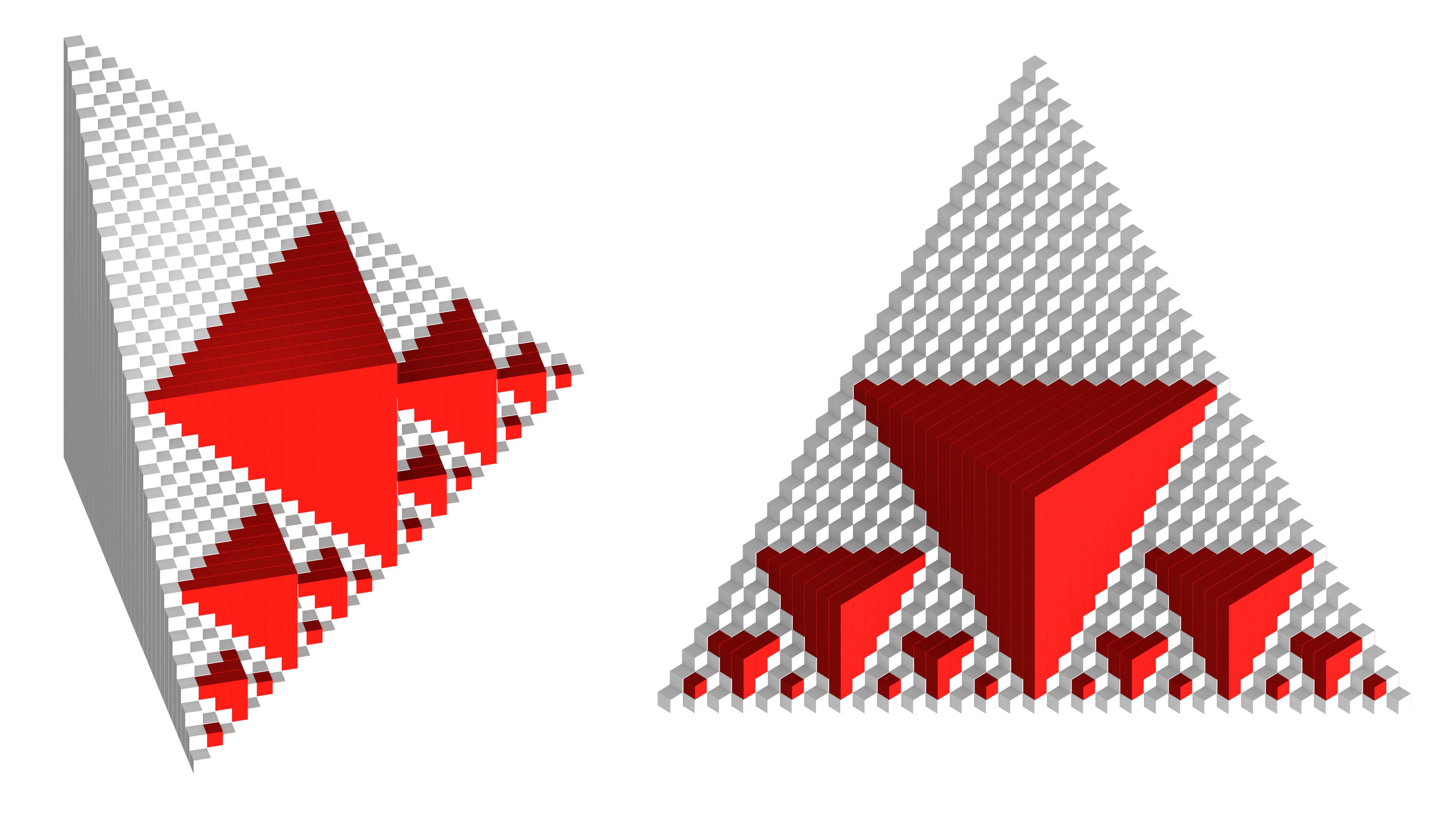}
\caption{The organization of $S_n^3$ (red) on $\Delta_{n-1}^3$ (grey).}
\label{fig_optimization-m3}
\end{figure}

Although the extra red regions can be re-used by mapping them onto their upper tetrahedron-like regions, recursively, packing the entire set into a super-orthotope of regular base $\frac{n}{2} \times \frac{n}{2}$ (which is the idea behind $\mathcal{H}$ in $2$-simplices) will introduce a fraction of extra thread-space.

\begin{lemma}
\label{lemma_4}
Given a set $S_n^3$ with a half-size sub-division scheme, packing $S_n^3$ into a super-orthotope of regular base $\frac{n}{2} \times \frac{n}{2}$ introduces a fraction of extra thread space.
\end{lemma}
\begin{proof}
    The largest sub-orthotope of the set $S_n^3$ has dimensions $\frac{n}{2} \times \frac{n}{2} \times \frac{n}{2}$, therefore two of the dimensions of the
    super-orthotope $\Pi_{a,b,c}^3$ will necessarily be $a=\frac{n}{2}, b=\frac{n}{2}$ acting as the regular base. The third dimension will stack the rest of the sub-sets of regular orthotopes. At recursion level $k$, $S_n^3$ provides $2^{k-1}$ regular orthotopes of size $(\frac{n}{2^{k}})^3$, increasing the third dimension of $\Pi_{a,b,c}^3$ by $\frac{n}{2^k}$. This extra space induced on $\Pi_{a,b,c}^3$ by the height increase is a stack of size $\frac{n}{2} \times \frac{n}{2} \times \frac{n}{2^k}$. When $k>1$ the elements provided by $S_n^3$ are not sufficient to fill the $k$-th stack,\textit{i.e.}, 
    \begin{equation}
        \sum_{i=1}^{2^{k-1}} \Big(\frac{n}{2^k}\Big)^3 = \frac{n^3}{2^{2k+1}} < \frac{n^3}{2^{k+2}}.
    \end{equation}
     Therefore, the super-orthotope $\Pi_{a,b,c}^3$ containing $S_n^3$ will necessarily have empty spaces at each stack level. 
\end{proof}
By Lemma \ref{lemma_4}, $\mathcal{H}$ will not be efficient by Definition \ref{def_1} as it introduces an extra space of $O(n^3)$. Nonetheless, such map may still be useful in practice as long as the extra space is a small fraction of $\Delta_{n}^3$. A convenient way to pack the sub-orthotopes in thread space is to displace the first major cube below and begin the stacking process horizontally from level $k \ge 2$. Such
approach produces a super-orthotope $\Pi_{a,b,c}^3$ of dimensions $\Big(\frac{n}{2}\Big) \times \Big(\frac{n}{2}\Big) \times \frac{3(n-1)}{4}$ for
$a,b,c$, respectively (see Figure \ref{fig_map-fast-m3}, left).  
\begin{figure}[ht!]
\centering
\includegraphics[scale=0.065]{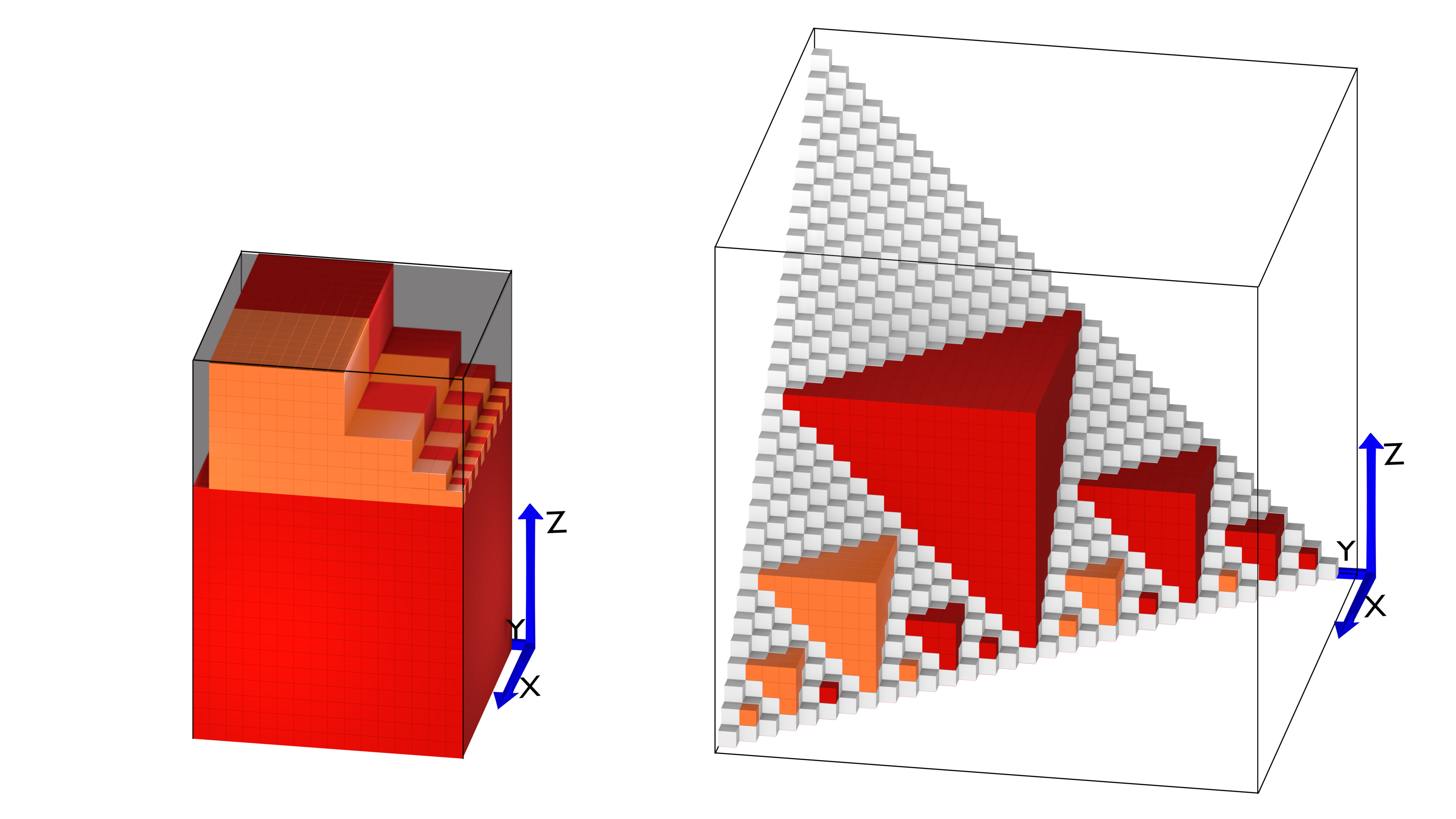}
    \caption{A practical packing scheme to map $S_n^3$ onto $\Delta_{n-1}^3$.}
\label{fig_map-fast-m3}
\end{figure}

While the packing approach is not perfectly efficient, in the infinite limit of $n$, the extra fraction of space of
$\Pi_n^3$ over $\Delta_{n-1}^3$, denoted as $\alpha(\Pi_n^3,\Delta_{n}^3)$, is 
\begin{equation}
    \alpha(\Pi_n^3,\Delta_{n-1}^3) = \lim_{n\to\infty}\frac{V(\Pi_n^{3})}{V(\Delta_{n-1}^{3})} - 1 =
    \lim_{n\to\infty}\frac{\frac{3n^2(n-1)}{16}}{\frac{(n-1)n(n+1)}{6}} - 1=
    \frac{1}{8}
\end{equation}
which is $12.5\%$ of $V(\Delta_n^3)$. Such amount of thread space constitutes a small fraction compared to a bounding-box strategy that surrounds the tetrahedron and generates practically $600\%$ of the space of $\Delta_n^3$.
Therefore, while $\Pi^3$ is not an efficient super-orthotope, there is still potential
performance improvement that can be exploited by GPUs assuming a time efficient $\mathcal{O}(1)$ cost map $\mathcal{H}$ can exist.

\begin{theorem}
    For any discrete $3$-simplex $\Delta_{n-1}^3$, there exists a time efficient map $\mathcal{H}: \mathbb{Z}_+^2 \mapsto \mathbb{Z}_+^2$ from
    $\Pi_{n/2,n/2,3(n-1)/4}^2$ to $\Delta_{n-1}^3$ that requires only
    $\mathcal{O}(1)$ arithmetic or transcendental operations.
\end{theorem}
\begin{proof}
Analogous to Theorem (\ref{theorem_1}), for convenience in the GPU programming model, the origins of both $\Pi_{a,b,c}^3$ and $\Delta_n^3$ are placed at the lower-right corner as shown in Figure \ref{fig_map-fast-m3}, with the axes aligned to the orthogonal facets. It is worth noticing that the $z$ axis point upwards. By Lemmas \ref{lemma_3} and \ref{lemma_4}, there exists a super-orthotope that contains the set $S_n^3$, with each sub-set of equally
sized regular orthotopes organized by stack levels, except for the first sub-orthotope that is treated as a special case for being displaced from the stacking pattern. 
    
Let $h(\omega)$ be an auxiliary block-space mapping of the single sub-orthotope of size $(n/2)^3$ onto the center of the tetrahedron, defined as
\begin{equation}
    h(\omega) = \omega + (0,\frac{n}{2} + 1,0).
\end{equation}
For the $b,q$ parameters, they keep the definitions from Theorem \ref{theorem_1} as the mapping layout allows projection onto the plane with an additional consideration
for the extra spaces that fall out of $\Delta_n^3$ (the red spaces from Figure \ref{fig_optimization-m3}), for which the position is transposed in $x,y$ and complemented in $z$, acting as a hinge. With these settings, the combination of the parameters $\omega, q, b$ and the auxiliary function $h(\omega)$ allow the formulation of
\begin{equation}
\label{eq_3simplex_map}
\small
\mathcal{H}(\omega) = 
    \begin{cases} 
        h(\omega),\ \omega \in S_{\frac{n}{2}}^3 \\
        (\omega_x + qb, \omega_y + 2qb + 1, \omega_z - n/2),\ \text{otherwise}

    \end{cases}
\end{equation}
When the mapped blocks fall outside $\Delta_{n-1}^3$, then its coordinates are remapped following a hinge like pattern. Defining $L = (\omega_x \mod b, \omega_y \mod b, \omega_z \mod b)$ as the local coordinates of $\omega$ in their corresponding $S^3$ region, the hinge movement becomes $H(\omega) + (b-1 - 2L_x, b-1 - 2L_y, 2b-1 - L_z)$. This map is only time efficient, \textit{i.e.}, $T(h(\omega)) + T(\mathcal{H}(\omega)) = \mathcal{O}(1)$.
\end{proof}
It is possible to make $\mathcal{H}$ fully efficient, \textit{i.e.}, to satisfy condition $V(\Pi^3) = V(\Delta_n^3) \pm o(n^m)$ of Definition \ref{def_1}, by including concurrent GPU kernel executions in the analysis. Such approach would allow to create multiple parallel spaces, one for each stack level of equally sized sub-orthotopes, without needing to pack them into a single super-orthotope. The number of simultaneous parallel spaces would be limited by the maximum number of concurrent kernels supported by the GPU, which as of 2022 is up to 128 concurrent kernels. In practice, such limit is sufficiently large for many applications as the number of stack levels is in the order of $O(\log_2(n))$. Furthermore, the map works in block-space which acts as an extra factor for reaching larger problem sizes. The next Section presents experimental performance results for $\mathcal{H}$ using $2$-simplex and $3$-simplex tests, running on different GPUs. The results are compared to state of the art approaches.

\section{Experimental Evaluation}
\label{sec_experimental-results}
The proposed map $\mathcal{H}$, as well as state of the art approaches, are measured and compared against the bounding box approach. This gives a total of four approaches (Figure \ref{fig:all-approaches}) to be compared:
\begin{itemize}
    \item \textbf{RB}: Rectangular box \cite{Jung2008}
    \item $\bm{\lambda}$: Lambda map \cite{navhitmat2014,8392762}
    \item \textbf{DP}: CUDA's Dynamic Parallelism
    \item $\bm{\mathcal{H}}$: New proposed map
\end{itemize}
All approaches are available to the community at \url{https://github.com/temporal-hpc/simplex-gpu-mappings}.

\begin{figure*}[ht!]
\centering
\includegraphics[scale=0.13]{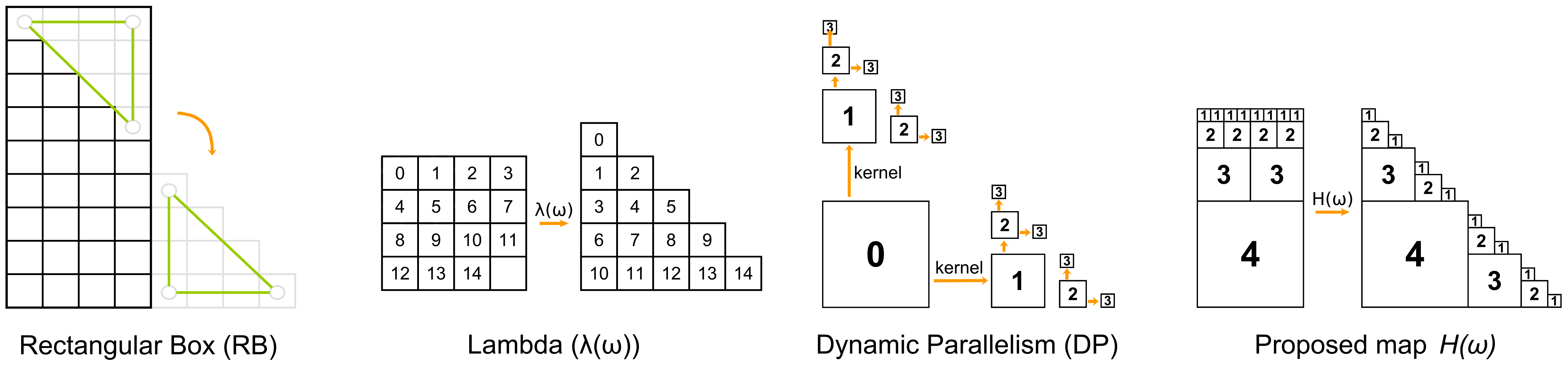}
\caption{All four approaches illustrated for $2$-simplices.}
\label{fig:all-approaches}
\end{figure*}

\subsection{Experimental Design}
Four different computational patterns are tested on $2$-simplex and $3$-simplex domains. These are listed in Table \ref{tab:tests}. 
\begin{table}[ht!]
\centering
\caption{Set of kernel tests and simplex domains}
\label{tab:tests}
\resizebox{\columnwidth}{!}{
\begin{tabular}{|l|c|c|}
\hline
Test                                        & 2-simplex & 3-simplex \\ \hline
1) [\textbf{MAP}] Only map                  & *         & *         \\ \hline 
2) [\textbf{ACCUM}] $+1$ Accumulation       & *         & *         \\ \hline
3) [\textbf{EDM}] Euclidean Distance Matrix & *         &           \\ \hline
4) [\textbf{CA}] Cellular Automaton         & *         & *         \\ \hline
\end{tabular}
}
\end{table}

These tests serve to know what is the impact of application work in the overall speedup through an improved mapping. The first test, MAP, is a GPU kernel with no application work, \textit{i.e.}, just the thread mapping stage. It uses C++ volatile variables in order to prevent the compiler to optimize the mapping stage as it is not used afterwards. This kernel should reflect the theoretical speedup as it is mostly mapping work. The ACCUM test is a kernel where each mapped thread adds $+1$ to its corresponding data-element in the simplex domain. It is a memory-bound kernel as it does 1 read, 1 write and only 1 addition on each data element. The EDM test computes the Euclidean distance matrix, that is the distance $d(p_i, p_j)$ at each data-element of the simplex with location $i,j$. This test is a medium intensity application with more arithmetic work than the accumulation test, as it computes a square root for each data-element. Lastly, the cellular automaton computes John Conway's game of life \cite{ConwaysLife} on a simplex structure for fixed number of steps, with periodic boundary conditions for the $2$-simplex, and free boundary conditions (fixed dead cells on the boundary) for the $3$-simplex case. This test is a more intense memory-bound application which serves as a case where improvements to the mapping stage could have a lesser impact on performance.  

The average running time as well as the power consumption are measured during kernel execution. Sufficient repeats are employed for each test in order to reach a statistical standard error of less than $1\%$. Then, plots of speedup (over the bounding-box approach) and energy efficiency are computed from these values.

The hardware used for all tests is detailed in Table \ref{tab:hardware}. The TITAN GPUs share the same system while the A100 GPU is from a DGX A100 node of the Patag\'on supercomputer of Austral University of Chile \cite{patagon-uach}. 
\begin{table*}[ht!]
\label{tab:hardware}
\centering
\resizebox{0.85\textwidth}{!}{
\begin{tabular}{|ll|l|l|l|}
\hline
\multicolumn{2}{|l|}{\textbf{\cellcolor{temporal}Attribute}} & \cellcolor{temporal}\textbf{Nvidia A100} & \cellcolor{temporal}\textbf{Nvidia TITAN RTX} & \cellcolor{temporal}\textbf{Nvidia TITAN V} \\\hline
\multicolumn{2}{|l|}{\textbf{Segment}} & HPC Server (DGX A100) & Workstation & Workstation \\ \hline
\multicolumn{2}{|l|}{\textbf{Architecture}} & Ampere & Turing & Volta \\ \hline
\multicolumn{2}{|l|}{\textbf{GPU Chip}} & GA100 & TU102 & GV100       \\ \hline
\multicolumn{2}{|l|}{\textbf{FP32 Units}} & $6912$ & $4608$ & $5120$\\ \hline
\multicolumn{2}{|l|}{\textbf{SMs}} & $108$ & $72$ & $80$              \\ \hline
\multicolumn{2}{|l|}{\textbf{FP32 Units/SM}} & $64$ & $64$ & $64$     \\ \hline
\multicolumn{2}{|l|}{\textbf{Memory}} & $40\ \text{GB}$ & $24\ \text{GB}$ & $12\ \text{GB}$      \\ \hline
\multicolumn{2}{|l|}{\textbf{Memory Bandwidth}} & $1.5\ \text{TB}/s$ & $672\ \text{GB}/s$ & $651.3\ \text{GB}/s$ \\ \hline
\multicolumn{2}{|l|}{\textbf{Max Power (W)}} & $400$W & $280$W & $250$W \\ \hline
\multicolumn{1}{|l|}{\multirow{3}{*}[-0.5em]{\textbf{System}}} & \textbf{OS}  & DGX OS 5.2 (Ubuntu Based)  & Arch Linux  & Arch Linux \\ \cline{2-5} 
\multicolumn{1}{|l|}{} & \textbf{CUDA}  & 11.4.2 & 11.6 & 11.6 \\\cline{2-5} 
\multicolumn{1}{|l|}{} & \textbf{CPU} & 2 x 64-core AMD EPYC 7742 & Intel 10-core i7-6950X & Intel 10-core i7-6950X \\\cline{2-5} 
\multicolumn{1}{|l|}{} & \textbf{RAM} & 1 TB RAM DDR4-3200Hz & 128GB DDR4 2400Mhz & 128GB DDR4 2400Mhz\\ \hline
\end{tabular}
}
\caption{Hardware specifications of the GPUs and computer systems used for the experiments.}
\end{table*}

\subsection{Speedup and Energy Efficiency for $2$-simplices}
Figure \ref{fig:2s-speedup} presents the speedup of each mapping technique over the bounding box approach, under different tests (columns) and GPUs (rows). In the MAP test (first column), $\mathcal{H}$ and RB reach the theoretical speedup of $2\times$ over the Bounding box approach, for all GPUs. With the TITAN RTX, the $\lambda$ map performance gets reduced to $0.5\times$ of speedup once $n \ge 32154$. The reason is because that GPU has very few FP64 units, and at that problem size (with block size $16\times 16$ which is the fastest) the map needs to switch from FP32 to FP64 precision in the computation of a square root, otherwise it generates incorrect coordinates. The DP approach also fails to reach the maximum speedup for all GPUs. Instead it reaches up to $1.5\times$ of speedup in the best case using the A100 GPU. 
In the ACCUM test (second column), the maximum speedup is $\sim 1.2\times$, with $\mathcal{H}$ within the fastest maps. The RB map is also within the fastest, and DP becomes competitive as well once the problem size is large enough. In the case of $\lambda$, it stays with the fastest maps as long as the GPU has sufficient FP64 units, otherwise it underperforms with $\sim 0.4$ of speedup as in the TITAN RTX.   
For the EDM2D test (third column), speedup is unstable at small problem sizes, but stabilizes at medium to high values. Here, both $\mathcal{H}$ and RB stay competitive through all the range of $n$, followed by $\lambda$ (except when running in the TITAN RTX GPU) and DP.
For the CA2D test (fourth column), the maximum speedup is $1.25\times$ and is achieved by all approaches using the A100 GPU. On the TITAN RTX GPU the fastest approaches are $\mathcal{H}$ and RB, followed by DP and lastly $\lambda$. It is worth noticing that in this case the performance of $\lambda$ did not drop because the fastest version used thread blocks of size $32\times 32$, allowing the FP32 square root to be precise for the whole range. In the TITAN V GPU all approaches reach a speedup of $\sim1.15\times$ for large $n$.
\begin{figure*}[ht!]
\centering
\includegraphics[scale=0.28]{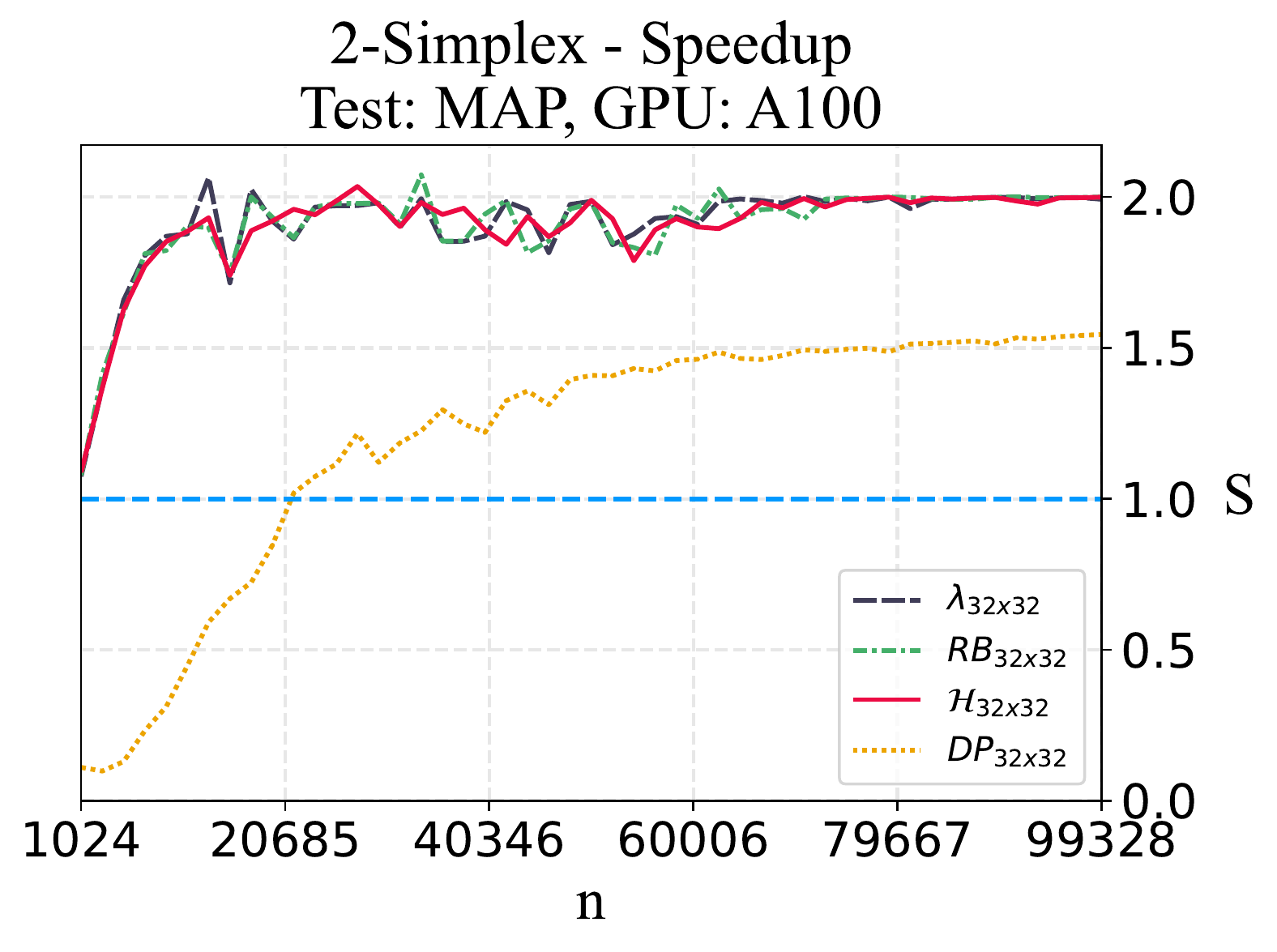}
\includegraphics[scale=0.28]{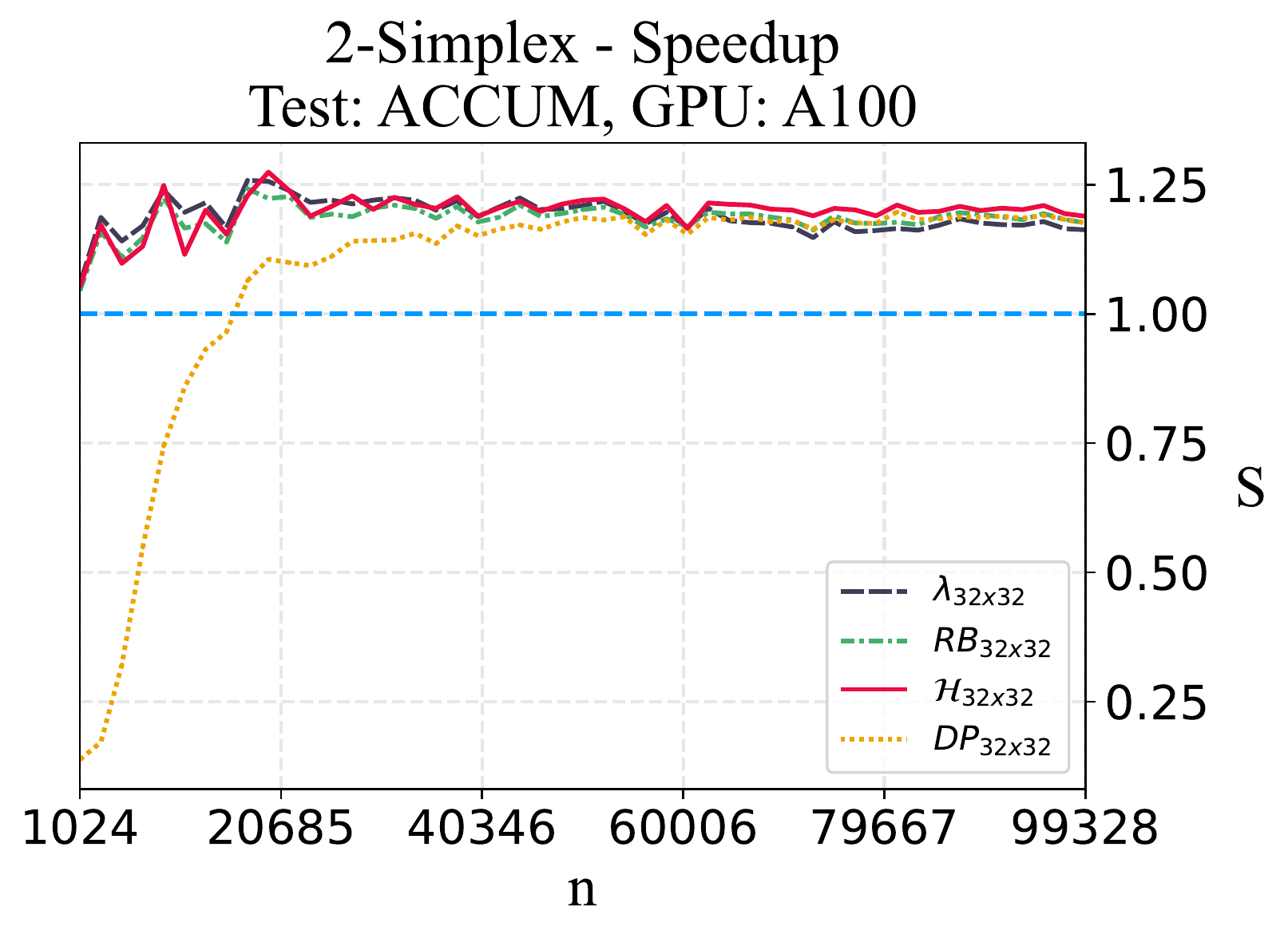}
\includegraphics[scale=0.28]{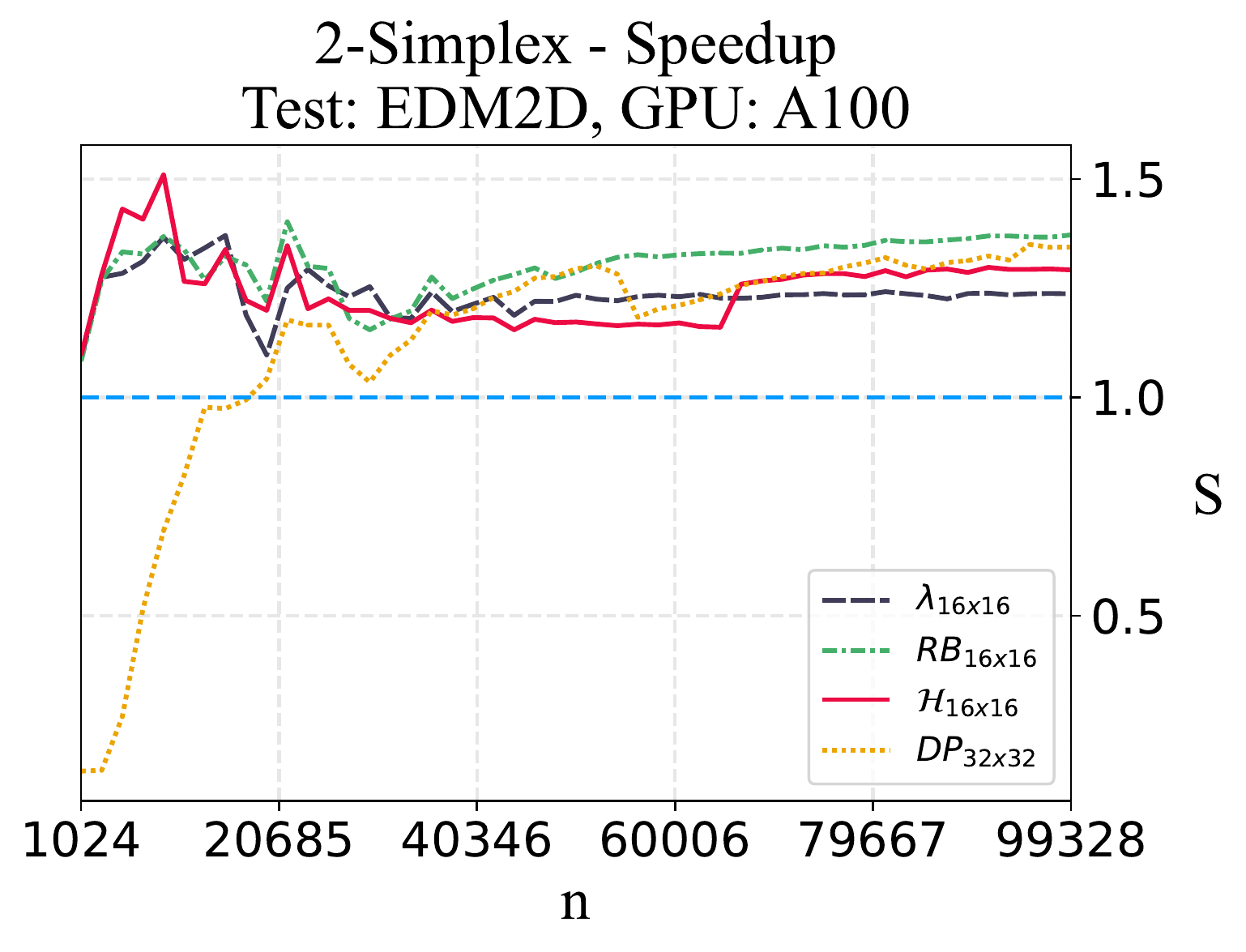}
\includegraphics[scale=0.28]{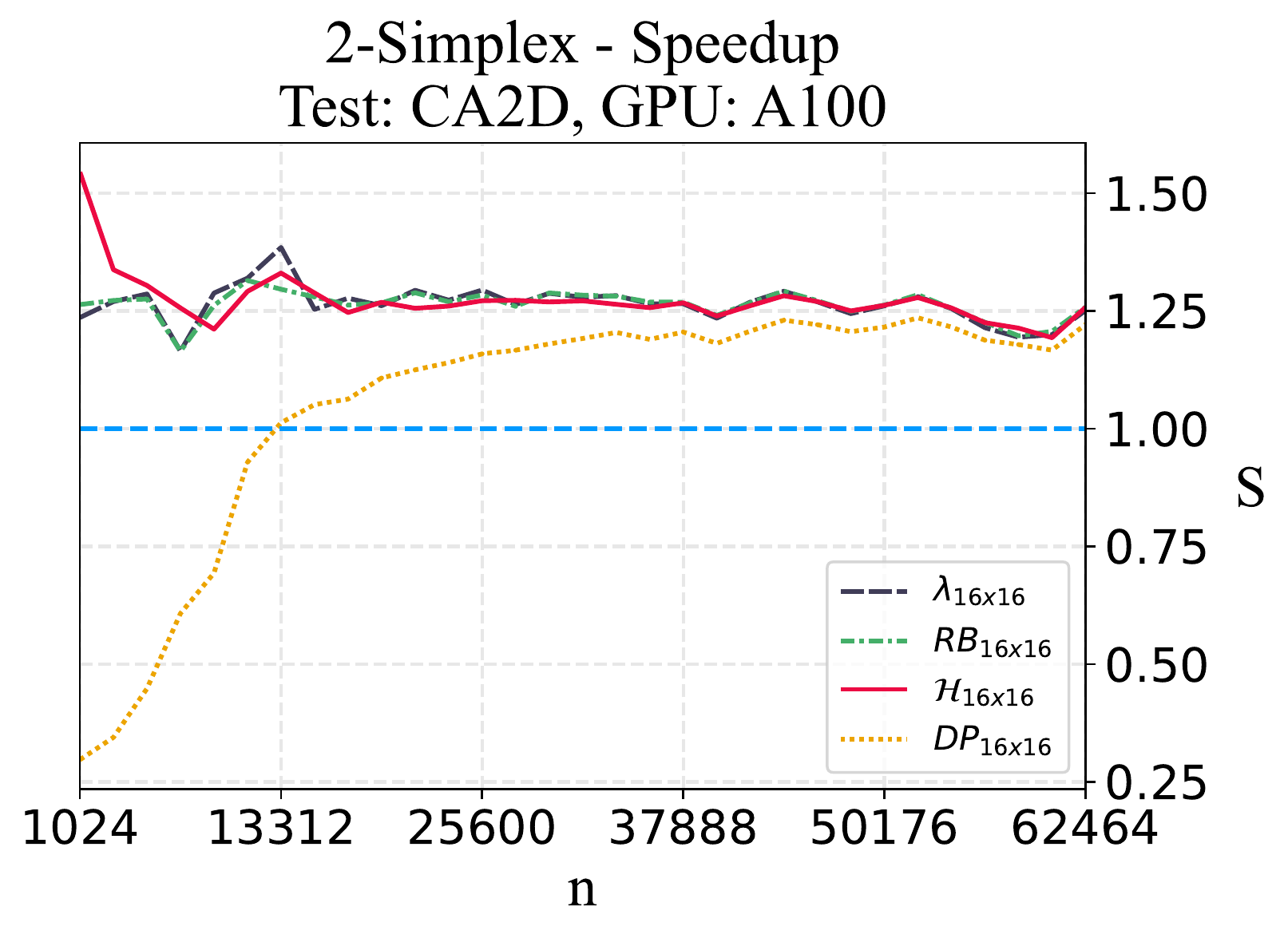}\\
\includegraphics[scale=0.28]{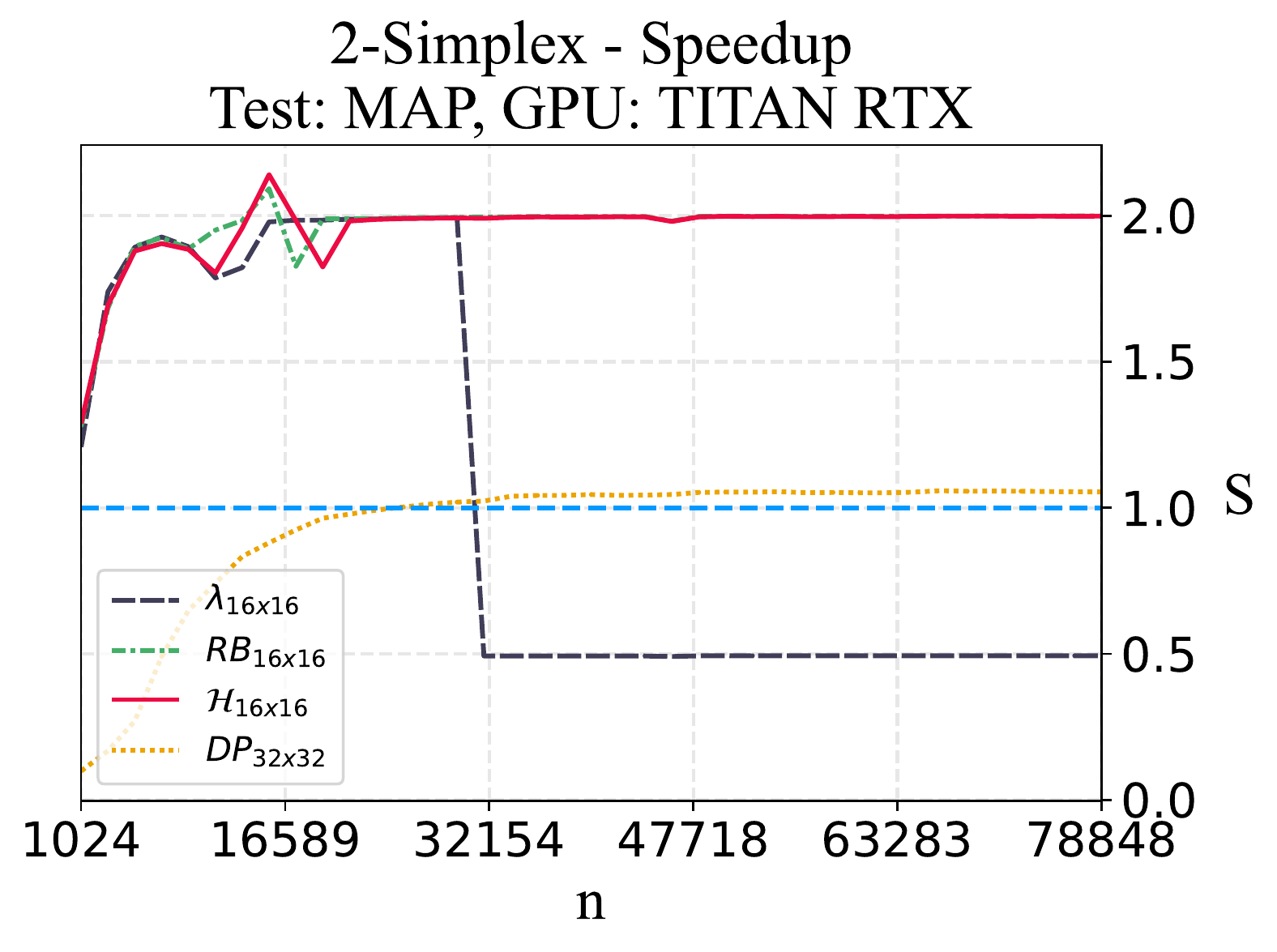}
\includegraphics[scale=0.28]{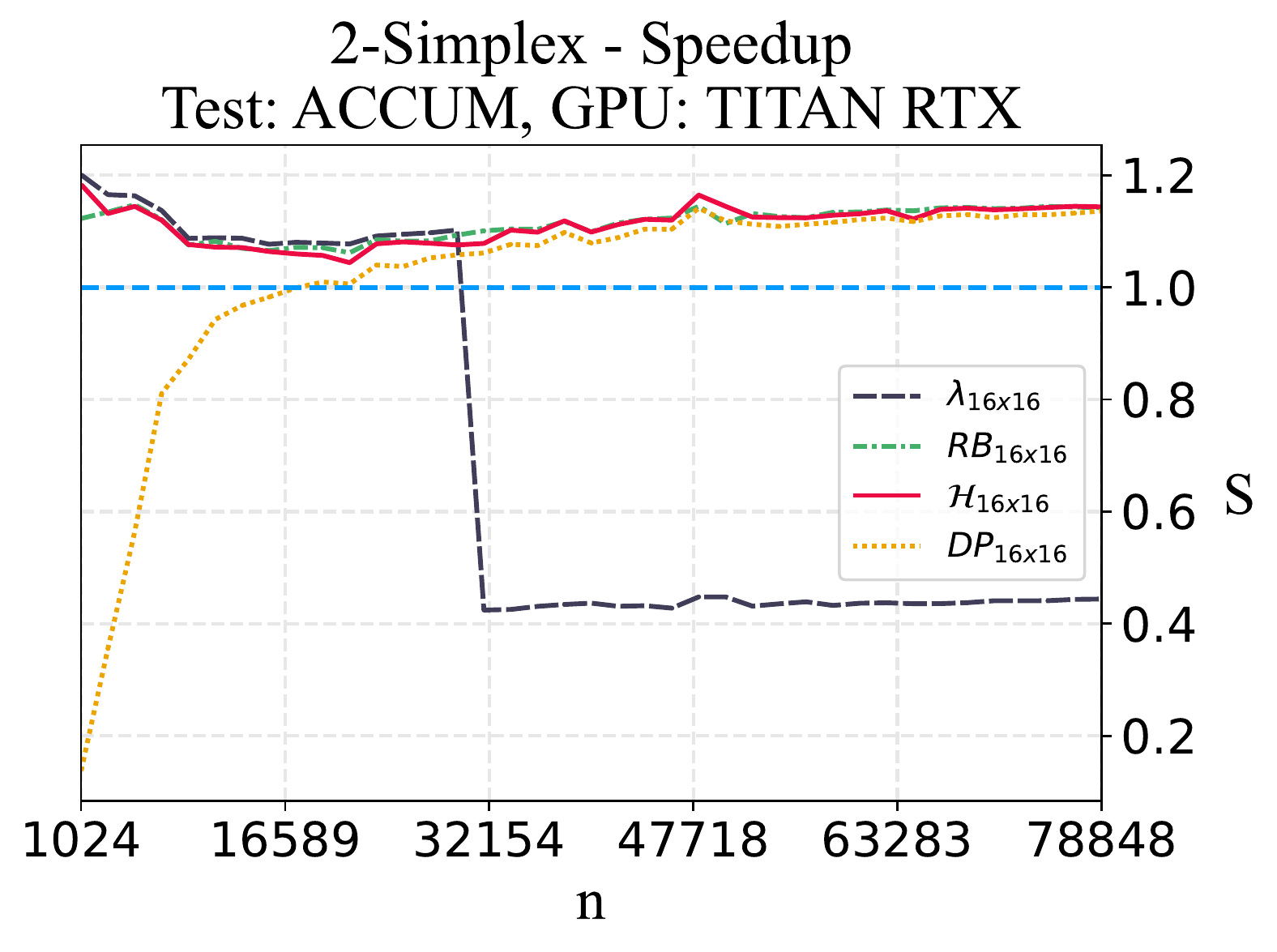}
\includegraphics[scale=0.28]{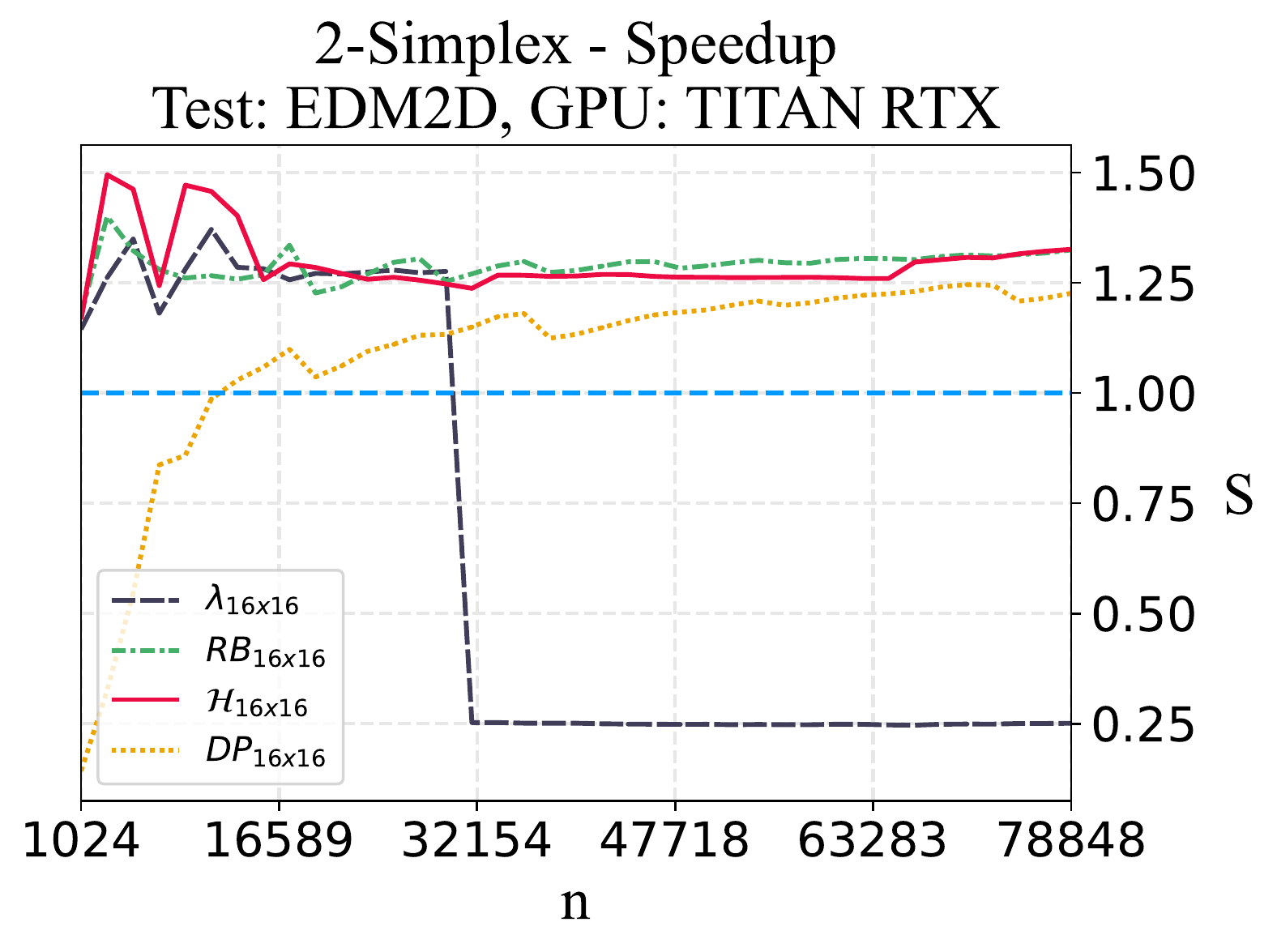}
\includegraphics[scale=0.28]{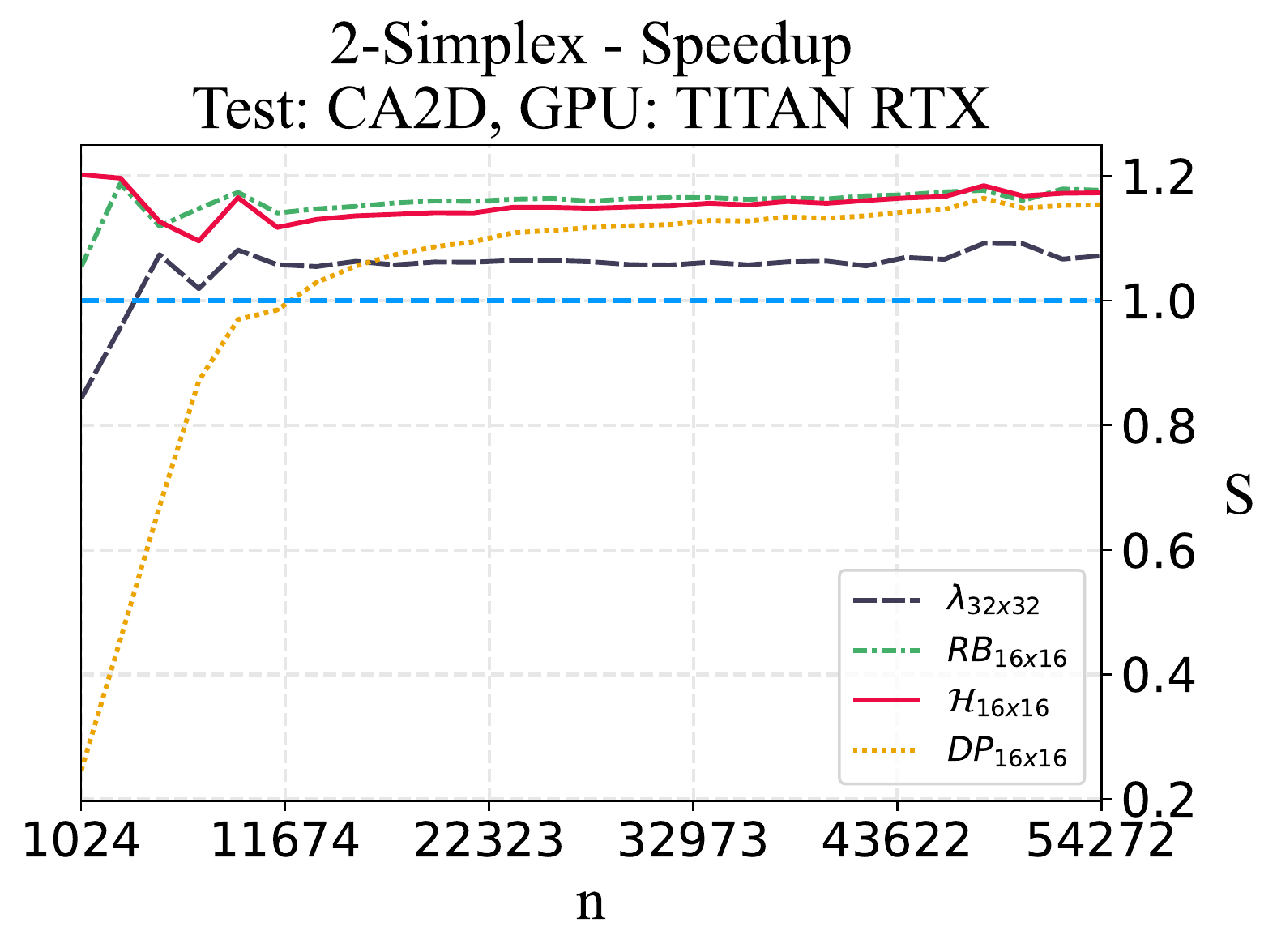}\\
\includegraphics[scale=0.278]{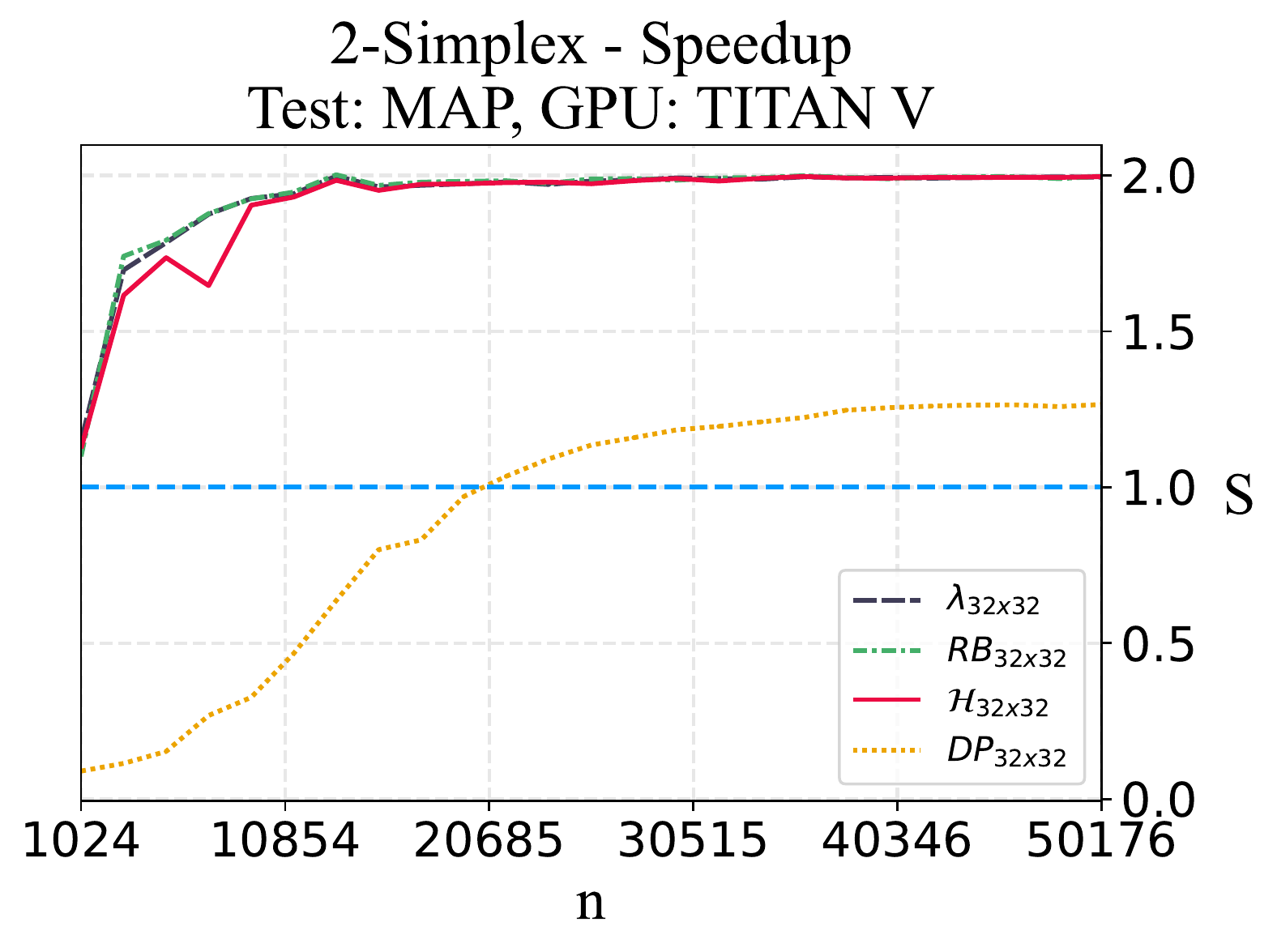}
\includegraphics[scale=0.278]{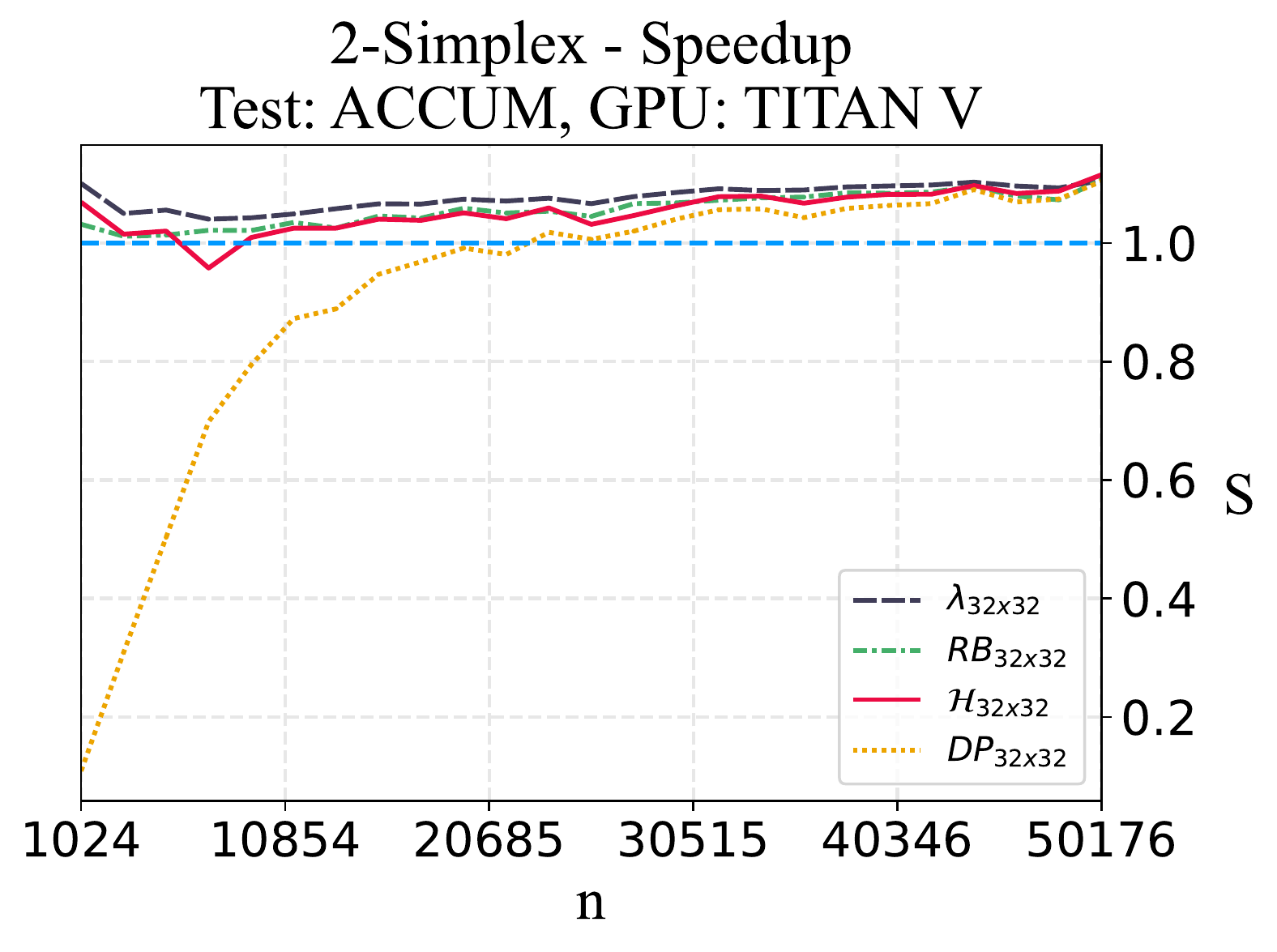}
\includegraphics[scale=0.278]{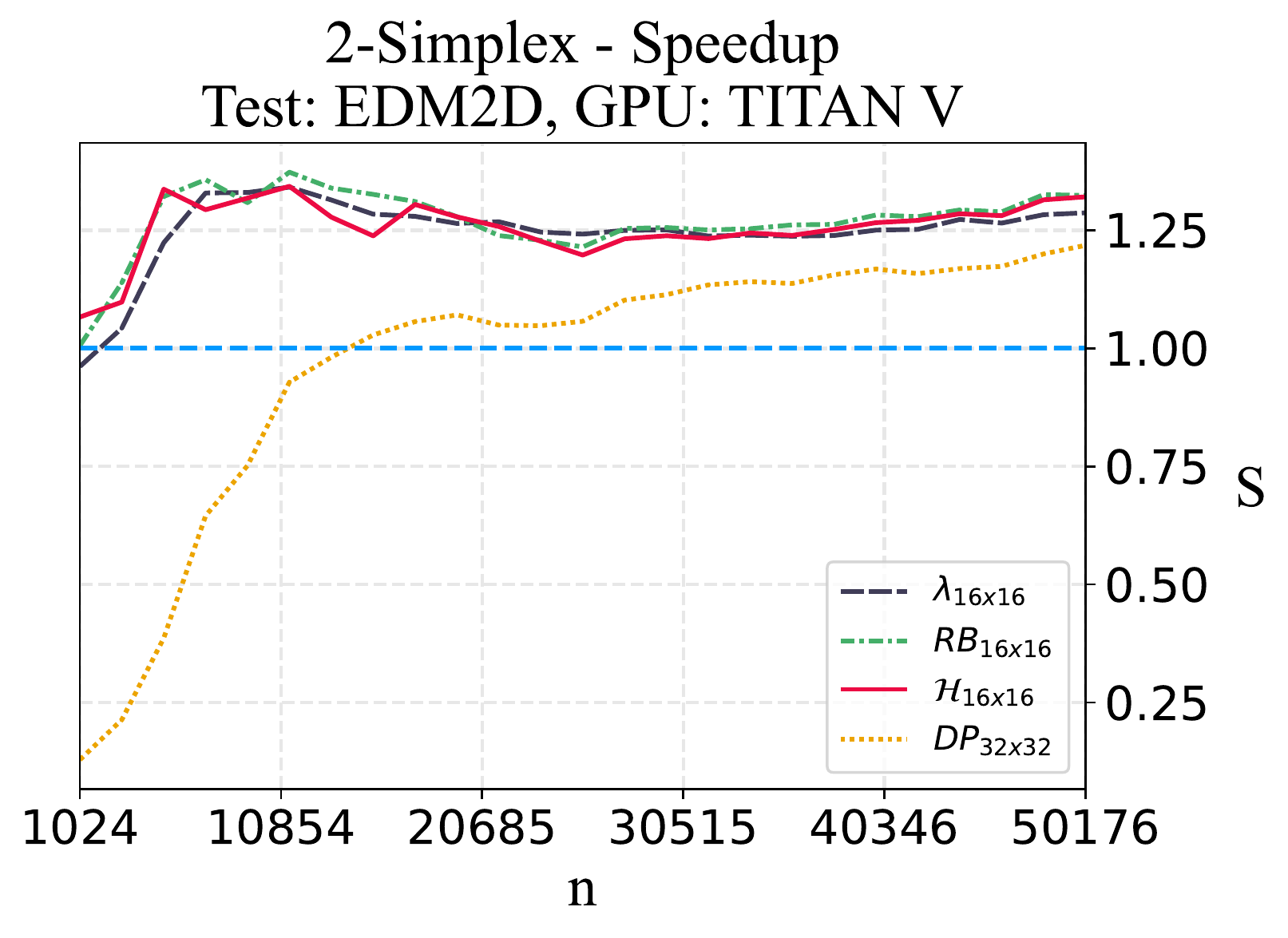}
\includegraphics[scale=0.278]{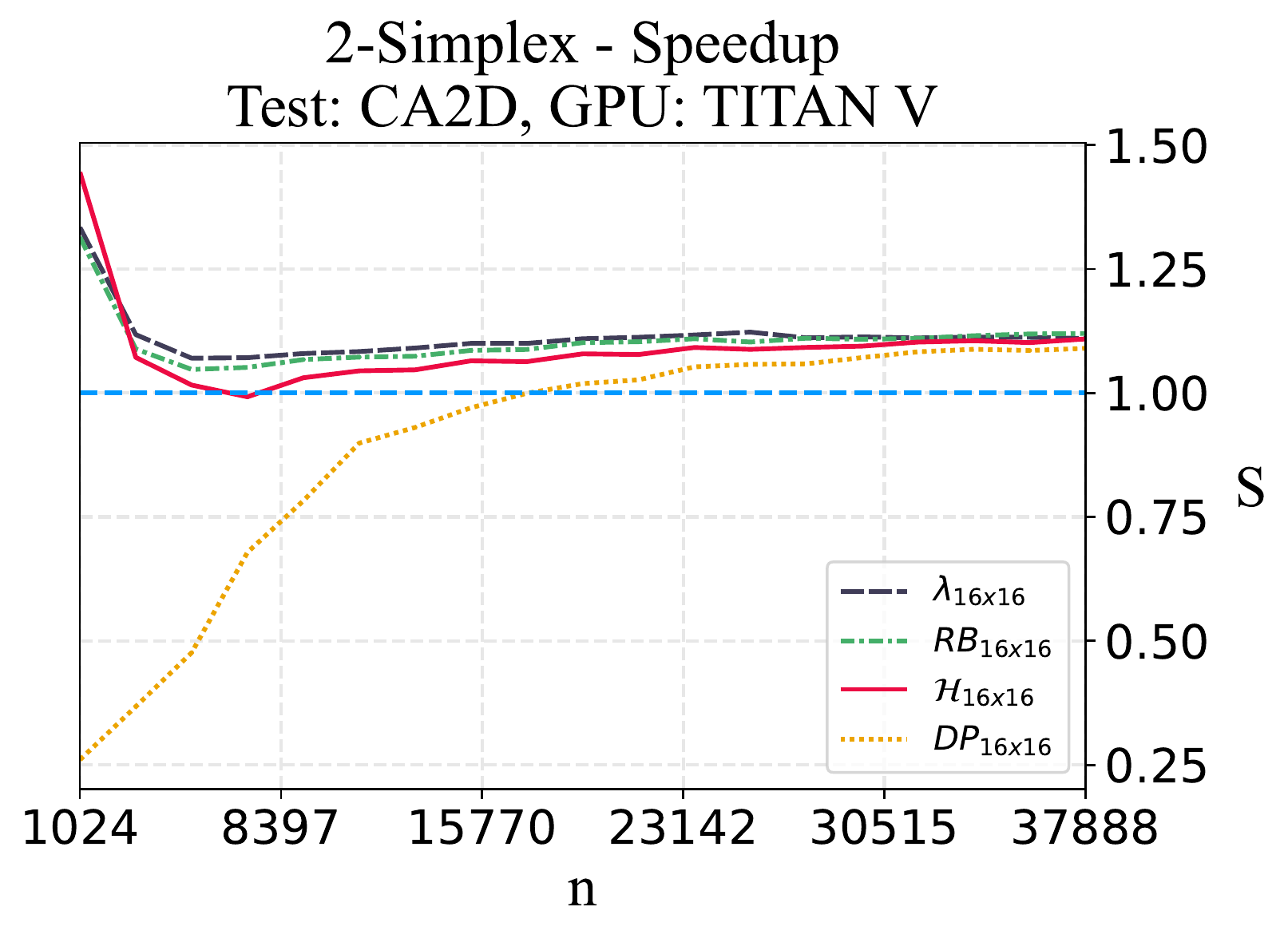}\\
\caption{Speedup for different $2$-simplex tests, running on different GPUs.}
\label{fig:2s-speedup}
\end{figure*}

Figure \ref{fig:2s-power} shows the power consumption in the case of the A100 GPU. 
\begin{figure*}[ht!]
\centering
\includegraphics[scale=0.27]{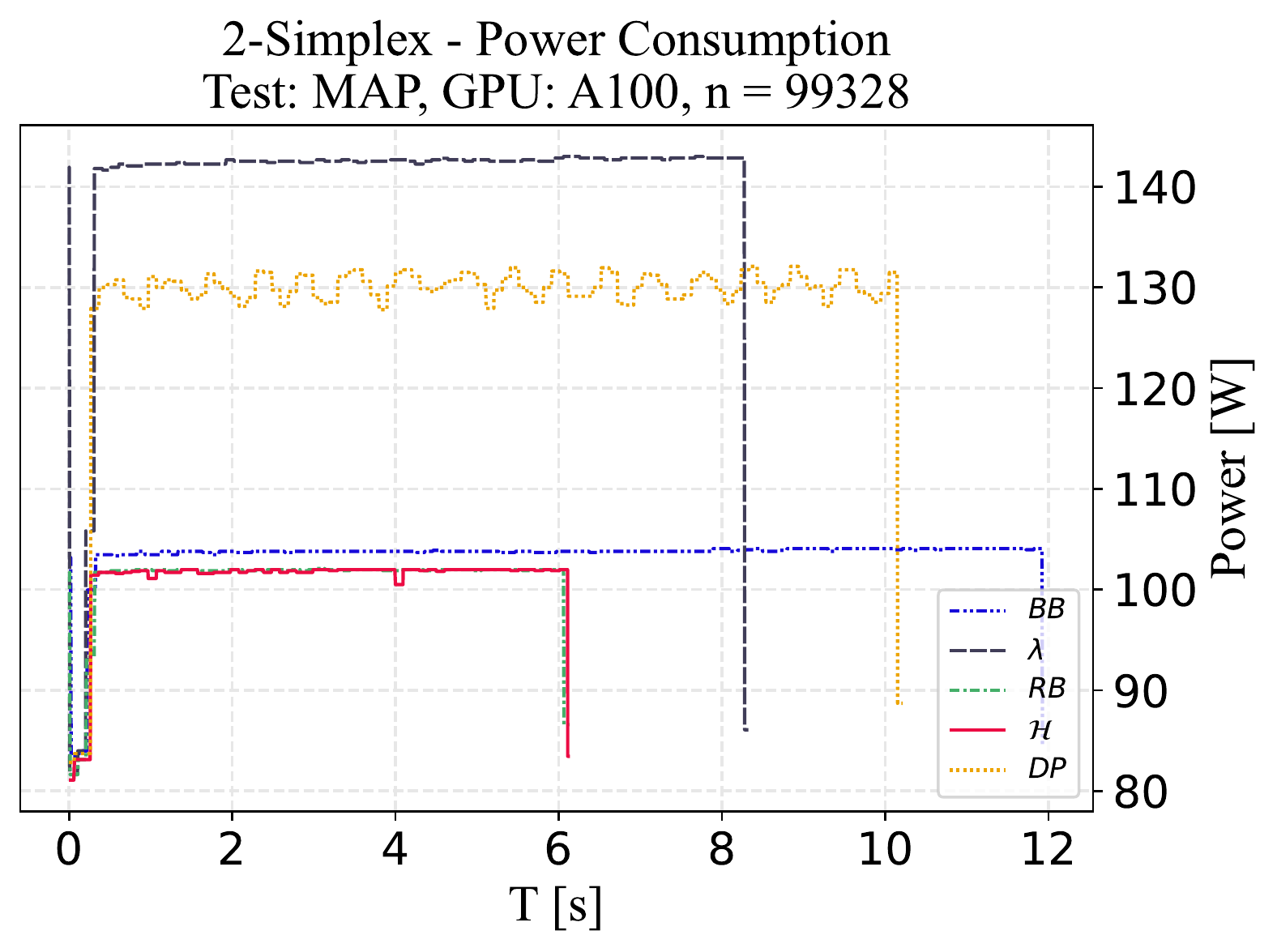}
\includegraphics[scale=0.27]{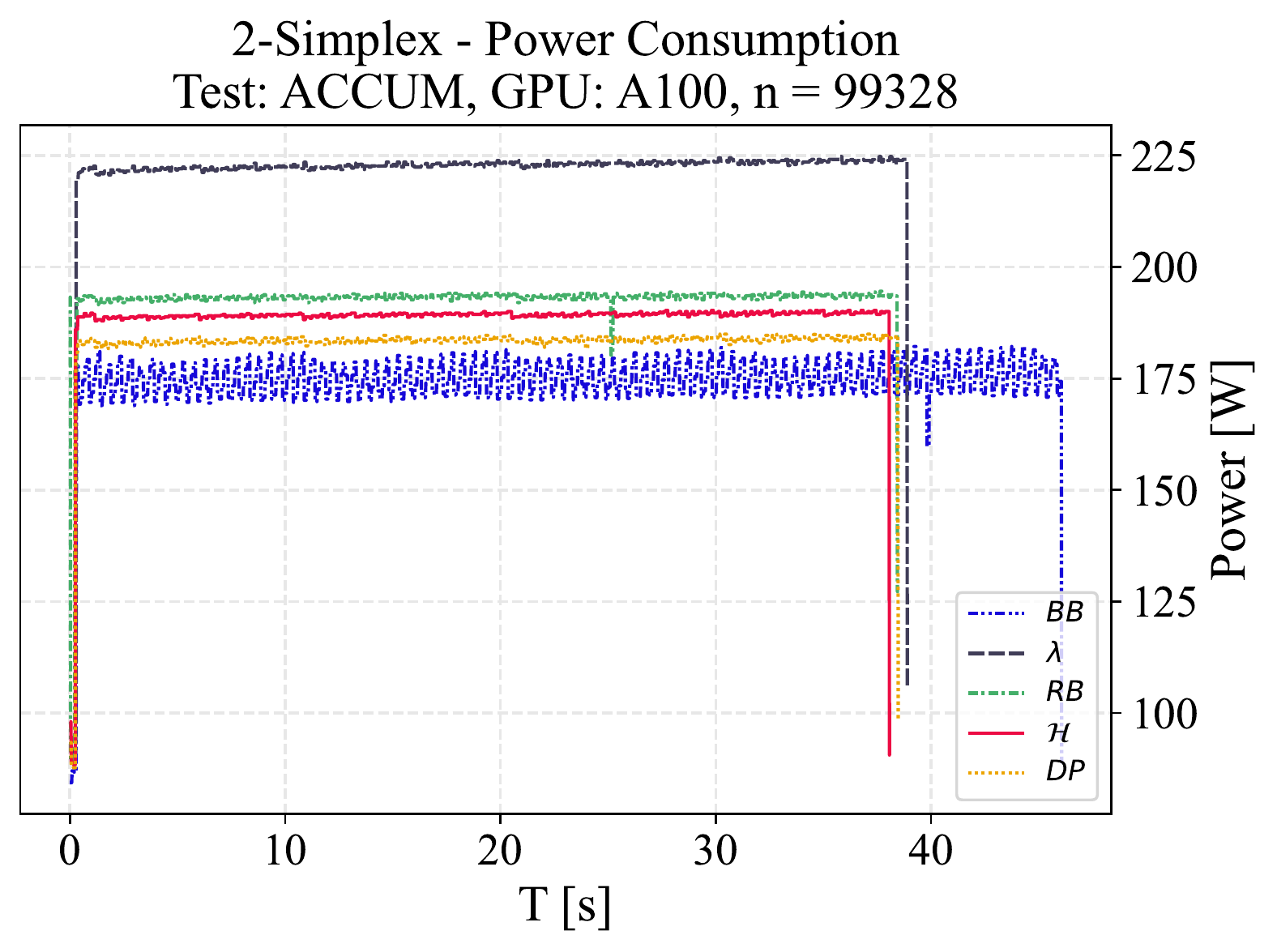}
\includegraphics[scale=0.27]{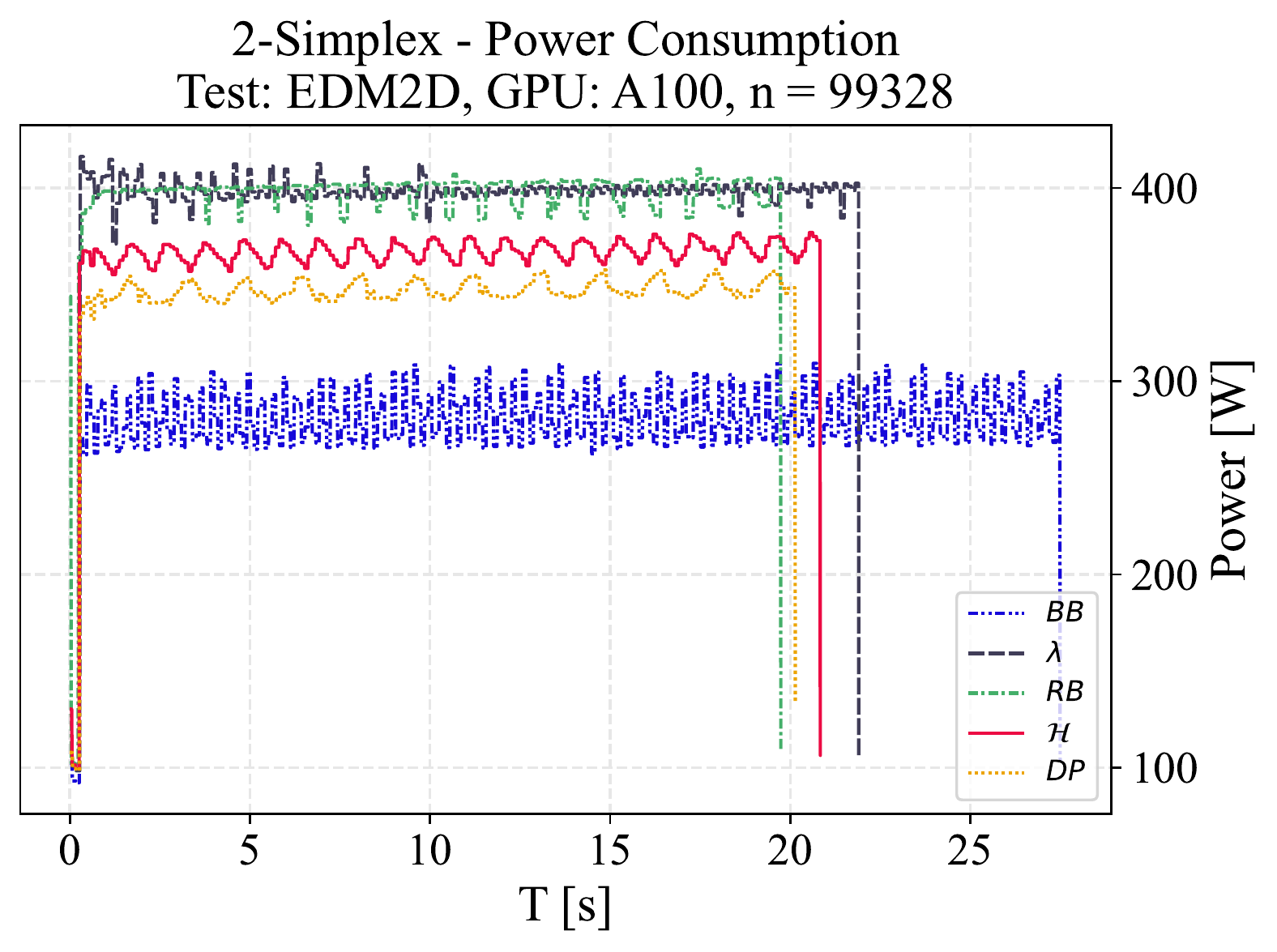}
\includegraphics[scale=0.27]{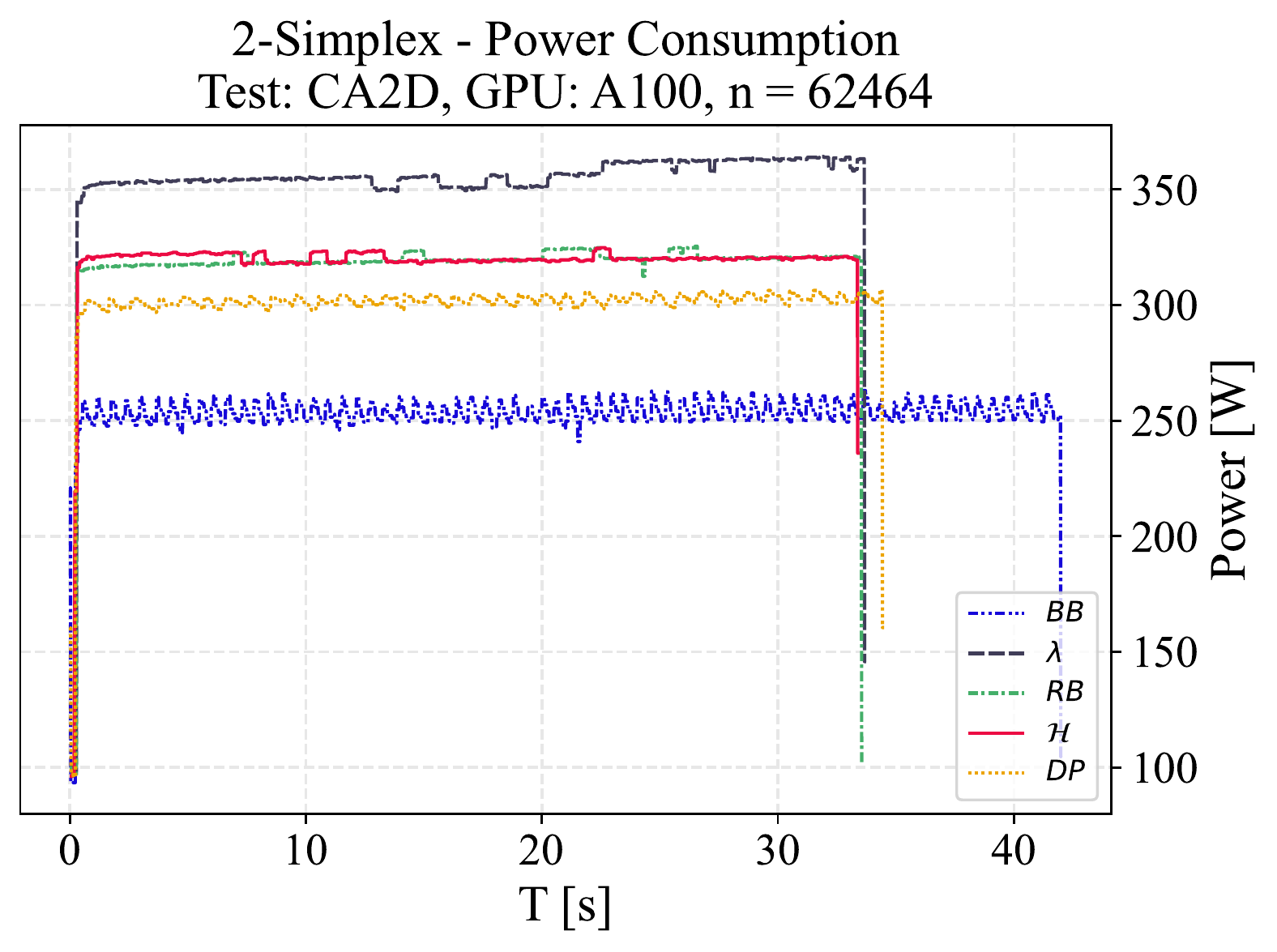}
\caption{Speedup for different $2$-simplex tests, running on different GPUs.}
\label{fig:2s-power}
\end{figure*}
As a global observation, the duration of all curves are in the same proportion of the speedups reported above. For the MAP test the peak power is near $\sim 100$W, reached by all approaches except for DP which reaches a peak slightly above $90$W. For the ACCUM test, $\mathcal{H}$ has an intermediate consumption of up to $\sim 180$W, while $\lambda$ is the most power consuming one with a peak of $\sim 210$W. In the EDM2D test, $\lambda$ reaches the Max Thermal Design Power (TDP) of the A100 GPU, which is $400$W. Here $\mathcal{H}$ is positioned in the middle again, being less consuming than RB, and slightly more than DP. The reference Bounding approach oscillates between $270$W and $300$W. In the last test, the CA2D, $\mathcal{H}$ and RB reach the same peak power consumption while DP is the less consuming one.

Figure \ref{fig:2s-energy-eff} presents the energy efficiency of all approaches running on the A100 GPU. The bar charts represent the energy efficiency using the metric of number of elements per second (EPS) per Watt ($EPS/W$), while the numerical labels at the top of each bar show the total energy used for that particular run. In terms of energy efficiency, which is in log-Y scale, $\mathcal{H}$ stays on average as one of the most energy efficient approaches, shared by RB and DP. On the other hand, the $\lambda$ approach is the less energy efficient one. 
\begin{figure}[ht!]
\centering
\includegraphics[scale=0.53]{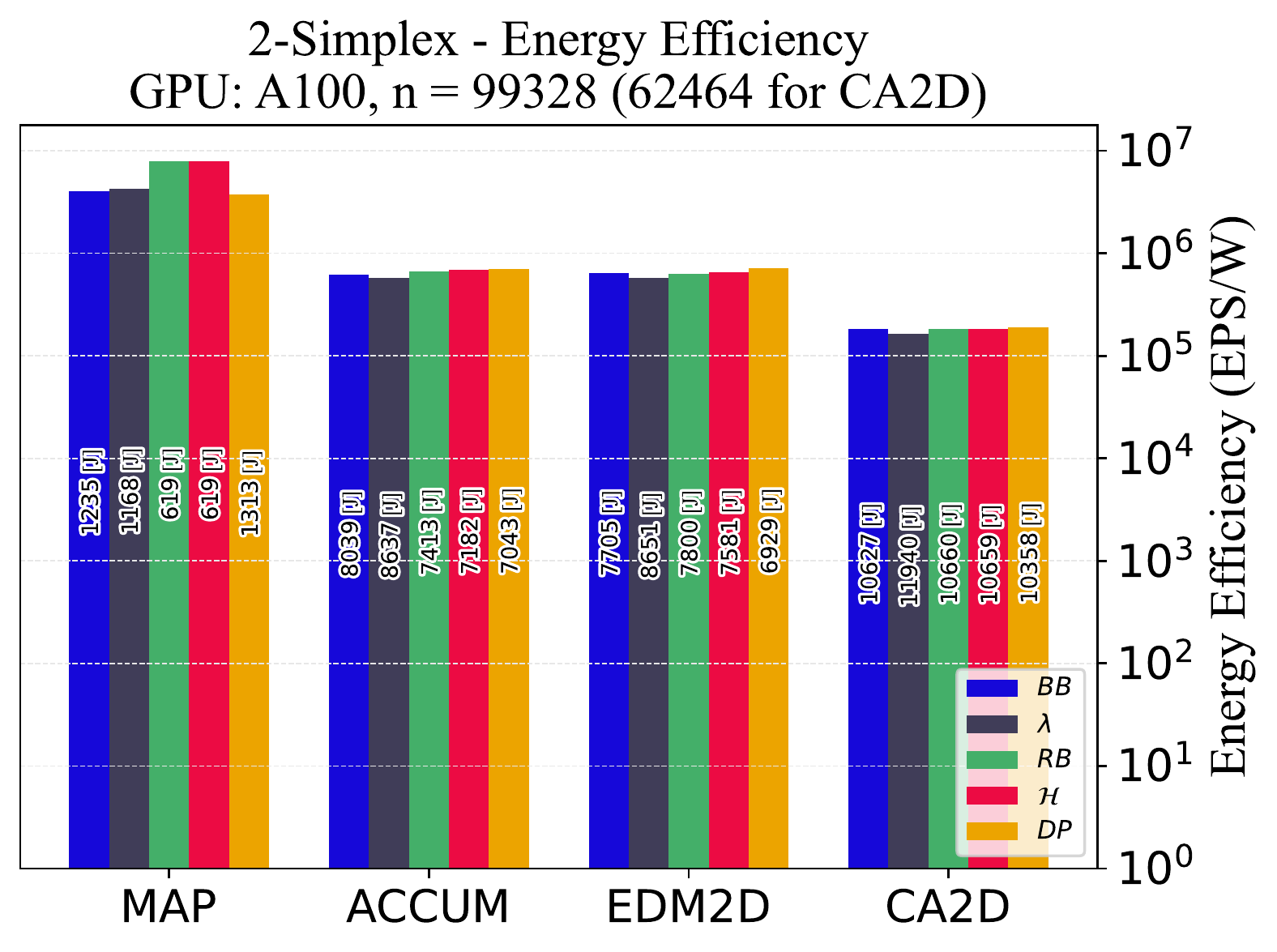}
\caption{Energy efficiency (higher is better), as well as total energy used.}
\label{fig:2s-energy-eff}
\end{figure}

\subsection{Speedup and Energy Efficiency for $3$-simplices}
In the case of $3$-simplices, the RB approach is discarded because it cannot be extended efficiently to 3D, as $n/2 \times n/2 \times n$ orthotope is double the necessary space and the arithmetic operations involved would require further research as it will be more complex. The map $\lambda$ is also discarded because it generates a $3$-simplex ($i < j < k$ condition) in which no vertex has all of its incident facets orthogonal \cite{8392762}, breaking the $x + y + z \le n$ condition. Therefore $3$-simplex tests consider $\mathcal{H}$ and two possible approaches using Dynamic Parallelism (DP). The first one ($DP_{vanilla}$) is the default extension of the DP approach for $2$-simplices, while the second one ($DP_{hinge}$) implements the hinge idea of $\mathcal{H}$ as in Figure \ref{fig_optimization-m3}, now using DP. 

Figure \ref{fig:3s-speedup} presents the speedup of $\mathcal{H}$ and the DP approaches over BB, using different $3$-simplex tests and GPUs.
\begin{figure*}[ht!]
\centering
\includegraphics[scale=0.37]{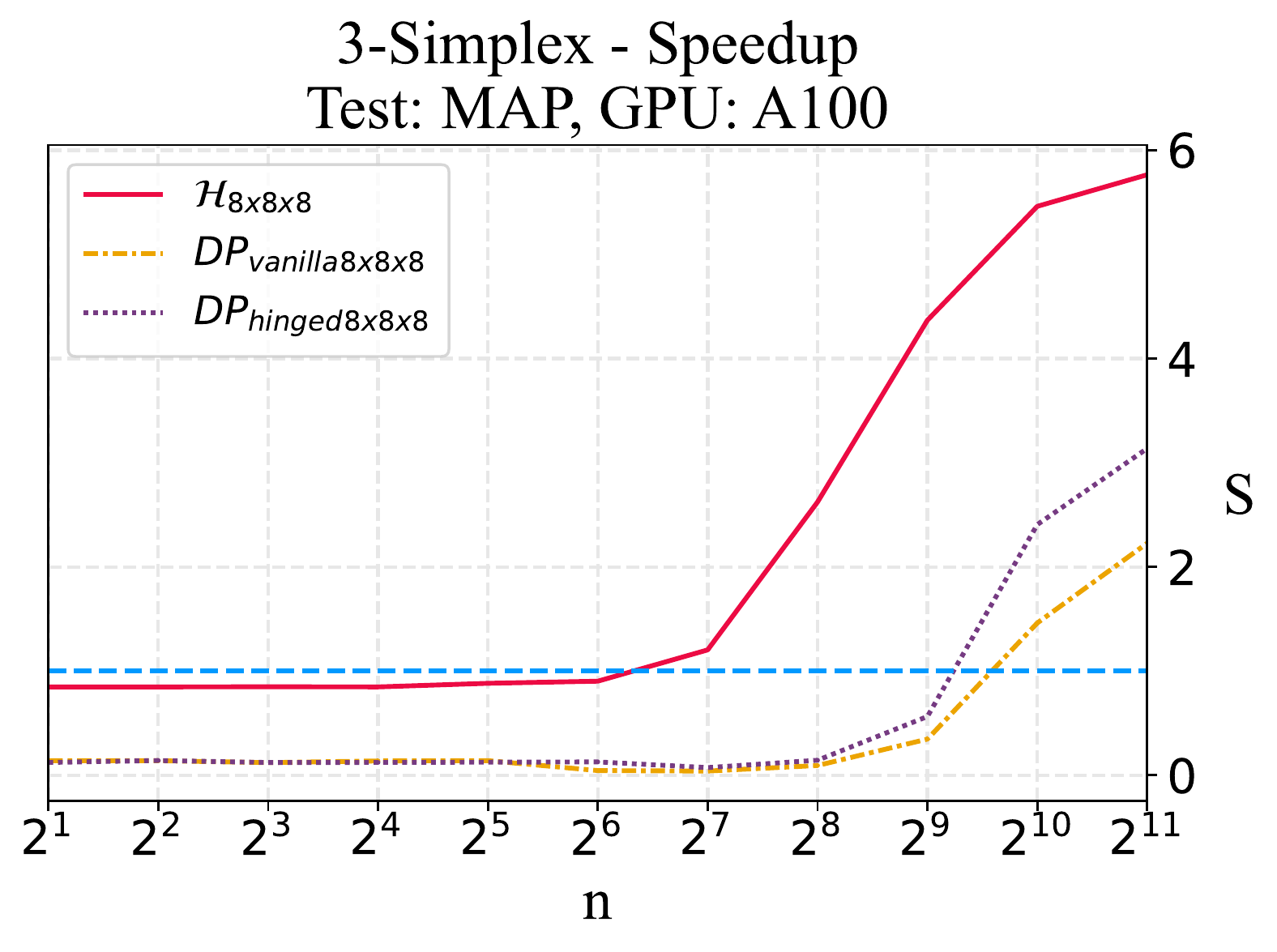}
\includegraphics[scale=0.37]{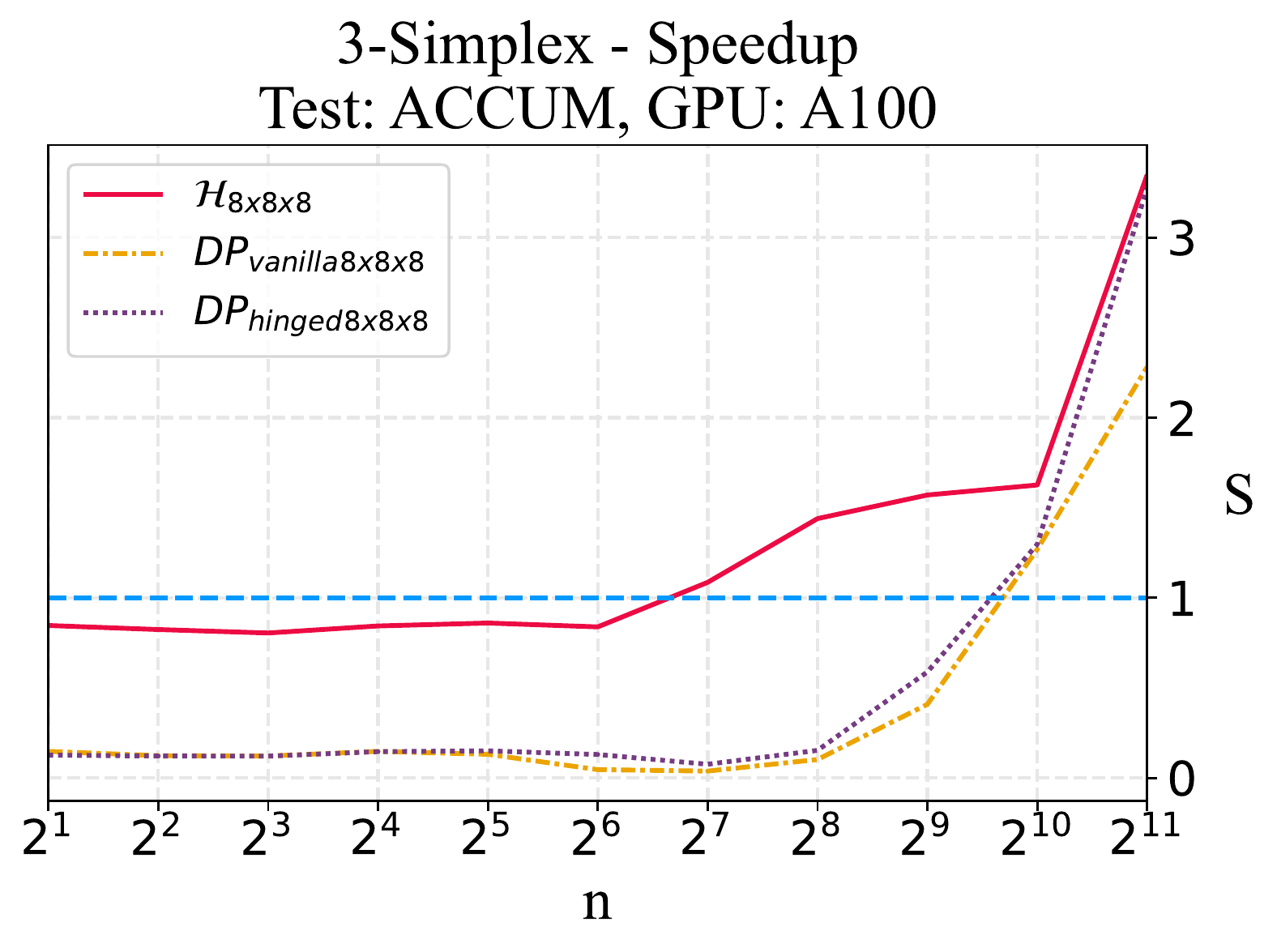}
\includegraphics[scale=0.37]{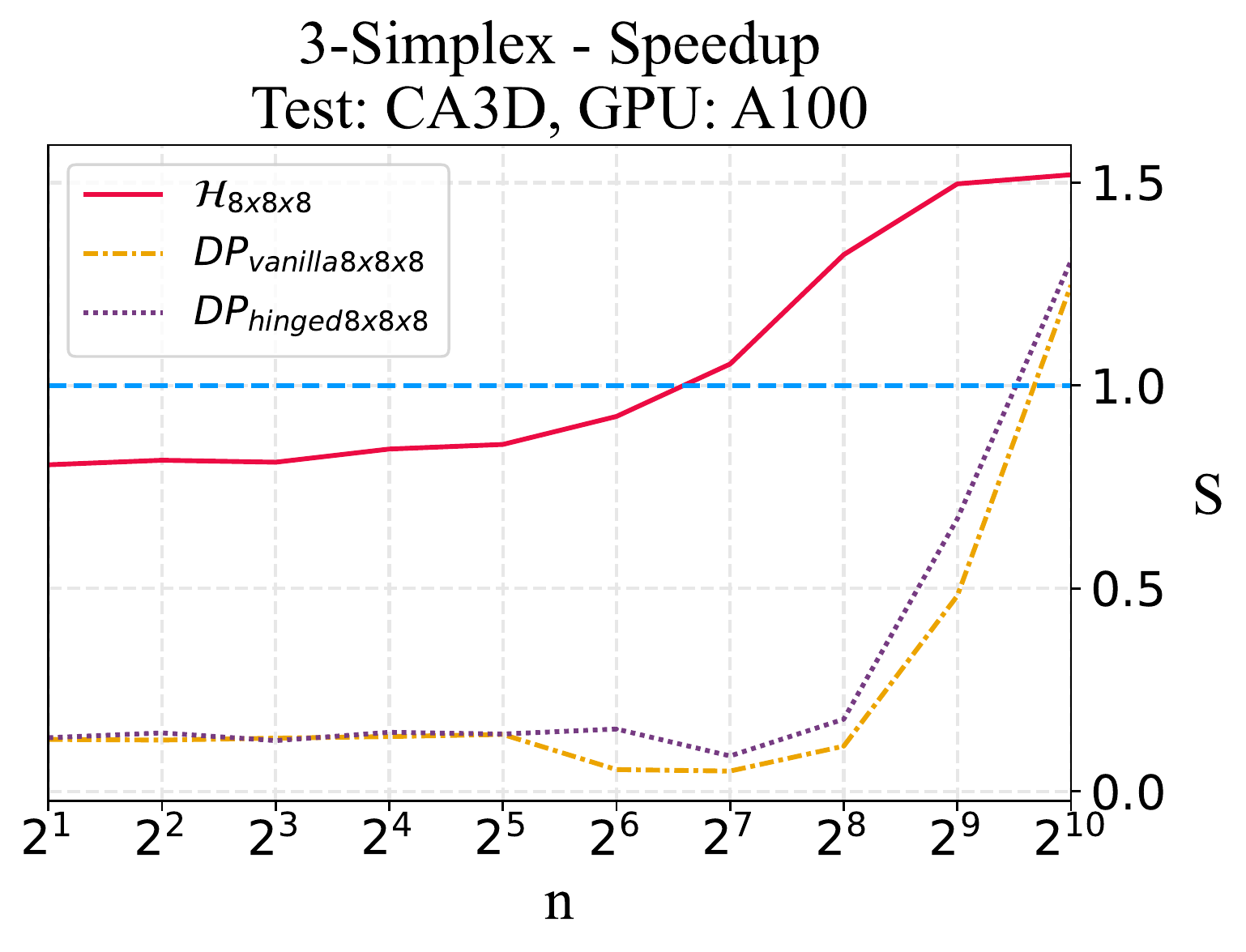}\\
\includegraphics[scale=0.37]{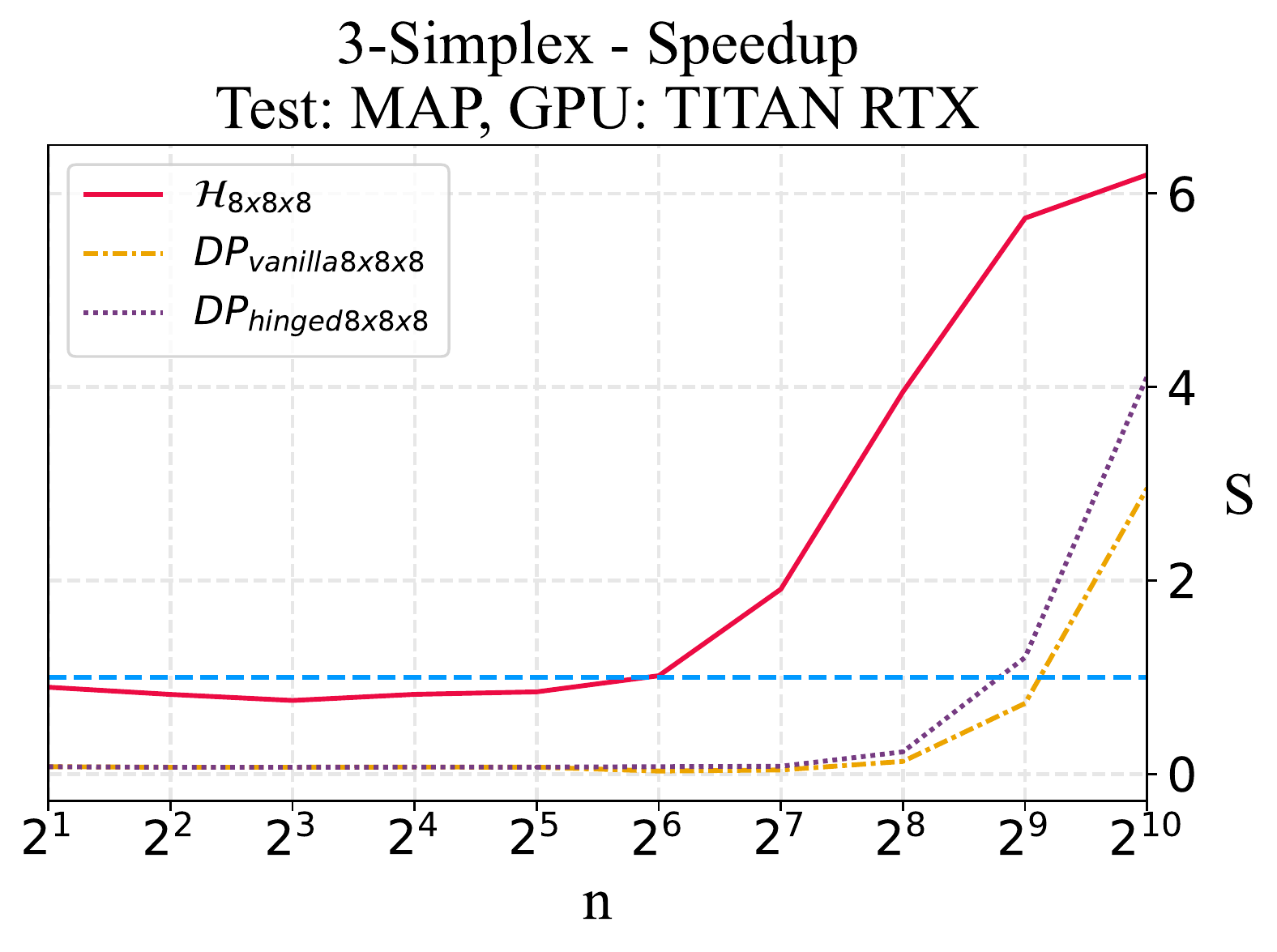}
\includegraphics[scale=0.37]{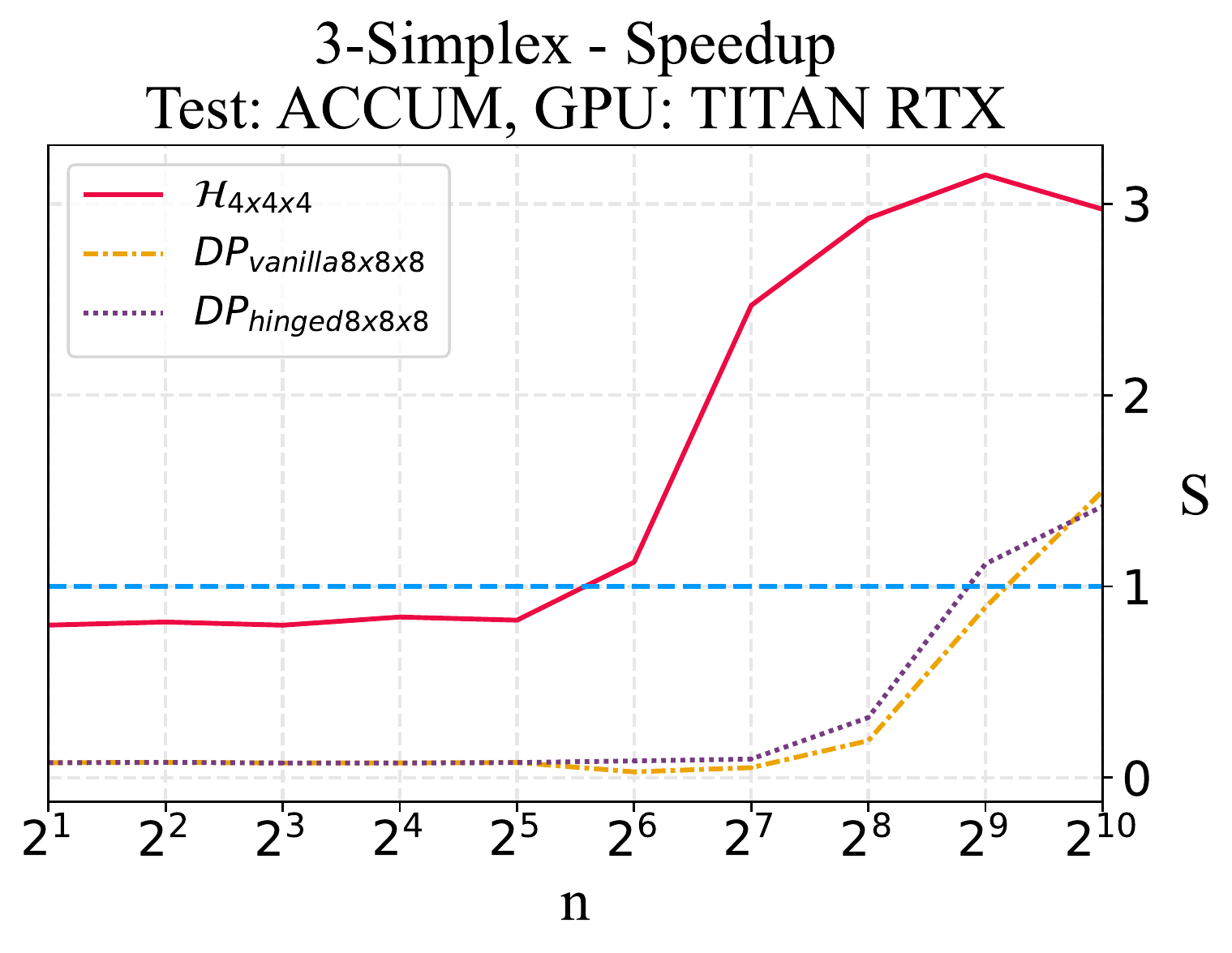}
\includegraphics[scale=0.37]{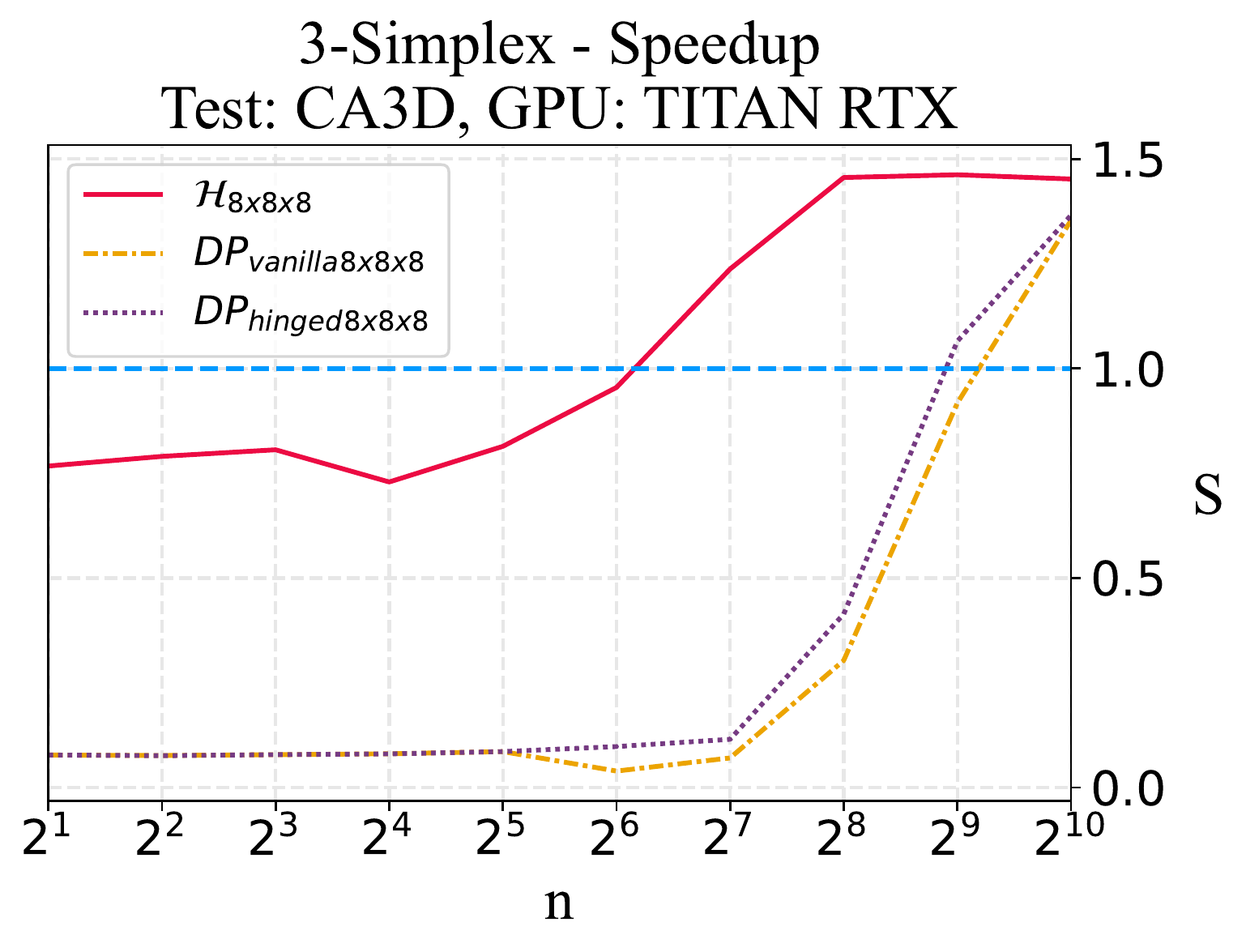}\\
\includegraphics[scale=0.37]{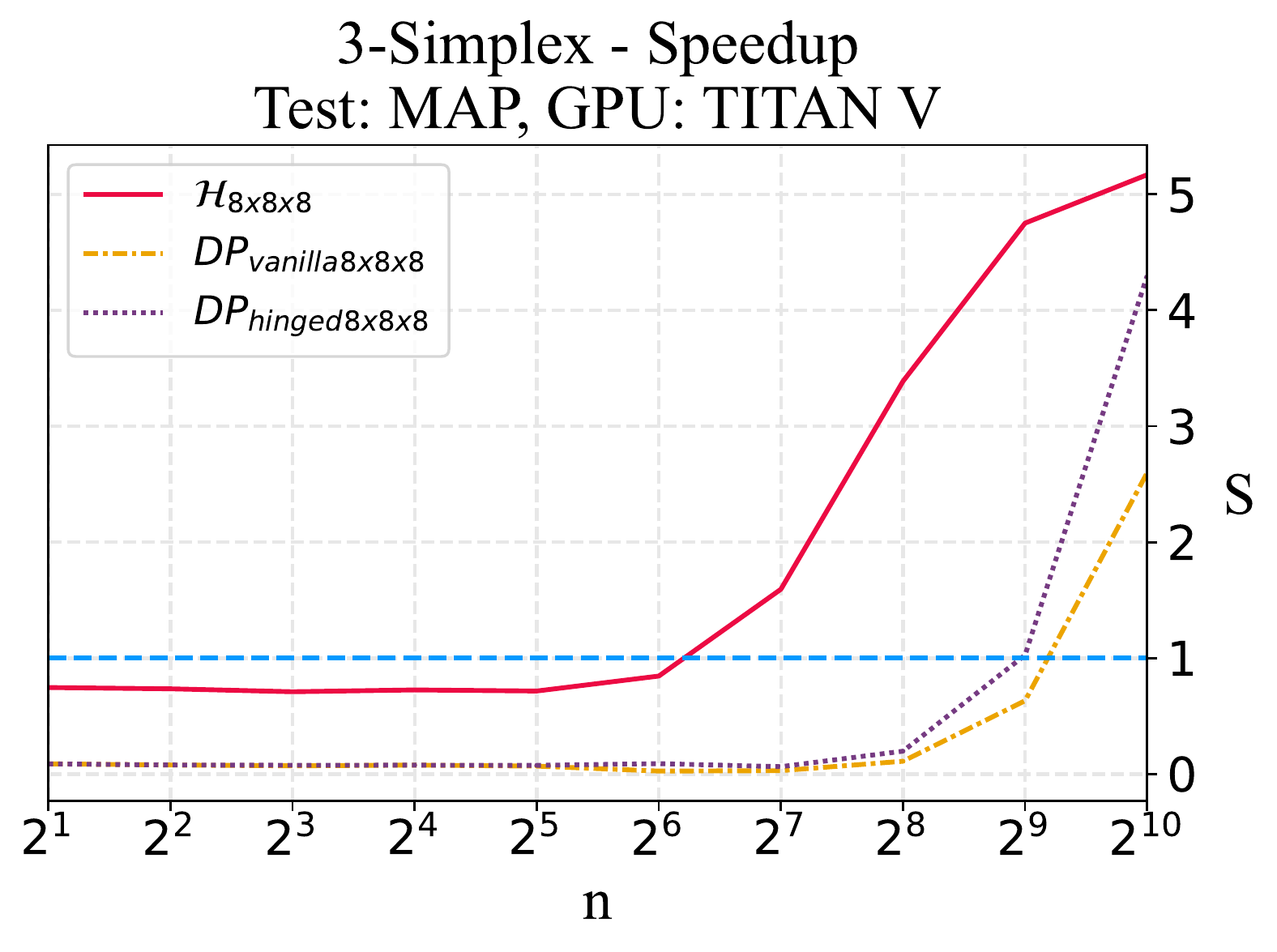}
\includegraphics[scale=0.37]{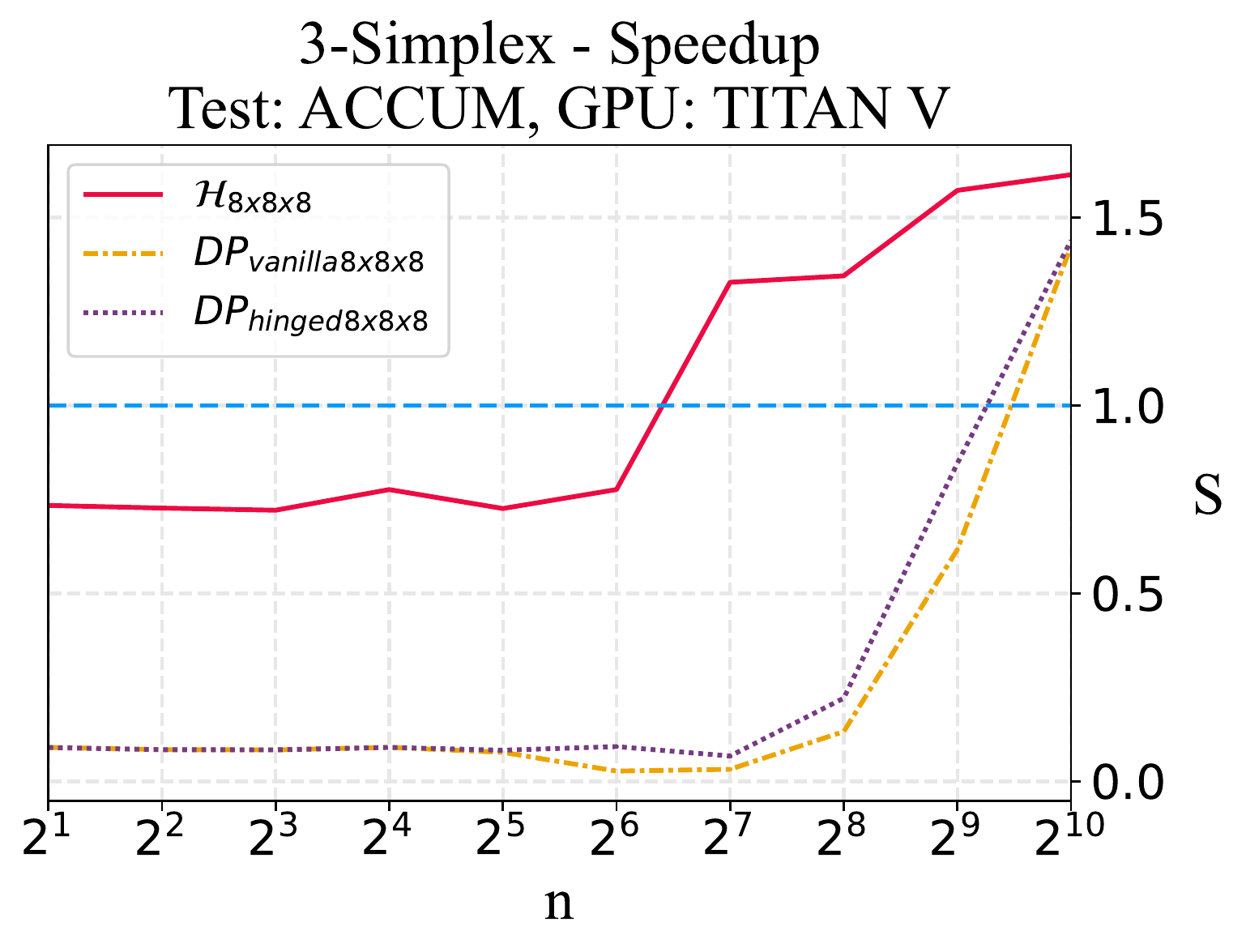}
\includegraphics[scale=0.37]{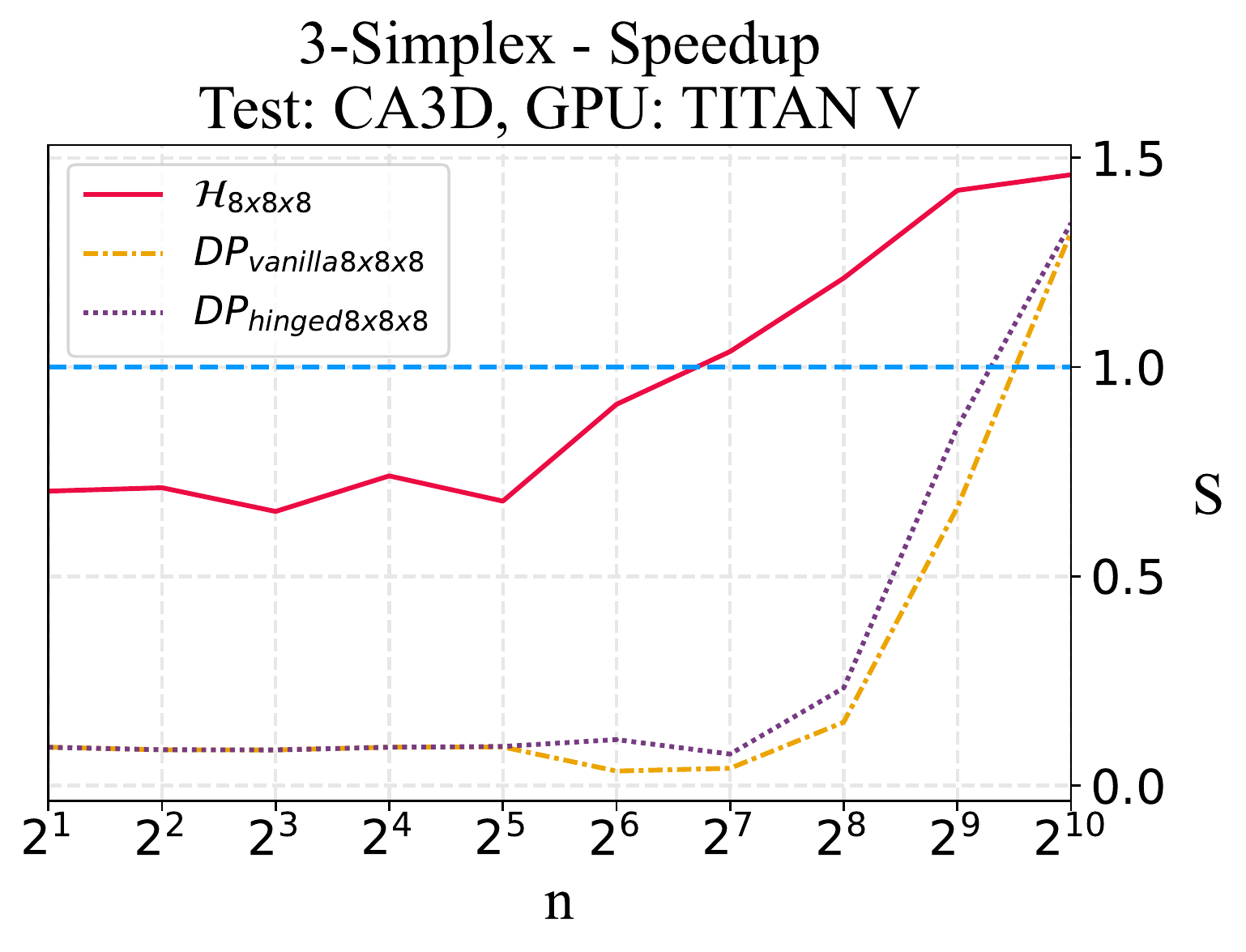}\\
\caption{Speedup for different $3$-simplex tests, running on different GPUs.}
\label{fig:3s-speedup}
\end{figure*}
In the MAP test (first column) $\mathcal{H}$ approaches to $\sim 6\times$ of speedup, matching the theoretical expected value. For the DP variants, the $DP_{hinge}$ version exhibits a higher speedup than the vanilla variant. On the ACCUM test (second column) $\mathcal{H}$ is the fastest map again, but this time the speedup reaches up to $\sim 3\times$. The DP variants show similar performance for the TITAN GPUs, and favors the hinge variant in the A100 GPU. Lastly, on the CA3D test (third column), $\mathcal{H}$ reaches a maximum speedup near $\sim 1.5\times$ followed by the DP approaches.

Figure \ref{fig:3s-power} presents the power consumption time series for all tests running on the A100 GPU.
\begin{figure*}[ht!]
\centering
\includegraphics[scale=0.35]{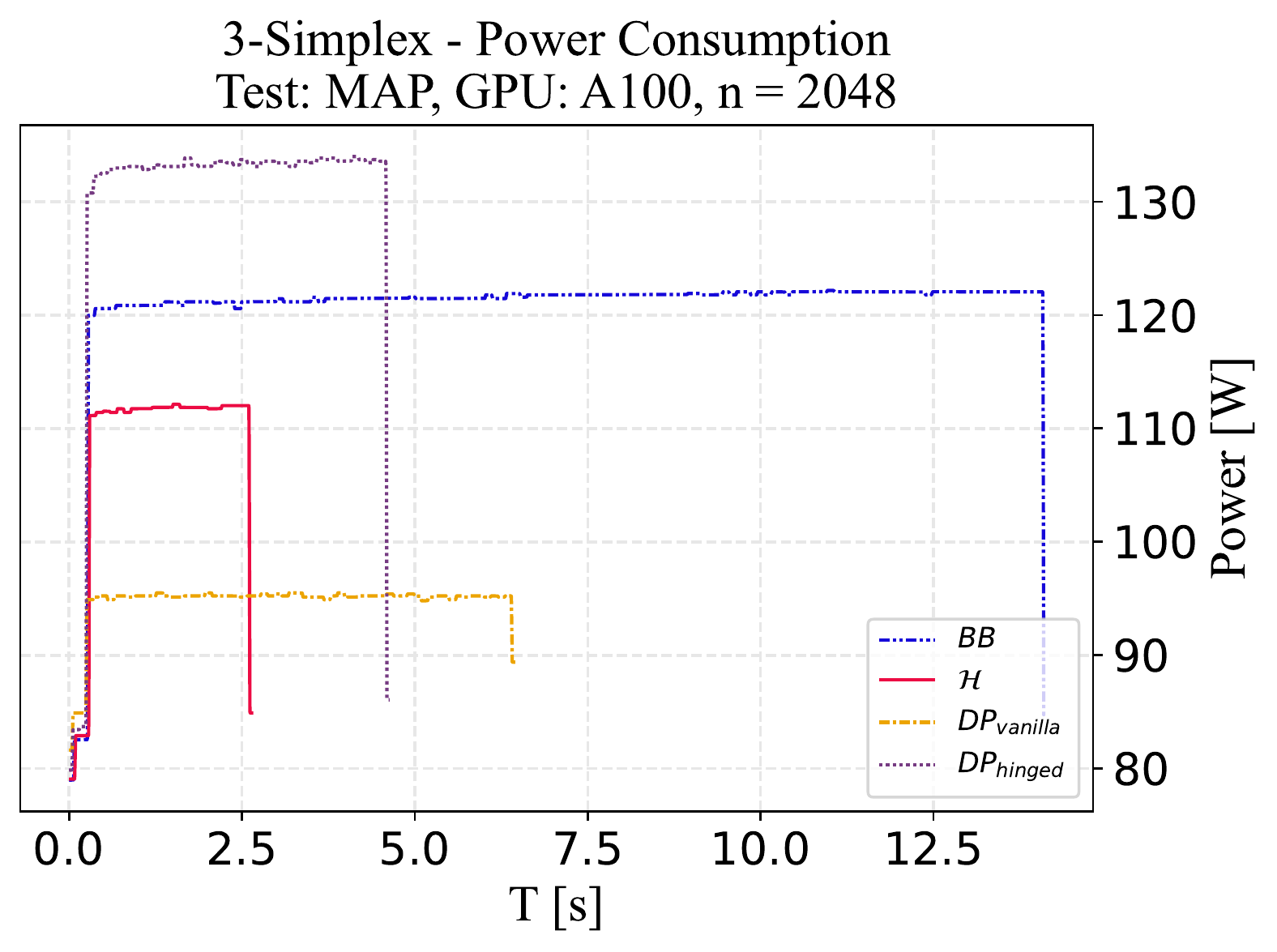}
\includegraphics[scale=0.35]{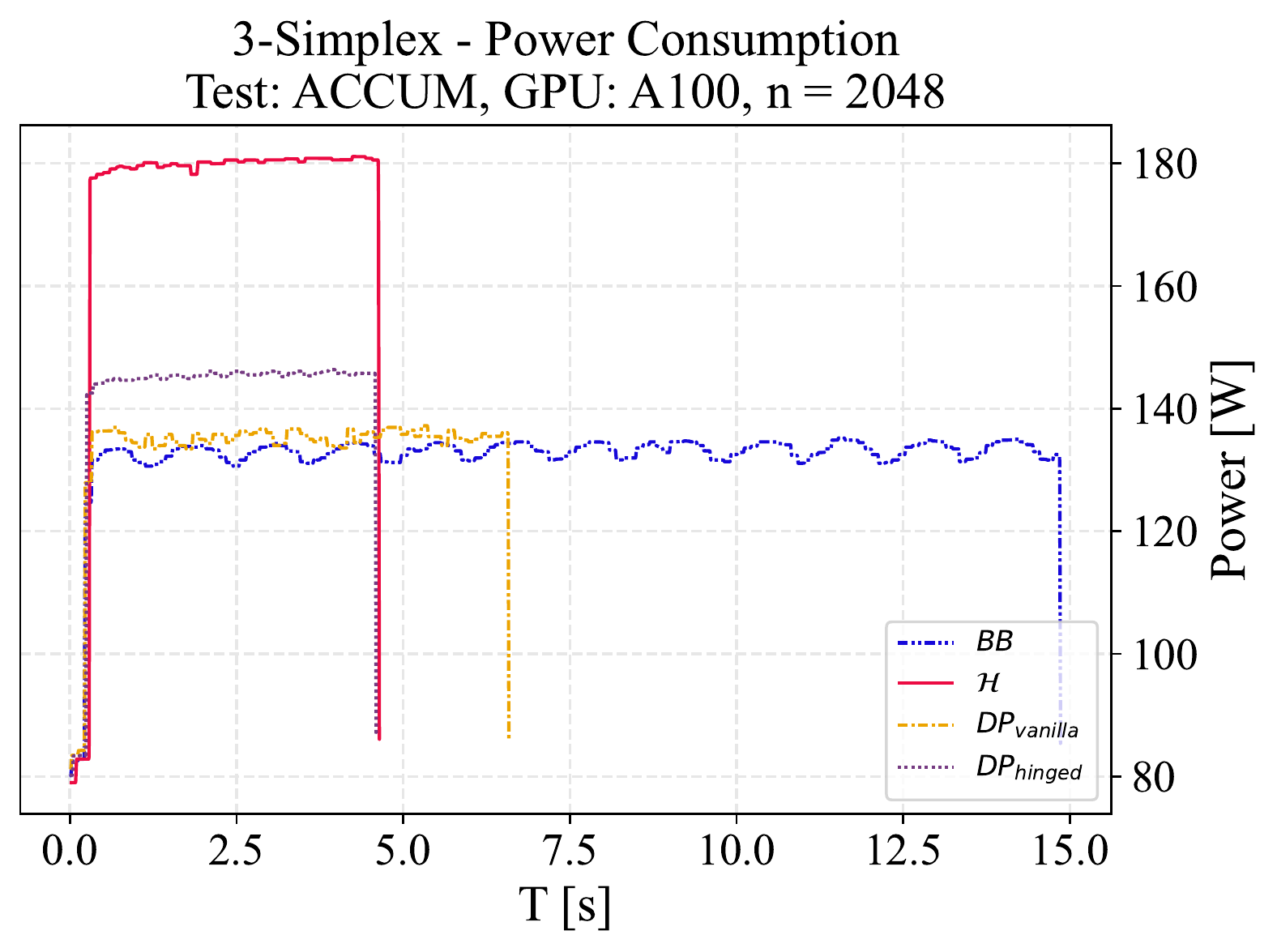}
\includegraphics[scale=0.35]{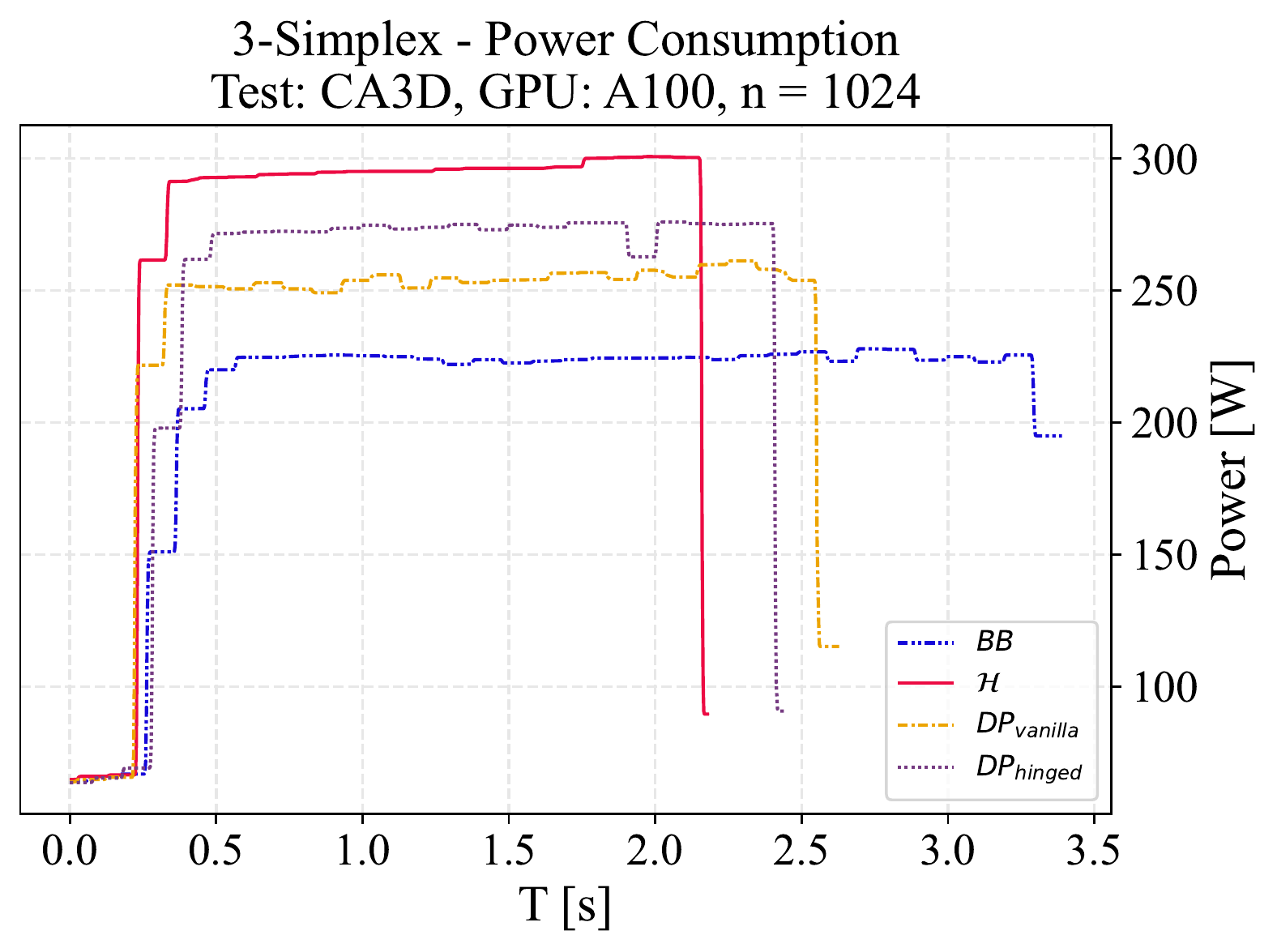}
\caption{Power time series for $3$-simplex tests running on the A100 GPU.}
\label{fig:3s-power}
\end{figure*}
It is worth noticing that the power consumption behavior for the MAP test puts $DP$ and $\mathcal{H}$ as the approaches with less power peak. On the ACCUM and CA3D test $\mathcal{H}$ increases its peak power consumption surpassing the other approaches, but is the first one in making the power consumption curve to go down. Figure \ref{fig:3s-energy-eff} presents the energy efficiency of all approaches running on the A100 GPU. This plots show more easily which map is the most energy efficient in terms of elements per second per Watt. Overall $\mathcal{H}$ stands as the most energy efficient one along with $DP_{hinge}$. 
\begin{figure}[ht!]
\centering
\includegraphics[scale=0.53]{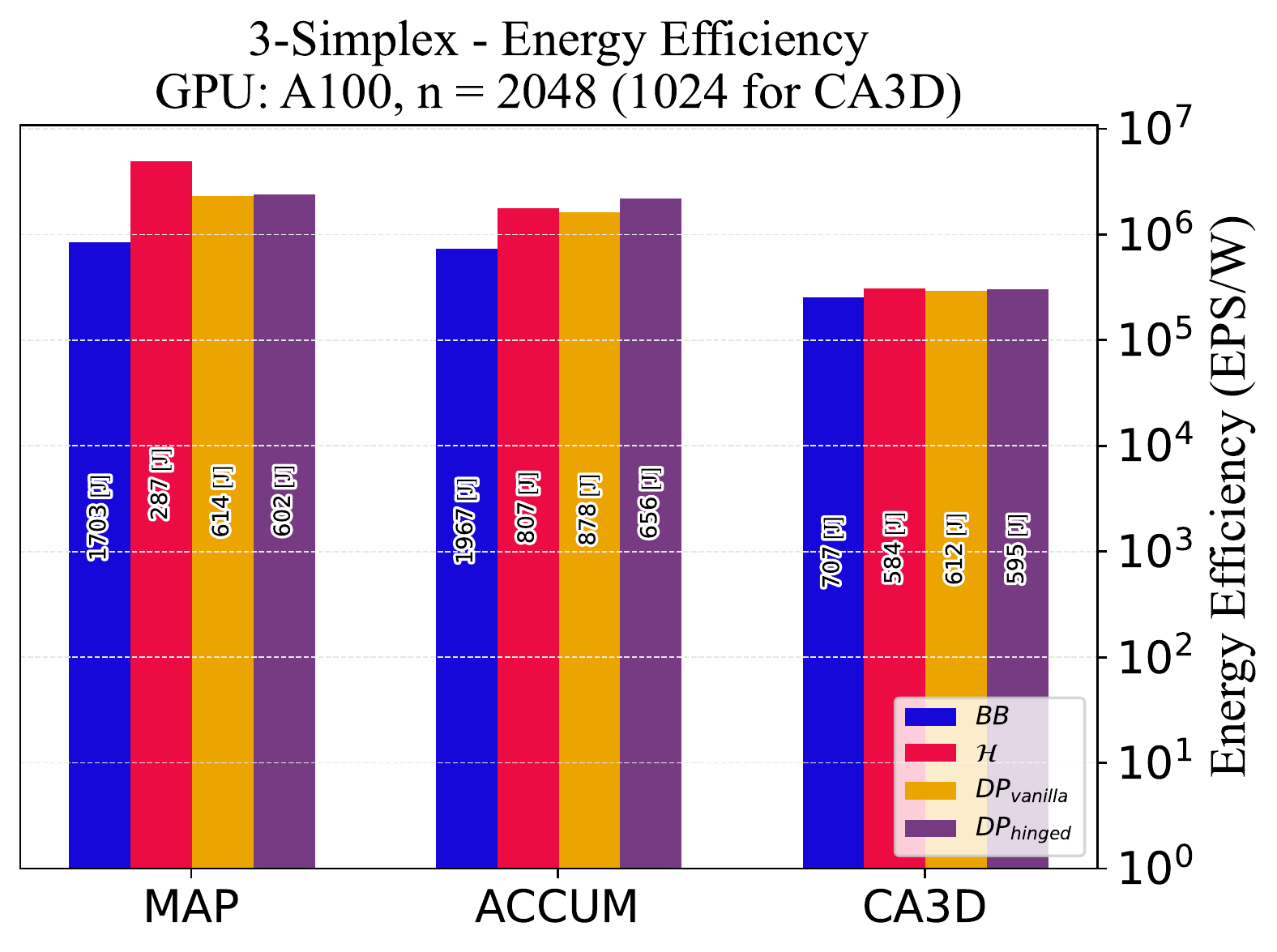}
\caption{Energy efficiency (higher is better), as well as total energy used.}
\label{fig:3s-energy-eff}
\end{figure}

\section{Considerations for extending $\mathcal{H}$ to higher dimensions}
\label{sec:considerationshigher-dimensions}
The formulation of $\mathcal{H}$ for the $2$-simplex and $3$-simplex followed specific designs for their corresponding dimensions. Although the map takes constant time in
both cases, it is worth noticing that for the $3$-simplex it was necessary
to introduce $12\%$ of extra parallel volume in order to fit the set $S_n^m$ on 
$\Pi_n$ as a single-pass map, unless multiple parallel spaces are assumed in the model. 
When generalizing the approach to $m$-simplices, it is important to first verify if $V(S_n^m) = V(\Delta_n^m) + o(n^m)$ satisfies or not, and if not then analyze how much extra space is introduced.
In the general case, the volume of a set $S_n^m$ of self-similar regular orthotopes is 
\begin{alignat}{2}
    V(S_n^m) &= (rn)^m + \beta V(S_{rn}^m) = (rn)^m \sum_{i=0}^{
        \log_{1/r}(n) - 1} (\beta r^m)^i
\end{alignat}
where $r$ is the scaling factor and $\beta$ the arity or multiplicity of the
recursion. Applying the geometric series, the expression becomes
\begin{alignat}{2}
    V(S_n^m) &= (nr)^m \Bigg( \frac{(\beta r^m)^{\log_{1/r}(n)} - 1}{\beta r^m - 1}\Bigg)\\
    \label{eq_generic-m}
            &= \frac{n^m - \beta^{\log_{1/r}(n)}}{1/r^m -\beta}.
\end{alignat}
One can verify that setting $r=1/2$ and $\beta = 2$ leads to equations
(\ref{eq_m2}) and (\ref{eq_m3}) for $m=2,3$, respectively. 
\begin{lemma}
\label{lemma_5}
The parallel space efficiency of a self-similar set $S_n^m$ of orthotopes with parameters
$r=1/2$ and $\beta = 2$ is efficient for $m=2,3$ and factorially inefficient in
$m$ for $m \ge 4$.
\end{lemma}
\begin{proof}
For the general case of $m$, the self-similar strategy of using $r=1/2$, $\beta=2$
produces a fraction of extra volume of
\begin{equation}
    \alpha(S,\Delta)_n^m = \lim_{n\to\infty} \frac{\frac{n^m- n}{2^m - 2}}{ {{(n-1) + m -
    1}\choose{m}}} - 1= \frac{m!}{2^m - 2} - 1 = \mathcal{O}(m!)
\end{equation}
The efficient maps $\mathcal{H}$ proposed for $\Delta_n^2$ and
$\Delta_n^3$ in fact take advantage of the early behavior of the extra space function,
which remains near zero. For $m \ge 4$ the factorial behavior of the extra volume
dominates the cost. 
\end{proof}
For example, when $m=4$ the volume is 
\begin{equation}
    V(S_n^m) = \frac{n^4 - n}{14} > \frac{(n-1)n(n+1)(n+2)}{24}
\end{equation}
and $\alpha(S_n^m,\Delta_n^m)$ approaches to $5/7$ of $\Delta_n^m$. For $m=5$ and
$m=7$ it approaches to $3\times$ and $39\times$ the volume of $\Delta_n^m$, and so on.
\begin{theorem}
\label{theorem_3}
Finding an efficient self-similar set $S_n^m$ for $\Delta_n^m$ is an optimization problem.
\end{theorem}
\begin{proof}
By Lemma \ref{lemma_5}, the $r=1/2, \beta=2$ approach is only efficient for
$m=2,m=3$. For all the rest of the cases, finding an efficient set requires to 
find parameters $r, \beta$ such that it minimizes the target function $|V(S_n^m) -
V(\Delta_n^m)|$.  The constraints are $\beta > 1$, $1/r > \beta$
as well as $1/r, \beta \in \mathbb{Z}_+$. 
\end{proof}    
Additional considerations include that $\beta^{\log_{1/r}(n)}$ should not grow
too fast as it has an impact on what is the initial $n_0$ for which $V(S_n^m)
\ge V(\Delta_n^m)$ holds.  
As an example, a value of $r=1/(m^{-1/m})$ produces the required $m!$, making
$\beta$ a free parameter to be adjusted with $\beta \ge 2$.  Choosing
$\beta = 2$ provides a set $S_n^m$ that covers $\Delta_n^m$ from a certain $n
\ge n_0$, where $n_0$ is a value that increases with $m$. It is possible to
bring $n_0$ closer to the origin by increasing $\beta$, however the extra volume
increases as well, presenting a trade-off. What is interesting is that from $n
\ge n_0$, the parallel space is practically $m!$ times more efficient than a
bounding box approach, presenting a great potential for transforming this space
improvement into a performance one.

\section{Insights for Leveraging Tensor/Ray-Tracing Cores}
\label{sec:tensor-rtx-cores}
Since 2018, GPU architectures have began introducing two new types of ASICs (Application Specific Integrated Circuits) onto the chip, i) Tensor cores (TC) and ii) Ray-tracing (RT) cores to further speedup Deep Learning and Real-time graphics applications, respectively. Leveraging the multiple TCs and RT cores that sit on a modern GPU depends on the application and how well it can reformulate its computation in terms of TC and RT operations.  

\subsection{Leveraging Tensor Core Units}
The tensor core exposes the matrix multiply accumulate (MMA) operation: $D_{m\times n} = A_{\text{m}\times \text{k}}\times B_{\text{k} \times \text{n}} + C_{\text{m} \times \text{n}}$ with different $\textbf{m} \times \textbf{n} \times \textbf{k}$ formats\footnote{Internally, the GPU schedules how the operation is decomposed into even smaller matrices, such as $4\times 4$, to be processed by each TC. 
} to choose from depending on the numerical precision; $\bm{16}\times \bm{16} \times \bm{16}$ or $\bm{8} \times \bm{32} \times \bm{16}$ for FP16, $\bm{16} \times \bm{16} \times \bm{8}$ for TF32, or $\bm{8} \times \bm{8} \times \bm{4}$ for FP64, among others. Each MMA operation requires a dedicated warp of threads (32 threads), therefore parallelism granularity is at the warp level. Adapting $\mathcal{H}$ to TCs requires the reformulation of its computation, first by proposing a matrix layout to place its variables/constants and second by producing the desired computation from the fused multiply add (FMA) operations that take place in each row-column product of the MMA. In the case of $\mathcal{H}$, given the map is at block-level, an efficient adaptation should try to obtain several block coordinates for each MMA executed. 
An efficient adaptation of $\mathcal{H}$ to TCs would employ a format where the product $A \times B$ acts on non-square matrices, as the FMA operations are relatively short in length. The TF32 format $\bm{16} \times \bm{16} \times \bm{8}$ can be an attractive solution if more precision is required, or the FP16 format $\bm{8} \times \bm{32} \times \bm{16}$ to increase the number of simultaneous block coordinates, although at less precision. In any case, the layout is to place the constants of Eq. (\ref{eq_map2d}) or (\ref{eq_trapezoids}) as rows in $A$, and the $x,y$ block-related inputs for $\mathcal{H}$ as pairs of columns in $B$. Eq. (\ref{eq:tensor-core-mma}) expresses the mentioned TC layout acting for two blocks (distinguished by the super-indices in $B$) in the case of $2$-simplices, under a format of $\textbf{m}\times \textbf{n} \times \textbf{k} = 4 \times 4 \times 2$ just for simplicity. 
\begin{align} 
\label{eq:tensor-core-mma}
\scriptsize
D &= 
\begin{bmatrix}
    a_1 	& a_2   \\
    a_3 	& a_4   \\
    a_1 	& a_2   \\
    a_3 	& a_4   
\end{bmatrix}
\times
\begin{bmatrix}
    b^1_{x_1} & b^1_{y_1} & b^2_{x_1} & b^2_{y_1} \\
    b^1_{x_2} & b^1_{y_2} & b^2_{x_2} & b^2_{y_2} 
\end{bmatrix}
+
\begin{bmatrix}
    c_{x_1}  & c_{y_1}   & 0         & 0       \\
    c_{x_2}  & c_{y_2}   & 0         & 0       \\
    0 	     & 0         & c_{x_3}   & c_{y_3} \\
    0	     & 0         & c_{x_4}   & c_{y_4}
\end{bmatrix}
\end{align}

The are four constants $a_1, a_2, a_3, a_4$, and the product $A \times B$ generates a block coordinate for each block twice. This is to allow threads to add their local coordinate on top of it, through the addition of matrix $C$. The resulting matrix is $D$ which contains the definitive positions per-thread to access the simplex domain. It is worth mentioning that this multi-block approach can be further improved by the use of the recent thread block cluster construct, which allows to group several CUDA blocks and synchronize them to share data, such as matrix fragments.

\subsection{Leveraging Ray Tracing Cores}
Ray Tracing (RT) cores are hardware-implemented pipelines that accelerate the process of querying the collision of a ray with a set of geometric primitives, such as triangles in 3D space. Internally, RT cores employ a hardware-implemented bounding volume hierarchy (BVH) to handle its irregular memory access patterns on the BVH more efficiently. Although the standard application of RT cores is to accelerate the generation of photo-realistic computer graphics from a 3D scene, it is still possible to use the ray-triangle intersection query as a fast tool to solve other kinds of problems, such as thread mapping. 

In the case of thread mapping to a simplex domain, RT cores can contribute with a fast hardware accelerated framework that works similar to Dynamic Parallelism, to explore a non-boxed data domain during execution. Moreover, given that RT cores work in a fully dynamic context, it opens the possibility of mapping threads to fully dynamic changing domains, dropping the static restriction assumed in $\mathcal{H}$ and making it a general purpose mapping. A possible idea to leverage RT cores for simplex domains is to load the simplex as a voxel type geometry in 3D space, and launch waves of rays at different levels of thickness. This would generate a dynamic discovery process where the first waves of rays do a coarse-grained exploration, useful for discarding the 3D space with no simplex data. The process would gradually refine rays until they become voxel grained and access the data elements. The only limitation is that the hardware BVH was designed for 3D space.

\section{Discussion and Conclusions}
\label{sec_discussion-conclusions}
This work presented a new map, $\mathcal{H}$, for mapping blocks of threads onto $m$-simplex domains. The map was entirely formulated for the $2$-simplex and $3$-simplex cases, obtaining theoretical upper bounds on speedup and thread usage efficiency over a standard bounding-box approach. It was also experimentally evaluated and compared to other state of the art approaches, showing that $\mathcal{H}$ is competitive with respect to the other alternative approaches, reaching up to $2\times$ of speedup just for the mapping stage, and between $1.1\times \sim 1.3\times$ of speedup under different tests. For $3$-simplices it was overall the fastest map, approaching the theoretical $6\times$ for the mapping stage, and between $1.5\times \sim 3.2\times$ under different tests. When comparing $\mathcal{H}$ with the state of the art approaches, we note that it is the only map that does exhibits a favorable performance and numerical precision in both 2D and 3D cases. Moreover, $\mathcal{H}$ is potentially extendable to higher dimensions offering up to $m!\times$ more efficiency. This generalization to $m$-simplices presents a challenge though, as obtaining an optimal set of orthotopes with minimal extra volume becomes an optimization problem where the scaling and replication parameters, $r, \beta$ respectively, produce a trade-off between extra space and the starting value $n_0$ from which the mapping can take place. Knowing what parameters are the optimal for building a self-similar set of orthotopes for any $m$-simplex, as well as find general methods for packing $S_n^m$ into  super-orthotope are indeed interesting questions that deserve further study.

In terms of power consumption, $\mathcal{H}$ showed that although it achieves high power consumption during execution, it compensates by having a shorter duration, leading to reach among the highest values of elements per second per Watt ($EPS/W$). This makes $\mathcal{H}$ energy efficient, as for the same amount of energy it can process more elements than other approaches.

This work also analyzed the possibility of leveraging both Tensor Cores (TC) and Ray Tracing (RT) Cores that are otherwise unused during a regular CUDA application. A preliminary analysis shows that it is indeed feasible to implement mappings using these special purpose cores. In fact, TC can serve to speedup $\mathcal{H}$ directly, while RT core units show that it can offer a dynamic mapping approach more general than $\mathcal{H}$. Future research can work on implementing such approaches considering the modern $2022+$ GPU architectures, and measure what is the empirical acceleration that these special purpose cores can provide to the thread mapping process. 

As a conclusion, this work has shown that $\mathcal{H}$ is an scalable and energy efficient solution that can be of great interest researchers working on applications that need to execute long simulations on simplex structures, such as PDEs or Cellular Automata.

\section*{Acknowledgment}
This project was supported by FONDECYT grant \#$1221357$, the Temporal research group and the Patag\'on Supercomputer from Universidad Austral de Chile.

\bibliographystyle{elsarticle-num}
\bibliography{main} 

\begin{thebibliography}{10}
\expandafter\ifx\csname url\endcsname\relax
  \def\url#1{\texttt{#1}}\fi
\expandafter\ifx\csname urlprefix\endcsname\relax\def\urlprefix{URL }\fi
\expandafter\ifx\csname href\endcsname\relax
  \def\href#1#2{#2} \def\path#1{#1}\fi

\bibitem{4490127}
J.~Owens, M.~Houston, D.~Luebke, S.~Green, J.~Stone, J.~Phillips, Gpu
  computing, Proceedings of the IEEE 96~(5) (2008) 879--899.

\bibitem{Nickolls}
J.~Nickolls, W.~J. Dally, The gpu computing era, IEEE Micro 30~(2) (2010)
  56--69.

\bibitem{navhitmat2014}
C.~A. Navarro, N.~Hitschfeld-Kahler, L.~Mateu, A survey on parallel computing
  and its applications in data-parallel problems using {GPU} architectures,
  Commun. Comput. Phys. 15 (2014) 285--329.

\bibitem{nvidia_cuda_guide}
Nvidia-Corporation, Nvidia CUDA C Programming Guide (2016).

\bibitem{opencl08}
{Khronos OpenCL Working Group}, The OpenCL Specification, version 1.0.29 (8
  December 2008).

\bibitem{10.1145/2517327.2442523}
B.~Wu, Z.~Zhao, E.~Z. Zhang, Y.~Jiang, X.~Shen,
  \href{https://doi.org/10.1145/2517327.2442523}{Complexity analysis and
  algorithm design for reorganizing data to minimize non-coalesced memory
  accesses on gpu}, SIGPLAN Not. 48~(8) (2013) 57–68.
\newblock \href {https://doi.org/10.1145/2517327.2442523}
  {\path{doi:10.1145/2517327.2442523}}.
\newline\urlprefix\url{https://doi.org/10.1145/2517327.2442523}

\bibitem{4408272}
W.~W. Fung, I.~Sham, G.~Yuan, T.~M. Aamodt, Dynamic warp formation and
  scheduling for efficient gpu control flow, in: 40th Annual IEEE/ACM
  International Symposium on Microarchitecture (MICRO 2007), 2007, pp.
  407--420.
\newblock \href {https://doi.org/10.1109/MICRO.2007.30}
  {\path{doi:10.1109/MICRO.2007.30}}.

\bibitem{8947754}
L.~V. Lucas~Vespa, Unraveling the divergence of gpu threads, in: 2018
  International Conference on Computational Science and Computational
  Intelligence (CSCI), 2018, pp. 1398--1403.
\newblock \href {https://doi.org/10.1109/CSCI46756.2018.00270}
  {\path{doi:10.1109/CSCI46756.2018.00270}}.

\bibitem{7445236}
X.~Mei, X.~Chu, Dissecting gpu memory hierarchy through microbenchmarking, IEEE
  Transactions on Parallel and Distributed Systems 28~(1) (2017) 72--86.
\newblock \href {https://doi.org/10.1109/TPDS.2016.2549523}
  {\path{doi:10.1109/TPDS.2016.2549523}}.

\bibitem{10.1145/2304576.2304619}
J.~Holewinski, L.-N. Pouchet, P.~Sadayappan,
  \href{https://doi.org/10.1145/2304576.2304619}{High-performance code
  generation for stencil computations on gpu architectures}, in: Proceedings of
  the 26th ACM International Conference on Supercomputing, ICS '12, Association
  for Computing Machinery, New York, NY, USA, 2012, p. 311–320.
\newblock \href {https://doi.org/10.1145/2304576.2304619}
  {\path{doi:10.1145/2304576.2304619}}.
\newline\urlprefix\url{https://doi.org/10.1145/2304576.2304619}

\bibitem{harris2007optimizing}
M.~Harris, et~al., Optimizing parallel reduction in cuda, Nvidia developer
  technology 2~(4) (2007) 70.

\bibitem{navarro2020gpu}
C.~A. Navarro, R.~Carrasco, R.~J. Barrientos, J.~A. Riquelme, R.~Vega, Gpu
  tensor cores for fast arithmetic reductions, IEEE Transactions on Parallel
  and Distributed Systems 32~(1) (2020) 72--84.

\bibitem{hermann2010multi}
E.~Hermann, B.~Raffin, F.~Faure, T.~Gautier, J.~Allard, Multi-gpu and multi-cpu
  parallelization for interactive physics simulations, in: European Conference
  on Parallel Processing, Springer, 2010, pp. 235--246.

\bibitem{yadan2013multi}
O.~Yadan, K.~Adams, Y.~Taigman, M.~Ranzato, Multi-gpu training of convnets,
  arXiv preprint arXiv:1312.5853 (2013).

\bibitem{stuart2011multi}
J.~A. Stuart, J.~D. Owens, Multi-gpu mapreduce on gpu clusters, in: 2011 IEEE
  International Parallel \& Distributed Processing Symposium, IEEE, 2011, pp.
  1068--1079.

\bibitem{sorna2018optimizing}
A.~Sorna, X.~Cheng, E.~D'azevedo, K.~Won, S.~Tomov, Optimizing the fast fourier
  transform using mixed precision on tensor core hardware, in: 2018 IEEE 25th
  International Conference on High Performance Computing Workshops (HiPCW),
  IEEE, 2018, pp. 3--7.

\bibitem{quezada2022squeeze}
F.~A. Quezada, C.~A. Navarro, N.~Hitschfeld, B.~Bustos, Squeeze: Efficient
  compact fractals for tensor core gpus, Future Generation Computer Systems 135
  (2022) 10--19.

\bibitem{navarro2020efficient}
C.~A. Navarro, F.~A. Quezada, N.~Hitschfeld, R.~Vega, B.~Bustos, Efficient gpu
  thread mapping on embedded 2d fractals, Future Generation Computer Systems
  113 (2020) 158--169.

\bibitem{zhu2022rtnn}
Y.~Zhu, Rtnn: accelerating neighbor search using hardware ray tracing, in:
  Proceedings of the 27th ACM SIGPLAN Symposium on Principles and Practice of
  Parallel Programming, 2022, pp. 76--89.

\bibitem{zellmann2020accelerating}
S.~Zellmann, M.~Weier, I.~Wald, Accelerating force-directed graph drawing with
  rt cores, in: 2020 IEEE Visualization Conference (VIS), IEEE, 2020, pp.
  96--100.

\bibitem{morrical2020accelerating}
N.~Morrical, I.~Wald, W.~Usher, V.~Pascucci, Accelerating unstructured mesh
  point location with rt cores, IEEE Transactions on Visualization and Computer
  Graphics (2020).

\bibitem{DBLP:conf/hpcc/NavarroH14}
C.~A. Navarro, N.~Hitschfeld,
  \href{http://dx.doi.org/10.1109/HPCC.2014.64}{{GPU} maps for the space of
  computation in triangular domain problems}, in: 2014 {IEEE} International
  Conference on High Performance Computing and Communications, 6th {IEEE}
  International Symposium on Cyberspace Safety and Security, 11th {IEEE}
  International Conference on Embedded Software and Systems, {HPCC/CSS/ICESS}
  2014, Paris, France, August 20-22, 2014, 2014, pp. 375--382.
\newblock \href {https://doi.org/10.1109/HPCC.2014.64}
  {\path{doi:10.1109/HPCC.2014.64}}.
\newline\urlprefix\url{http://dx.doi.org/10.1109/HPCC.2014.64}

\bibitem{CLEI-2016-navarro}
C.~A. Navarro, B.~Bustos, N.~Hitschfeld, Potential benefits of a block-space
  {GPU} approach for discrete tetrahedral domains, in: CLEI-2016, XLII
  Conferencia Latinoamericana de Inform{\'a}tica, Valparaiso, Chile, October
  10-14, 2016, 2016.

\bibitem{8392762}
C.~A. Navarro, M.~Vernier, B.~Bustos, N.~Hitschfeld, Competitiveness of a
  non-linear block-space gpu thread map for simplex domains, IEEE Transactions
  on Parallel and Distributed Systems 29~(12) (2018) 2728--2741.
\newblock \href {https://doi.org/10.1109/TPDS.2018.2849705}
  {\path{doi:10.1109/TPDS.2018.2849705}}.

\bibitem{9188081}
C.~A. Navarro, B.~Bustos, N.~Hitschfeld, Analysis of a self-similar gpu thread
  map for data-parallel m-simplex domains, in: 2019 International Conference on
  High Performance Computing \& Simulation (HPCS), 2019, pp. 1002--1010.
\newblock \href {https://doi.org/10.1109/HPCS48598.2019.9188081}
  {\path{doi:10.1109/HPCS48598.2019.9188081}}.

\bibitem{5695222}
D.~Man, K.~Uda, H.~Ueyama, Y.~Ito, K.~Nakano, Implementations of parallel
  computation of euclidean distance map in multicore processors and gpus, in:
  Networking and Computing (ICNC), 2010 First International Conference on,
  2010, pp. 120--127.
\newblock \href {https://doi.org/10.1109/IC-NC.2010.55}
  {\path{doi:10.1109/IC-NC.2010.55}}.

\bibitem{Li:2010:CME:1955604.1956601}
Q.~Li, V.~Kecman, R.~Salman, A chunking method for euclidean distance matrix
  calculation on large dataset using multi-gpu, in: Proceedings of the 2010
  Ninth International Conference on Machine Learning and Applications, ICMLA
  '10, IEEE Computer Society, Washington, DC, USA, 2010, pp. 208--213.

\bibitem{Man:2011:GIC:2117688.2118809}
D.~Man, K.~Uda, Y.~Ito, K.~Nakano, A gpu implementation of computing euclidean
  distance map with efficient memory access, in: Proceedings of the 2011 Second
  International Conference on Networking and Computing, ICNC '11, IEEE Computer
  Society, Washington, DC, USA, 2011, pp. 68--76.

\bibitem{AvrilGA12}
Q.~Avril, V.~Gouranton, B.~Arnaldi, Fast collision culling in large-scale
  environments using gpu mapping function, in: EGPGV, 2012, pp. 71--80.

\bibitem{kepner2011graph}
J.~Kepner, J.~Gilbert,
  \href{http://books.google.com.hk/books?id=BnezR\_6PnxMC}{Graph Algorithms in
  the Language of Linear Algebra}, Software, Environments, Tools, Society for
  Industrial and Applied Mathematics, 2011.
\newline\urlprefix\url{http://books.google.com.hk/books?id=BnezR\_6PnxMC}

\bibitem{ConwaysLife}
M.~Gardner, {The fantastic combinations of John Conway's new solitaire game
  ``life''}, Scientific American 223 (1970) 120--123.

\bibitem{Ries:2009:TMI:1654059.1654069}
F.~Ries, T.~De~Marco, M.~Zivieri, R.~Guerrieri, Triangular matrix inversion on
  graphics processing unit, in: Proceedings of the Conference on High
  Performance Computing Networking, Storage and Analysis, SC '09, ACM, New
  York, NY, USA, 2009, pp. 9:1--9:10.

\bibitem{DBLP:journals/corr/abs-1108-5815}
R.~Yokota, L.~A. Barba, Fast n-body simulations on {{GPU}s}, CoRR abs/1108.5815
  (2011).

\bibitem{Bedorf:2012:SOG:2133856.2134140}
J.~B{\'e}dorf, E.~Gaburov, S.~Portegies~Zwart,
  \href{http://dx.doi.org/10.1016/j.jcp.2011.12.024}{A sparse octree
  gravitational n-body code that runs entirely on the {GPU} processor}, J.
  Comput. Phys. 231~(7) (2012) 2825--2839.
\newblock \href {https://doi.org/10.1016/j.jcp.2011.12.024}
  {\path{doi:10.1016/j.jcp.2011.12.024}}.
\newline\urlprefix\url{http://dx.doi.org/10.1016/j.jcp.2011.12.024}

\bibitem{Ivanov:2007:NPT:1231091.1231100}
L.~Ivanov, \href{http://dl.acm.org/citation.cfm?id=1231091.1231100}{The n-body
  problem throughout the computer science curriculum}, J. Comput. Sci. Coll.
  22~(6) (2007) 43--52.
\newline\urlprefix\url{http://dl.acm.org/citation.cfm?id=1231091.1231100}

\bibitem{1054599}
J.~Costello, On the number of points in regular discrete simplex (corresp.),
  IEEE Transactions on Information Theory 17~(2) (1971) 211--212.
\newblock \href {https://doi.org/10.1109/TIT.1971.1054599}
  {\path{doi:10.1109/TIT.1971.1054599}}.

\bibitem{Jung2008}
J.~H. Jung, D.~P. O’Leary, Exploiting structure of symmetric or triangular
  matrices on a gpu, Tech. rep., University of Maryland (2008).

\bibitem{patagon-uach}
A.~U. of~Chile, \href{https://patagon.uach.cl}{Patagón supercomputer} (2021).
\newline\urlprefix\url{https://patagon.uach.cl}

\end{thebibliography}
\end{document}